\documentclass[sigconf, review=false, authorversion=false, balance=false, nonacm]{acmart}
\usepackage{popets}
\usepackage{custompackages}
\usepackage{ulem}

\setcopyright{none}
\copyrightyear{Preprint. Under review}

\acmYear{}
\acmVolume{}
\acmNumber{d}
\acmDOI{}
\acmISBN{}
\acmConference{}
\begin{document}

\title{Private and Collaborative Kaplan-Meier Estimators}


\author{Shadi Rahimian}
\affiliation{%
 \institution{CISPA Helmholtz Center for Information Security}
 \country{}
 }
\email{shadi.rahimian@cispa.de}

\author{Raouf Kerkouche}
\affiliation{%
 \institution{CISPA Helmholtz Center for Information Security}
 \country{}
  }
\email{raouf.kerkouche@cispa.de}
  
\author{Ina Kurth}
\affiliation{%
 \institution{DKFZ German Cancer Research Center}
 \country{}
  }
\email{ina.kurth@dkfz-heidelberg.de}

\author{Mario Fritz}
\affiliation{%
 \institution{CISPA Helmholtz Center for Information Security}
 \country{}
  }
\email{fritz@cispa.de}


\begin{abstract}
Kaplan-Meier estimators are essential tools in survival analysis, capturing the survival behavior of a cohort. Their accuracy improves with large, diverse datasets, encouraging data holders to collaborate for more precise estimations. However, these datasets often contain sensitive individual information, necessitating stringent data protection measures that preclude naive data sharing.

In this work, we introduce two novel differentially private methods that offer flexibility in applying differential privacy to various functions of the data. Additionally, we propose a synthetic dataset generation technique that enables easy and rapid conversion between different data representations. Utilizing these methods, we propose various paths that allow a joint estimation of the Kaplan-Meier curves with strict privacy guarantees. Our contribution includes a taxonomy of methods for this task and an extensive experimental exploration and evaluation based on this structure. We demonstrate that our approach can construct a joint, global Kaplan-Meier estimator that adheres to strict privacy standards ($\varepsilon = 1$) while exhibiting no statistically significant deviation from the nonprivate centralized estimator.
\end{abstract}


\maketitle

\section{Introduction}
\label{sec:introcution}
Survival analysis, or time-to-event analysis~\cite{kleinbaum2012survival}, encompasses methods that provide statistics on the survival characteristics of a population. It is employed in various fields such as medical research to predict patient mortality~\cite{goldhirsch1989costs,brenner2002long}, in finance to model customer defaults or service unsubscriptions~\citep{baesens2005neural, dirick2017time, lu2002predicting}, and generally to study population behavior over time for specific events. A widely used statistic in this field is the Kaplan-Meier estimator~\cite{goel2010understanding}, a nonparametric tool that approximates the survival probability of a population directly from survival data. This estimator is particularly valuable in the medical field as it allows straightforward analysis of the effects of treatments or markers on survival outcomes without complex formulations.

The effectiveness and precision of Kaplan-Meier estimators in modeling the true survival probability can be maximized when they are constructed from large datasets. However, individual data centers and research institutes often lack access to such comprehensive datasets, prompting the need for collaborative efforts among multiple data holders to develop more robust estimators. Despite the benefits of collaboration, data protection regulations like the General Data Protection Regulation (GDPR)~\cite{GDPR} impose strict limitations~\cite{SUP,GDPR}, preventing the straightforward sharing of raw data and ensuring the privacy of data contributors, such as patients.

Attempts to address security concerns with the use of a collaborative Kaplan-Meier estimator have so far only utilized secure aggregation and encryption schemes~\cite{froelicher2021truly, vogelsang2020secure, von2021privacy, xi2022review}. However, these approaches are computationally intensive and time-consuming and do not scale effectively with an increasing number of collaborators. Additionally, they fail to provide privacy guarantees for the released global model. An adversary with access to aggregated statistics could still perform attacks such as reidentification~\cite{el2011systematic, rocher2019estimating} or inference~\cite{backes2016membership, homer2008resolving}, compromising the privacy of data contributors.

A theoretically robust solution to ensure individual privacy within a dataset is differential privacy~\citep{Dwork2014book} (DP). DP applies controlled randomization to data, functions of data, or aggregated statistics to safeguard individual privacy while allowing the extraction of summary statistics. An interesting property of differential privacy is that an adversary, with access to any auxiliary information, is not able to infer further information from any function applied on the output of a DP mechanism. This is known as the post-processing property of differential privacy.

Studies that attempt to combine differential privacy and Kaplan-Meier estimators have been very limited so far~\citep{gondara2020differentially} and focus on protecting the count numbers and do not propose any solution when we do \textit{not} have access to the count numbers at each time. The previous method also does not offer protection for the specific \textit{times} of events. To date, no work has suggested a differentially private framework that can facilitate collaboration for this problem. 

In this work, we first introduce two new differentially private methods that can be applied on different functions of survival data, and based on our methods, we suggest various paths that collaborators can take to privately build a joint global Kaplan-Meier estimator. Our paths offer great flexibility for the preferred shared information in a collaborative system and are easy to apply and fast to compute. In summary:  
\begin{itemize}
    \item We present the first approach to the problem of privacy-preserving joint survival estimation over an aggregate of clients and provide a systematic analysis of how to achieve this global model.
    \item We propose two differentially private methods that local clients can utilize for the privacy of their data. We then suggest multiple paths that these clients can propagate their private information through, in order to construct a final joint KM estimator. We are able to achieve good utility compared to the centralized setting at a high privacy level ($\varepsilon=1$).
    \item We are able to release client-level differentially private surrogate datasets which enable us to construct an accurate, private and joint Kaplan-Meier estimator by pooling these private datasets. 
\end{itemize}

\section{Background}\label{sec:background}


\subsection{Survival Analysis and Kaplan-Meier (KM) Estimators}
\label{sec:km-definition}
Survival analysis is the collection of statistical methods that aim to model and predict the time duration to an event of interest for a set of data points. As an example, the events of interest in medical survival analysis might be the time it takes for a patient to die from an initial point when the patient enters a study, the time to metastasis, time to relapse, etc.

The survival analysis dataset is in the form of $D=\{t^i, e^i\}_{i=1}^{N}$ where $t^i$ is the event time for the data point $i$ and $e^i\in \{0, 1\}$ is the corresponding type of event for the data point $i$. We say that a data point is \textit{right censored} when $e^i=0$. This happens when an individual is excluded from the study, usually for reasons other than the event of interest, or when the event of interest does not occur until the maximum study time $\tm$. When $e^i=1$, the event of interest occurs for the data point $i$.


Let $t^* \geq 0$ be a random variable. The survival function at time $t$ is defined as the probability of the event of interest $t^*$ happening after $t$:
\begin{eqnarray}
S(t) = \pr(t^* > t)
\label{eq:survival}
\end{eqnarray}

The survival function is a smooth non-increasing curve over time and its value is bound to $[0,1]$. However, in practice, to model the survival function based on a finite number of data points, we need to estimate the value of $S(t)$. The Kaplan-Meier (KM) estimator~\cite{kaplan1958nonparametric} $\s$ is a nonparametric step function of data, used to estimate the survival function:
\begin{eqnarray}
\s(t) = \prod_{t'\leq t}\frac{r_{t'} - d_{t'}}{r_{t'}}
\label{eq:kmestimator}
\end{eqnarray}
where $r_t$ is the number of datapoints at risk or more commonly known as the \textit{risk set} (those that have not experienced any type of event) at time $t$ and $d_t$ is the number of data points experiencing the event of interest (i.e. $e=1$) at time $t$. Here, we assume that there are $T$ distinct times of event in the whole dataset: $t'\in\{0=t_0, t_1, t_2, ...,t_{T-1}=\tm\}$. In practice, we can discretize the times of events with an equidistant grid with bin size $b$ and calculate the $\s(t)$ based on the  number of events that occur within each grid. 

\subsection{Event Probability Mass Function}
\label{sec:probs-definition}
As is evident from Equations~\ref{eq:survival} and~\ref{eq:kmestimator}, the Kaplan-Meier function estimates the probability of the event up to a certain point in time. This is a restrictive view, and instead we might want to measure the probability at each specific time or time interval. 
For this reason, we also consider the closely related concept of probability mass function:
\begin{eqnarray}
y(t) = \pr(t^* = t|x)
\end{eqnarray}
which represents the probability that a new data point $x$ will experience the event at time $t$. Throughout this paper, we use probability mass function and \textit{probability function} or simply \textit{probability}, interchangeably. The true probability for discretized times of events, $t\in\{0, t_1, ..., \tm\}$, and with the assumption that no event happens at time $t_0=0$, can be approximated by an estimator~\cite{lee2018deephit, kvamme2019time} $\yh$:
\begin{eqnarray}
\label{eq:s2y}
\yh(t_j) = \begin{cases}
0 & t_j=0\\
\s(t_{j-1}) - \s(t_j) & t_1\leq t_j \leq \tm \\
1 - \sum_{t'\leq  \tm} y(t') & t_j=\tm + 1
\end{cases}
\end{eqnarray}
\begin{eqnarray}
\label{eq:y2s}
\s(t_j) = 1 - \sum_{t'\leq t_j}\yh(t')
\end{eqnarray}
The probability mass function estimator $\yh$ shows the overall probability of incident during each time interval and it contains one element more than $\s(t)$. This extra element $\yh(\tm+1)$ is considered to capture the probability that the event will occur beyond the end time of the study~\cite{kvamme2021continuous}. With this extra element the sum of all the elements in the $\yh$ vector should be 1.0 (i.e. $\sum_{t_j=0}^{\tm+1} \yh(t_j) = 1$) as is expected from a probability mass function. As we can see from Equations~\ref{eq:s2y} and~\ref{eq:y2s}, the conversion between the estimator for probability mass function $\yh$, and the Kaplan-Meier estimator $\s$ is straightforward and fast. This gives us the opportunity to convert between the two when we need a different viewpoint on the survival status of the population.
\begin{figure}
    \centering
    \includegraphics[width=0.35\textwidth]{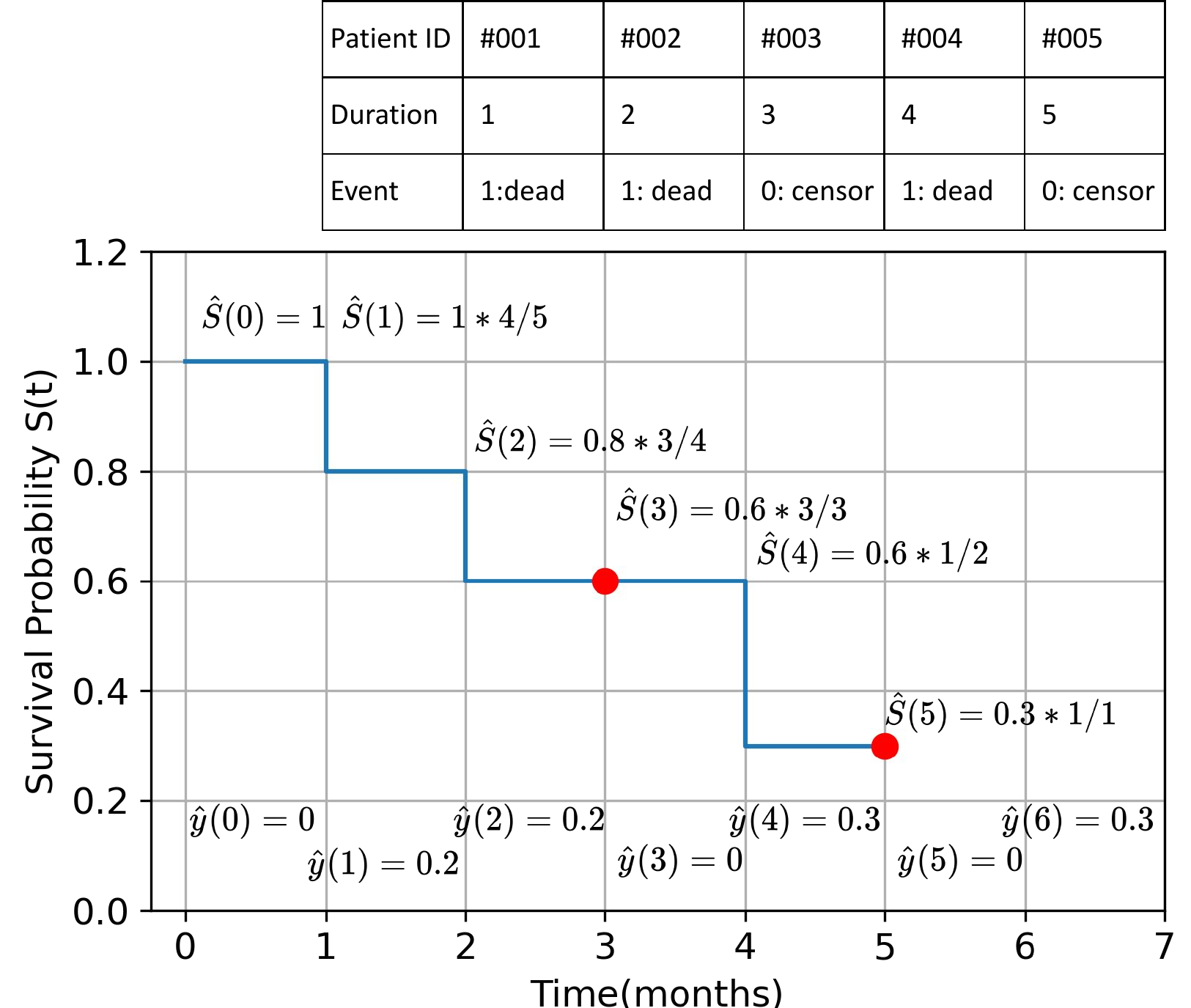}
    \caption{A simple illustrative example of Kaplan-Meier and probability estimators for a dataset of 5 individuals.}
    \label{fig:toy}
\end{figure}
To better demonstrate the relationship between these functions, we provide a toy example in Figure~\ref{fig:toy} for a small dataset of only 5 individuals. Here $t \in \{0, 1, 2, 3, 4, 5\}$ months and 2 individuals are censored at times $t=3$ and $t=5=\tm$ (shown with red circles on the survival plot).  At each step $t$ we have $\s(t) = \s(t-1)\times \frac{r_t-d_t}{r_t}$. Notice that the value of the survival function does not change when a point is censored; however, the individuals that are censored are not considered in the risk set of the next time step. 
Indeed, the probability mass function only models the probability of the events of interest happening at each time step. 

\subsection{Differential Privacy (DP)}
\label{sec:dp-definition}

The goal of this paper is to build a global survival model for the collection of data from multiple data owners. However, this collaboration now carries the risk of privacy leakage through these shared data or shared data statistics. Differential privacy~\cite{Dwork2014book} is the standard method for mathematically restricting the probability of information leakage from the data or a function of the data. DP adds calibrated randomness to the data or its functions such that the general statistics inferred from the dataset remain accurate but the sensitive information of individual data points is suppressed. There will always be a privacy-utility trade-off: the more randomness is added, the more private the algorithm and less accurate the statistics learned from the collection of the data, and vice versa. Thus, we always strive to find an operating point which offers the best privacy-utility trade-off.  

\begin{definition}[$\varepsilon$-Differential Privacy~\citep{Dwork2014book}]
\label{def:dp}
A randomized algorithm $\mathcal{A}$ is $\varepsilon$-differentially private, if for any two neighboring datasets $D$ and $D'$ and for any $S\subseteq \mathit{Range}(\mathcal{A})$ we have:
$$
\Pr(\mathcal{A}(D) \in S) \leq e^{\varepsilon} \Pr(\mathcal{A}(D') \in S)
$$
\end{definition}
Intuitively, this guarantees that an adversary, provided with the output of $\mathcal{A}$, can draw almost the same conclusion (up to ${\varepsilon}$) about whether dataset $D$ or $D'$ was used. That is, for any record owner, a privacy breach is unlikely to be due to its participation in the dataset.

In \textbf{\textit{bounded}} DP, $D'$ can be obtained from $D$ by changing the value of exactly one data point. And in \textbf{\textit{unbounded}} DP, $D'$ can be obtained from $D$ by adding or removing one data point. 

Throughout this paper, we choose to work only with \textit{bounded} differential privacy. Note that in the bounded setting, the neighboring datasets $D$ and $D'$ have the same fixed size.  

\subsubsection{Laplace Mechanism}
\label{sec:laplace}
As explained in Definition~\ref{def:dp}, a randomization process is necessary for differential privacy. There are many mechanisms that can be applied to data or functions of data to make these differentially private. Here we focus on the so-called \textit{Laplace mechanism}. But we first need to define the \textit{global sensitivity} of a function~\citep{Dwork2014book}:
\begin{definition}[Global $L_p$-sensitivity] 
\label{def:global-sens}
For any function $f:\mathcal{D} \rightarrow \mathbb{R}^ k$, and all possible neighboring datasets $D$ and $D'$, the $L_p$-sensitivity of $f$ is
$\Delta_p f = \max_{D, D'} || f(D)-f(D') ||_p$, where $||\cdot||_p$ denotes the $L_p$-norm.\vspace*{-0.15cm}
\end{definition}
The Laplace Mechanism~\citep{Dwork2014book} 
consists of adding Laplace noise to the true output of a function, in order to make the function differentially private. 
\begin{definition}[Laplace Mechanism~\citep{Dwork2014book}]
\label{def:laplace-mechanism}
For any function $f:\mathcal{D} \rightarrow \mathbb{R}^k$, the randomized function $\mathcal{A}$:
$$
\mathcal{A}(f(.), \varepsilon) = f + (\mathcal{L}_1, ...,\mathcal{L}_k)
$$
is differentially private. Where $\mathcal{L}_i$ are drawn independently and randomly from a Laplace distribution centered on 0, with $ 
\mathsf{pdf}_{\mathcal{L}(0, l)}(x) = \frac{1}{2l} \exp\left(-\frac{\lvert x \rvert}{l}\right) 
$ where the scale parameter $l$ depends on the sensitivity through $l=\frac{\Delta_1 f}{\varepsilon}$.
\end{definition}
Differential privacy is immune to postprocessing (closure under postprocessing); this means that an adversary cannot compute a function of the output of a differentially private mechanism $\mathcal{A}$ and make it less differentially private.
\begin{theorem}[Post-Processing Property~\citep{Dwork2014book}] 
\label{theorem:postprocess}
Let $\mathcal{A}$ be an $\varepsilon-$DP privacy mechanism which assigns a value $\mathit{Range}(\mathcal{A})$ to a dataset $D$. Let $\mathcal{B}$ be an arbitrary randomized mapping that takes as input $O \in \mathit{Range}(\mathcal{A})$ and returns $O' \in \mathit{Range}(\mathcal{B})$. Then $\mathcal{B} \circ \mathcal{A} $ is also $\varepsilon$-differentially private.

\end{theorem}

\section{Differentially Private Survival Statistics Estimators}
\label{sec:method}
In this section, we explain the methods that can be used to make a survival dataset or its functions differentially private. Based on these methods, we can later build paths that enable a collaborative learning system. We will first review a previously suggested method -which we call \dpmo- that perturbs the matrix of count numbers at each unique time of event. We then introduce our two novel methods, \dps and \dpy, which are more flexible and can be applied directly to the Kaplan-Meier function and the probability estimator, respectively. We lay out all 3 methods with the assumption of a bounded DP. 

\subsection{\dpm}
\label{sec:dpmdef}
The only available baseline method applies DP directly to the number counts $d_\ti$, number of censored points $c_\ti$ and the risk set $r_\ti$ in the dataset. In this method suggested by~\cite{gondara2020differentially}, first a partial matrix $M=[r_0, d_0, c_0, d_{t_1}, c_{t_1}, ...., d_{\tm}, c_{\tm}]$ of the number of events and the number of censoring at each unique incident time $d_\ti$ and the total number of individuals at the initial time $r_0$ is constructed, then Laplace noise with sensitivity of 2 is directly added to these numbers. The authors use this sensitivity, since adding or removing one data point from the data set will at most change the count number by 2 (simultaneously in the values $r_0$ and $d_\ti$ or $c_\ti$). To make this method comparable with our upcoming suggested DP methods, we formulate it in the bounded DP setting, where the size of neighboring datasets remains the same. This means that we keep the total number of points $r_0$ unperturbed and fixed. The sensitivity still remains 2 as the effect of changing one data point can change the norm $L_1$ by at most 2. So we can construct the DP partial matrix $M'$ as follows: 
\begin{alignat}{3}
\resizebox{0.9\hsize}{!}{$
M' = M + [0, \mathcal{L}_{d_0}, \mathcal{L}_{c_0},..., \mathcal{L}_{d_{\tm}}, \mathcal{L}_{c_{\tm}}] \qquad \text{where}\ \ \mathcal{L}_j \sim \mathcal{L}(0, 2/\varepsilon)
$}
\end{alignat}
After obtaining the noisy $d'_\ti$ and $c'_\ti$ values, the remaining number of at risk is calculated by:
\begin{eqnarray}
r'_{t_j} = r'_{t_{j-1}} - (d'_{t_{j-1}} + c'_{t_{j-1}})\qquad \forall t_j\in\{t_1,...,\tm\}
\label{eq:dpmpp}
\end{eqnarray}
We call this method \textit{differentially private matrix} or \dpmo for short.

A major point of concern when working with \dpmo is that it only perturbs the count numbers at distinct recorded times of events. This means that the algorithm would not perturb the times of incident; hence the times of the events are {\it not} protected by this method, and this still poses privacy concerns for the dataset. To correct this issue, we use a preprocessing step in which we discretize the times of events and work with the count numbers $d_\ti$ or $c_\ti$ accumulated per time bin. We call this improved version of the algorithm \dpm.

\subsection{\dps}
\label{sec:dpsurvdef}
In our first proposed method, we strive to offer more flexibility with respect to the function on which differential privacy is applied. This method, which we call \textit{ differentially private Kaplan-Meier estimator} or \dps for short, directly tweaks the Kaplan-Meier survival estimator to make it private and no access to the dataset is needed. We also address the issue of privacy of times of events (as discussed in Section~\ref{sec:dpmdef}) by sampling the KM estimator at equidistant time intervals. By doing so, the function no longer contains the sensitive distinct times of events, and applying differential privacy on this vector now also protects times.

Consider $\s(t)$ the vector of survival estimates that contains the sampled values of the continuous KM function at equidistant time intervals. The $L_1$ and $L_2$ sensitivities of this function are as follows: 
\begin{theorem}
\label{theorem:ssensitivity}
Let's denote the total number of data points in the dataset by $N$, the total number of censored points by $C$ and the total number of time bins in the equidistant grid over $t=0$ till $t=\tm$ by $T$, then when $C=0$: 
{\small
\begin{alignat}{3}
&&&\Delta_1\s = \frac{T-1}{N},\quad \Delta_2\s = \frac{\sqrt{T-1}}{N}\nonumber
\end{alignat}
}%
\end{theorem}
\begin{proof}
We provide the proofs for the sensitivity $L_1$ and $L_2$ of the Kaplan-Meier estimator $\s(t)$, in Appendix~\ref{sec:appendix-proof-s}. 
\end{proof}

Calculating sensitivities for when censored points are present in the dataset requires more care and is not as straightforward for the $\s$ function. We include our derivations for both cases of $C=0$ and $C\neq0$ in Appendix~\ref{sec:appendix-proof-s} and discuss why the latter is no longer DP. We defer solutions for a general case to future work.  

Now, let us inspect the sensitivity in the absence of censored points. Here, we see that the presence of $T$ in the numerator of $\Delta_1\s$, means that applying a simple Laplace mechanism (Definition~\ref{def:laplace-mechanism}) is not useful. This is because only for the special case of $T\ll N$, the sensitivity is reasonably small such that the DP noise would not destroy the utility. Inspired by the works of~\cite{rastogi2010differentially, kerkouche2021compression}, we take advantage of the equidistant sampling of the KM curve to first transform it to the discrete cosine transform (DCT) space~\cite{makhoul1980fast} (a refresher on DCT is provided in Appendix~\ref{sec:appendix-dct}). We can then add the DP noise to only the first $k$ coefficients of the $DCT(\s(t))$ vector, masking the remaining $T-k$ components with zero. The first coefficients of DCT capture the large-scale structure and the most condensed statistics of the signal, such as the mean value. The fine details contained in the remaining coefficients are protected from DP noise by setting them to zero.  
\begin{theorem}
\label{theorem:dctws}
Denote ${D}^k ={DCT}^k(\s(t))$ the first $k$ coefficients of the discrete cosine transform of $\s(t)$, then $\Delta_1{D}^k \leq \sqrt{k} \Delta_2\s(t)$. 
\end{theorem}
\begin{proof}
DCT is a transformation into orthogonal bases~\cite{strang1999discrete} and when the correct normalization factors are used, the bases form an orthonormal basis~\cite{hernandez1996first}. Any projection on orthonormal bases preserves the $L_2$ norm of vectors, so when taking only the first $k$ coefficients, we have $\Delta_2{D}^k\leq \Delta_2\s(t)$. Using the inequality between $L_1$ and $L_2$ norms, we have $\Delta_1{D}^k \leq \sqrt{k}\Delta_2{D}^k$.
\end{proof}
So, this method adds Laplace noise with a sensitivity of $\sqrt{k}\Delta_2\s$ to $DCT(\s(t))$, with $k$ being a publicly available hyperparameter. Note that the use of an equidistant-time grid is necessary for a meaningful discrete cosine transformation of the survival estimator vector, otherwise the coefficients would not represent details at gradually growing scales, correctly. 

We outline our \dps method in Algorithm~\ref{algo:dpsurv}, for the case of $C=0$. By definition, the Kaplan-Meier estimator should be a non-increasing function. To achieve this, after adding Laplace noise, masking with zeros and transforming back to the real-time space, we apply an isotonic regression-based clipping~\cite{chakravarti1989isotonic, de1977correctness} to $\s'(t)$ which has been shown to be a more effective post-processing technique compared to naive clipping~\cite{hay2009boosting}. This step is shown in line 5 of the algorithm. 

\begin{algorithm}[h]
 \caption{\small\texttt{\\DP-Surv}\\ Kaplan-Meier estimator values $\s(t)$ sampled at equidistant times $t=\{0, ..., \tm\}$ for total of $T$ time bins, total number of points $N$, $k$ number of first coefficients of $DCT(\s(t))$, privacy parameter $\varepsilon$.}
 \label{algo:dpsurv}%
\begin{algorithmic}[1]
 \small
    \State $\Delta_2\s \leftarrow \frac{\sqrt{T-1}}{N}$ \Comment{calculate $L_2$ sensitivity}
   \State $DCT'(\s(t)) \leftarrow DCT(\s(t)) + \mathcal{L}(0, \sqrt{k}\Delta_2\s/\varepsilon)$
   \State $DCT'(\s(t))\leftarrow DCT'(\s(t))[k+1:\tm]=0$ \Comment{choose the first $k$ coefficients of $DCT'(\s(t))$ and set the rest to zeros} 
   \State $\s'(t)\leftarrow DCT^{-1}(DCT'(\s(t)))$ \Comment{apply inverse DCT}
   \State $\s'(t) \leftarrow IRP(\s'(t))$ \Comment{isotonic regression projection}
   \State\Return{$\s'(t)$}
\end{algorithmic}
\end{algorithm}

\subsection{\dpy}
\label{sec:dpprobdef}
The next privacy-preserving method that we propose adds DP randomness to the probability mass function estimator $\yh(t)$. We call this method \textit{differentially private probability estimator} or \dpy for short. Here again, we have the advantage that no direct access to data is required and the probability estimator is modified directly. We also assume an equidistant-time grid when sampling the values of the probability estimator function, to address the issue of privacy of times of events. 

Consider $\yh(t)$ the vector of probability estimates that contains the sampled values of the continuous probability estimator function in equidistant time intervals. The sensitivities $L_1$ and $L_2$ of this function are as follows: 

\begin{theorem}
\label{theorem:ysensitivity}
Let's denote the total number of data points in the dataset by $N$ and the total number of censored points by $C$, then when $C=0$: 
{\small
\begin{alignat}{3}
&&&\Delta_1 \yh = \frac{2}{N}, \quad \Delta_2 \yh = \frac{\sqrt{2}}{N}\nonumber
\end{alignat}
}%
\end{theorem}
\begin{proof}
We provide the proofs for $L_1$ and $L_2$ sensitivity of the probability mass function estimator $\yh(t)$, in Appendix~\ref{sec:appendix-proof-yh}. 
\end{proof}

We show in Appendix~\ref{sec:appendix-proof-yh}, that the sensitivity calculation for when $C\neq 0$ is complicated and involves terms that make it non-DP. We again defer the investigation of a plausible sensitivity for this case to future work. 

In the absence of censored data and for a large enough dataset, $\Delta_1\yh$ is reasonably small such that we can apply the Laplace mechanism (Definition~\ref{def:laplace-mechanism}) directly to the vector $\yh(t)$. Our \dpy method is described in Algorithm~\ref{algo:dpprob}, for the cases of $C=0$. Again, we need to impose some properties on the noisy probability vector $\yh'$, after applying the DP mechanism. The probability estimator function, in general, should some up to 1 and the individual values $\yh(t)$ should be between 0 and 1, and this should also hold for $\yh'(t)$. For this purpose, we first clip the noisy values to a minimum value of 0 and then scale the whole vector by diving by the sum of all components. These are demonstrated in lines 3 and 4 of our algorithm.

\begin{algorithm}[h]
 \caption{\small\texttt{\\DP-Prob}\\ The vector of probability estimates $\yh(t)$ sampled at equidistant times $t=\{0, ..., \tm\}$ for total of $T$ time bins, total number of points $N$, privacy parameters $\varepsilon$}
 \label{algo:dpprob}%
\begin{algorithmic}[1]
 \small
    \State $\Delta_1\yh \leftarrow \frac{\sqrt{2}}{N}$ \Comment{calculate $L_1$ sensitivity}
   \State $\yh'(t) \leftarrow \yh(t) + \mathcal{L}(0, \Delta_1\yh/\varepsilon)$
   \State $\yh'(t) \leftarrow \texttt{clip}(\yh'(t), 0)$ \Comment{clip to min=0}
   \State $\yh'(t) \leftarrow \yh'(t)/\sum_{t=0}^{\tm+1} \yh'(t)$ \Comment{re-scale to make it a probability function}
   \State\Return{$\yh'(t)$}
\end{algorithmic}
\end{algorithm}


\section{Private Kaplan-Meier Estimator Across Multiple Sites}
\label{sec:paths}
In this section, we address the challenge of constructing a reliable and privacy-preserving KM curve over a dataset that is distributed across multiple sites.  

We have summarized the overall scheme of our solutions to this problem in the form of a graph in Figure~\ref{fig:overall-graph}. 
In what follows, we expand on this figure and clarify the possible routes that can be taken to arrive at the final goal of a global and private KM estimator $\s'(t)$. 

\begin{figure*}[!ht]
\centering
    \includegraphics[width=0.8\textwidth]{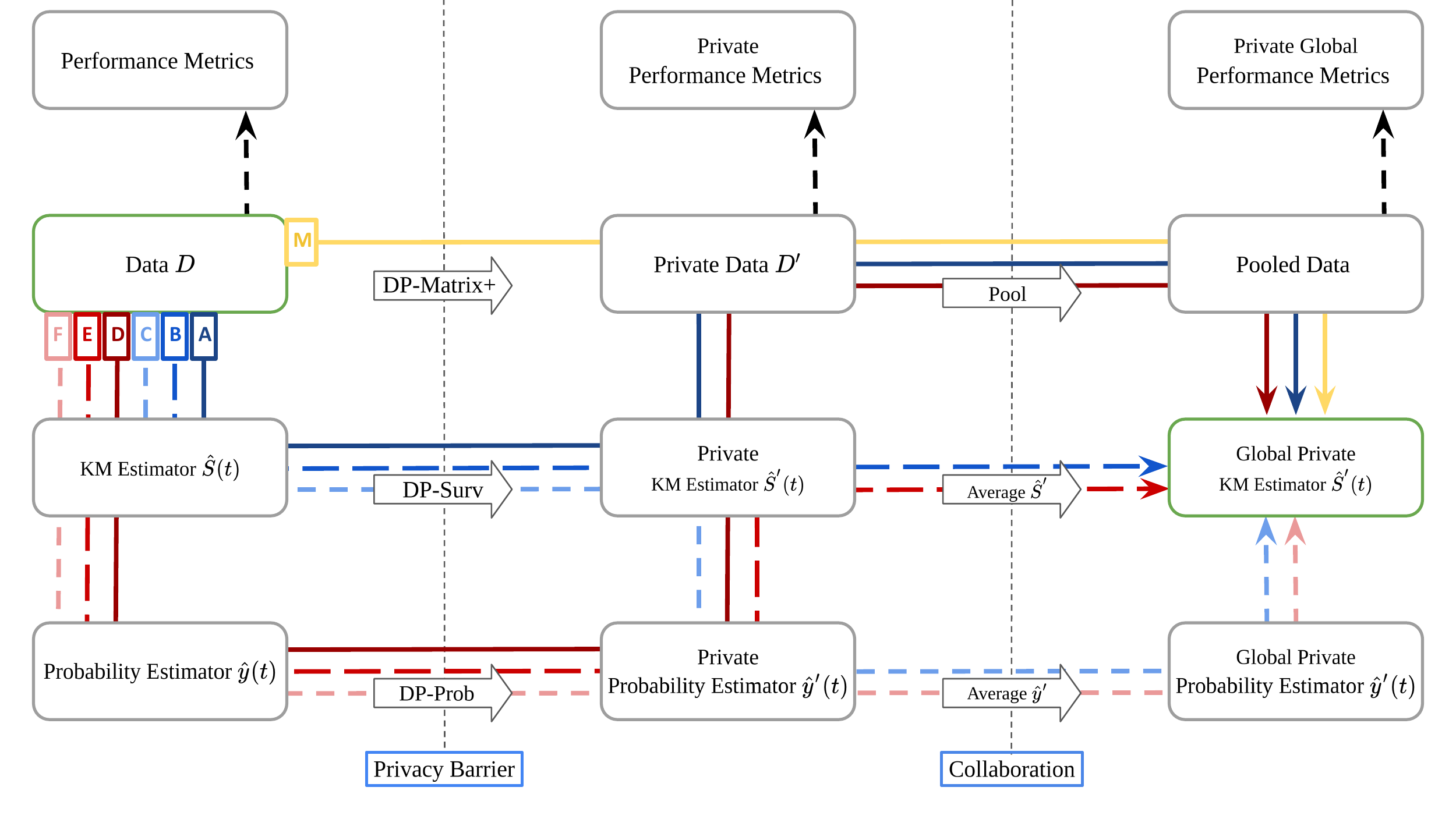}
    \caption{Overall scheme of paths that are possible to construct a collaborative private KM estimator over the union of datasets.}
    \label{fig:overall-graph}
\end{figure*}

\subsection{Vertical Movement in the Graph: Representation Conversion}
\label{sec:vertical}
A vertical movement at any of the ``stops'' in our graph will change the representation of the data. There are many reasons why one might want to move along the representations. We take advantage of different representations to add DP noise to different functions of the data, in order to assess the utility of our DP methods for a fixed level of privacy guarantee. In the end of this section we will also explain why a conversion back to a dataset is necessary to calculate the performance metrics for our methods. In the following, we will walk through conversion methods between these representations. 

\myparagraph{Data to Kaplan-Meier estimator} If a survival dataset is available (the second row from top), the KM estimator can easily be constructed using Equation~\ref{eq:kmestimator}. 

\myparagraph{Kaplan-Meier estimator to probability estimator and vice versa} If a KM estimator over a survival dataset is available (the third row from top), using Equation~\ref{eq:s2y}, the estimator for the probability mass function can easily be constructed. The conversion in the other way, from probability estimator (the lowest row) is also seamless by utilizing Equation~\ref{eq:y2s}. It is worth mentioning that no information is lost when converting between KM estimator and its counterpart, probability estimator, and we can easily convert between these two rows.

\myparagraph{Estimators to Data: Surrogate Dataset}
\label{sec:surrogate}
Kaplan-Meier estimator (and equally the probability estimator) summarizes the survival information contained in the dataset into one non-parametric curve over time~\cite{kaplan1958nonparametric, lee2018deephit, kvamme2019time}. For meta-analysis or computing more complicated metrics over the population, access to only KM/probability estimator functions is not sufficient, and we need the dataset. In our study, one might end up with access to only KM/probability functions (lower two rows) for two reasons: a) when only the KM/probability estimator is shared with other sites and b) when DP is used to modify these two functions. This motivates us to attempt to reconstruct a dataset based on the values of the probability estimator:
\begin{alignat}{3}
&&&\mathcal{R}_{\surr}: [0,1]^{\tm+1} \rightarrow \mathcal{D}
\end{alignat}
where $\mathcal{R}_{\surr}$ is a reconstruction function that takes probability values as input and outputs a surrogate dataset $D_{\surr}\subset\mathcal{D}$.

The problem of converting survival values to the corresponding dataset has been explored in, for example,~\cite{guyot2012enhanced, wei2017reconstructing}. Here, the most accurate version of the algorithm requires access to various parameters, such as number at risk at regular intervals during the time frame of the study and  total number of events $(\sum_{t_j=0}^{\tm} d_\ti)$ during the study period. When only the probability estimator or the KM estimator are provided, we do not have access to these two parameters, rather the overall probability of experiencing the event at each interval or the probability of survival at each interval, respectively. Previous work~\cite{guyot2012enhanced} states that when neither the total number nor multiple values are provided for the number at risk, we can assume that there are no censored observations. They mention that although this is a strong assumption, any other assumption about the data without further information would be just as strong.

Inspired by these arguments, we propose a simple, yet effective algorithm to construct a surrogate dataset with access to only probability mass function estimator. Algorithm~\ref{algo:surrogate} outlines the procedure. We assume that during the time frame of the study no censoring happens and the probabilities directly reflect the number of datapoints experiencing the event of interest at each time interval. As explained in Section~\ref{sec:probs-definition}, the extra element of the probability vector, $\yh({\tm+1})$, represents the probability that the event will occur after the maximum study time $\tm$. So we convert this value to censored data points at time $\tm$, as formulated in lines 9-12 of our algorithm. 

Since the conversion between the Kaplan-Meier estimator and the probability estimator is straightforward and lossless (see Section~\ref{sec:probs-definition}), when access to only KM values is granted, we can first convert to probability values and then apply Algorithm~\ref{algo:surrogate} to construct the surrogate dataset.  

\begin{algorithm}[h]
 \caption{\small\texttt{\\Surrogate Dataset Construction}    $\mathcal{R}_{\surr}$\\ Vector of survival probabilities $\yh = \{\yh_\ti\}_{t_j=0}^{\tm+1}$, tuple of datapoints $d^i=\{t^i, e^i\}$, number of data points to consider $n$, function to round to the nearest integer $round()$.}
 \label{algo:surrogate}%
\begin{algorithmic}[1]
 \small
 \State initialize empty surrogate dataset $D_{\surr}=[]$
    \For{$t_j = 0,...,\tm$} 
    \State $ num_\ti\leftarrow round(\yh_\ti*n)$
    \For{$i = 1,..., num_\ti$}
    \State $d^i=\{t_j, 1\}$
    \State$D_{\surr}$.append($d^i$)
    \EndFor
    \EndFor
    \State $num_{\tm+1} \leftarrow round(\yh_{\tm+1}*n)$
    \For{$i = 1,...,num_{\tm+1}$}
    \State $d^i=\{\tm, 0\}$
    \State $D_{\surr}$.append($d^i$)
    \EndFor
    \State\Return{$D_{\surr}$}

\end{algorithmic}
\end{algorithm}


\myparagraph{Data to Performance Metrics} As explained, when we have access to a real or surrogate dataset (the second row from top) we are much more flexible to run more complex metrics on the population to measure their survival properties. The measurement metrics will be explained in detail in Section~\ref{sec:exp-metrics} just before starting our experiments.

\subsection{Crossing the Privacy Barrier: DP methods}
\label{sec:horizontal-local}
In this section, we explain the horizontal movement in our graph across the \textit{Privacy Barrier} line (left column to the middle column). This is the stage in which each site uses differential privacy, locally, to construct a private dataset or a private function of their dataset. We look at 3 different DP methods to apply privacy, locally, as explained in Section~\ref{sec:method}. In the centralized setting, the adversary can be any external entity that has a view on any information that is released from the client's database. The adversary is \emph{passive} (i.e., honest-but-curious).

\myparagraph{\dpm}
When access to the full dataset and the times of events is provided, we can apply \dpm (see Section~\ref{sec:dpmdef}) directly to the count numbers at each event time.  

\myparagraph{\dps}
With more flexibility compared to the \dpm method, when only access to (discretized) survival functions is provided, each site can directly apply \dps (see Section~\ref{sec:dpsurvdef}) to their survival function. If the raw dataset is available, the KM curve is first constructed and then this DP method is applied to the KM function. If only a probability estimator function is provided, first the corresponding KM estimator is calculated, and then DP is applied. 

\myparagraph{\dpy}
When only access to the probability mass function is granted, DP is applied according to Section~\ref{sec:dpprobdef} to this function. If only the dataset is available, the probability function is first constructed, and then the DP noise is added to this vector. Again, when only the KM curve is provided, the conversion to the probability estimator function is easily possible through Equation~\ref{eq:s2y}.  

\subsection{Crossing the Collaboration Line}
\label{sec:horizontal-collaboration}
We will now continue in the horizontal direction, from client-level DP functions (middle column) to global, privacy-preserving survival statistics on all data (right column). The goal is to collect the private statistics from all sites and construct a KM estimator that uses the information contained in the union of the datasets of all the collaborating sites. A central server is responsible for collecting the survival statistics from all sites and aggregating them.  In the collaborative setting, the adversary can be the central server, other participants, or any other entity that has a view on any information released from the clients' side. The adversary is \emph{passive} (i.e., honest-but-curious), that is, it follows the protocol faithfully.

As explained in Section~\ref{sec:vertical}, we always have the option to convert to other representations along the vertical line. So regardless of which method we choose from Section~\ref{sec:horizontal-local}, the local sites can share one of the 3 different representations of the differentially private information with the central server for a joint calculation of the global model: 

\myparagraph{Private Data}
After applying DP to the data directly via \dpm we will be left with a differentially private dataset. However, applying \dps or \dpy only changes the KM function and the probability function, respectively. To go back to the space of datasets, we can deploy our surrogate data generation method from Section~\ref{sec:surrogate} to construct a differentially private dataset. Note that since the input of Algorithm~\ref{algo:surrogate} is the private probability estimator function, the constructed surrogate dataset will also be private due to the post-processing property of DP (Theorem~\ref{theorem:postprocess}). For when we choose the \dps method, we can convert the DP KM vector to a DP probability vector first using the Equation~\ref{eq:s2y} and then use the surrogate generation method. 
Each site can then share its private data set with the central server, and the metrics over the \textit{pooled data} can be used to construct a global and private KM estimator. 

\myparagraph{Private KM Estimators}
If we choose to convert our DP representations to a private local KM estimator, we can share this estimator with the central server. By inspecting Equation~\ref{eq:kmestimator}, for $K$ collaborating sites, we have:
\begin{eqnarray}
\s_{avg}(t) = \prod_{t'\leq t}\left(1 - \frac{\sum_{k=1}^K d_{t', k}}{\sum_{k=1}^K r_{t', k}}\right)
\end{eqnarray}
for $t \in \{0, ..., T_{\max, k}\}$ distinct times of events in the whole global dataset, $d_{t,k}$ being the number of data points experiencing the event of $e=1$ at distinct time $t$ for local site $k$ and $r_{t,k}$ being the risk set at distinct time $t$ for site $k$. We can see that it is mathematically not possible to calculate this average function solely based on the values of the local $\s'(t)$ that are shared, because we need access to the risk sets and number of events at a global level. We propose to estimate this by simply averaging the local KM estimators: 
\begin{eqnarray}
\s'_{avg}(t) =\frac{1}{N}\sum_{k=1}^Kn_k\s'_k(t)
\label{eq:avgsurv}
\end{eqnarray}
where $n_k$ is the dataset size for site $k$ and $N=\sum_{k=1}^Kn_k$ is the total number of points among all sites. Since we choose to work in the bounded differential privacy framework, as explained in Section~\ref{sec:dp-definition}, the size of the dataset can be shared publicly and we use this to make a weighted averaging over all sites. The final constructed private, global KM estimator can then be used, directly, or converted to its corresponding private surrogate dataset on the server side, to calculate the global metrics.

\myparagraph{Private Probability Estimators}
The last possible option is to use \dpy directly or to convert our private dataset or private KM estimator - which are obtained by \dpm or \dps, respectively - to a DP probability estimator $\yh'$. Unlike $\s(t)$, each of the local and private $\yh'(t)$'s is in the form of a probability mass function. In absence of auxiliary information, an effective method to combine probability mass functions is to take the average~\cite{hill2011combine} over all K sites: 
\begin{eqnarray}
\yh'_{avg}(t) = \frac{1}{N}\sum_{k=1}^K n_k\yh'_k(t)
\label{eq:avgy}
\end{eqnarray}
where again $n_k$ is the dataset size for site $k$ and $N$ is the total number of data points over all sites. The private global probability mass function can then be converted to its corresponding private surrogate dataset or the KM estimator to calculate the global metrics. 

\section{Experiments}
\label{sec:exp}
In this section we demonstrate the efficiency of our differentially private methods in a collaborative setting on real-world medical
datasets. We initially run the DP methods in a centralized setting in Section~\ref{sec:exp-centralized}. Then in Sections~\ref{sec:exp-collab} and~\ref{sec:exp-collab-uneven}, we progress to the main goal of our work, showing that our methods and suggested paths according to our workflow (as described in Figure~\ref{fig:overall-graph} of Section~\ref{sec:paths}) help us to accurately generate a joint and private Kaplan-Meier estimator over multiple clients.

\myparagraph{Datasets and Data Usage}
For our experiments, we chose 3 established publicly available survival medical datasets (details about characteristics and preprocessing can be found in Appendix~\ref{sec:app-dataset}). 

Since our proposed \dps (Section~\ref{sec:dpsurvdef}) and \dpy (Section~\ref{sec:dpprobdef}) methods are defined for datasets with no censored data, we only use the uncensored part of these 3 datasets for all our experiments. We include a detailed discussion of why this is a reasonable assumption and the shortcomings in Section~\ref{sec:discussion-censoring}. 
\subsection{Metrics}
\label{sec:exp-metrics}
\myparagraph{Logrank Test}
\label{sec:logrank}
In our experiments, we need to compare the quality of DP-generated KM curves and seek the KM distribution closest to the one generated from the original dataset. To compare KM distributions of two samples, hypothesis testing is usually used. The logrank test is the most common of such tests. 

The logrank test~\cite{mantel1966evaluation} is a nonparametric hypothesis test used to compare the survival distribution of two populations. \textbf{The null hypothesis states that the two populations have the same survival distribution}. So, a \pv smaller than the desired significance level leads to rejecting the null hypothesis, indicating a difference between the populations. A \pv greater than the significance level indicates no conclusive evidence to reject the null hypothesis. Common practice sets $p < 0.05$ as the threshold, ; however, much debate surrounds the topic, with many arguing that a much smaller value is needed~\cite{10.1001/jama.2019.4582, di2020statistical}. For details of the formulation of test statistics, refer to Appendix~\ref{sec:app-metrics}. 

\myparagraph{Median Survival Time and Survival Percentage} The logrank test has limitations in comparing survival curves for large \pv. This test also assumes noncrossing curves, so it may not accurately show similarities or differences for complex survival functions that intersect at any time point~\cite{bouliotis2011crossing}.
Given these constraints, other works~\cite[e.g.][]{guyot2012enhanced, wei2017reconstructing} recommend reporting median survival time and survival percentage at specific time points for a more comprehensive understanding. The median, which is the time at which the survival function reaches the value of $\s = 0.5$, is a robust measure and gives a general idea about the survival property of the dataset. Another important concept in survival studies is the behavior of the population at the beginning, middle and end of the study. Therefore, we choose to report the survival probability at three different time points $\{0.25\tm, 0.5\tm, 0.75\tm\}$, with $\tm$ being the maximum time in the study. The confidence intervals for these metrics are calculated directly from Kaplan-Meier (KM) curves using Greenwood's exponential \textit{log-log} formula~\cite{sawyer2003greenwood}. For more details on the confidence intervals, see Appendix~\ref{sec:app-metrics}. 
\myparagraph{Privacy Guarantees}
According to the postprocessing property of DP (Theorem~\ref{theorem:postprocess}), any function of a differentially private function is also differentially private. 

In \dpm, we directly add noise to the counts in the original dataset, thus any function of these noisy data, including the constructed $\s, \yh$, as well as the surrogate dataset that helps to calculate logrank test statistics and confidence intervals, is differentially private. 

For our \dps and \dpy methods, any function of these two functions is differentially private, and this includes the surrogate dataset used to calculate the test statistics and confidence intervals.

\subsection{Preliminary Experiments}
\label{sec:exp-prem}
We initially performed a series of preliminary experiments to check the soundness of our surrogate dataset generation algorithm and also set the necessary hyperparameters. These experiments can be found in Appendix~\ref{sec:app-exp-surrogates} and~\ref{sec:app-exp-hyperparameter} in detail. In summary, we found that a relatively small ($\sim 0.1\% - 1\%$ of the duration of the study) discretization binning size $b$ returns the most favorable operating point for the privacy/utility trade-off of our DP algorithms in a centralized setting. For \dps, we chose $b=\{1, 6, 2\}$ for GBSG, METABRIC, and SUPPORT, respectively. For \dpy, we chose $b=\{2, 4, 6\}$ for GBSG, METABRIC, and SUPPORT, respectively, and finally, for \dpm, we chose $b=\{2, 6, 6\}$ for GBSG, METABRIC, and SUPPORT, respectively. We also found that for \dps a value of $k=10\%$ for the first coefficients selected of the DCT works best. We will generalize these optimal values of the parameters to the decentralized experiments later. We also observed that our surrogate dataset generation method is robust with respect to the number $n$, we choose to populate the probability distribution with. We chose $n=\Bar{N}$ for all of our subsequent experiments, where $\Bar{N}$ is the number of uncensored data points in each dataset. 

\subsection{Centralized Performance of DP Algorithms}
\label{sec:exp-centralized}
We start with evaluation of our methods in the centralized setting, where all the data is available, centrally. The goal is to match the performance of a nonprivate KM curve with the best privacy guarantee (lower $\varepsilon$ values).

\myparagraph{Setup} To inspect the stability of Differential Privacy (DP) algorithms, which rely on random noise generation, we conducted 100 independent runs for each DP method. Given the unknown distribution of metrics after adding DP noise, we applied a bootstrapping~\cite{efron1993introduction, bootstrap} algorithm to determine the $95\%$ confidence interval for mean of the metrics. These metrics include the \pv, median, and survival percentage at three time points ($t = \{0.25\tm, 0.5\tm, \\ 0.75\tm\}$), representing the beginning, middle, and near the end of a study.

We run all our algorithms in the theoretically tight privacy regime~\cite{ponomareva2023dp} of $\varepsilon=\{0.5, 1\}$. Here, we analyze the results for the more stringent privacy value of $\varepsilon=0.5$, as the true capacity of our method (specifically \dps) becomes clear in this case. However, we include the complete results for $\varepsilon=1$ in Appendix~\ref{sec:app-centralizedDP}.

Table~\ref{tab:centralized-all-e05} shows the results of applying \dpm, \dps and \dpy to all datasets. For the non-DP baseline, we show the confidence intervals in parentheses. For DP methods, in parentheses we report the mean and its $95\%$ confidence interval. We also demonstrate these results for one random run of the DP algorithms and for two datasets in Figure~\ref{fig:centralized-e05}, where the blue line is the survival curve for the non-private dataset and the shaded blue region is its corresponding confidence area. 

\begin{table}[!ht]
\centering
\caption{Performance of the DP methods in the centralized setting for event $e=1$ and privacy budget $\varepsilon=0.5$.}
    \hspace*{-0.70cm}\scalebox{0.73}
    {
    \begin{tabular}{c|c|c|c|c|c|c}
    \toprule
         & & $p$- value& median& $25\%$ $T_{\text{max}}$ & $50\%$ $T_{\text{max}}$ & $75\%$ $T_{\text{max}}$ \\
         \midrule
          \multirow{4}{*}{\rotatebox[origin=c]{90}{\small{GBSG}}}&non-DP &-& $24 (22; 25)$ & $0.58 (0.55;0.60)$&$0.24 (0.22;0.26) $&$0.08 (0.07;0.11) $\\
         \cline{2-7}
         
          &\dps& $0.34(0.33, 0.35)$ & $24(24, 24)$ & $0.58(0.58, 0.58)$ & $0.25(0.24, 0.25)$& $0.08(0.08, 0.08)$\\
          
         \cline{2-7}
         
          &\dpy& $0.21(0.16, 0.27)$ & $25(25, 25)$ & $0.58(0.57, 0.58)$&$ 0.26(0.26, 0.26)$&$0.09(0.08, 0.09) $\\
         \cline{2-7}
         
          &\dpm& $0.30(0.23, 0.36)$ & $25(24, 25)$ & $0.57(0.57, 0.58)$&$0.24(0.23, 0.24)$&$0.04(0.04, 0.05)$\\
         
          \bottomrule
          \toprule
          
          \multirow{4}{*}{\rotatebox[origin=c]{90}{\small{METABRIC}}}& non-DP &-& $86 (81; 90)$ & $0.49 (0.46;0.51)$&$0.16 (0.14;0.18) $&$0.02 (0.01;0.03) $\\
         \cline{2-7}

          &\dps& $0.25(0.20, 0.30)$&$86(85, 86)$&$0.49(0.49, 0.49)$ & $0.18(0.17, 0.18)$& $0.02(0.02, 0.02)$\\
          
         \cline{2-7}
         
          &\dpy& $0.02(0.00, 0.04)$ & $91(91, 92)$&$0.51(0.50, 0.51)$&$ 0.19(0.19, 0.19)$&$0.05(0.05, 0.05) $\\

         \cline{2-7}
         
          &\dpm& $0.16(0.11, 0.21)$ & $87(86, 88)$ & $0.50(0.50, 0.51)$&$0.13(0.12, 0.13)$&$0.01(0.01, 0.02)$\\
          \bottomrule
          \toprule
          
          \multirow{4}{*}{\rotatebox[origin=c]{90}{\small{SUPPORT}}}&non-DP &-& $57 (53; 61)$ & $0.14 (0.13;0.15)$&$0.05 (0.04;0.05) $&$0.01 (0.01;0.01) $\\
         \cline{2-7}

          &\dps& $0.26(0.21, 0.32)$ & $59 (57, 61)$ & $0.14(0.14, 0.15)$ & $0.05(0.05, 0.05)$& $0.01(0.01, 0.01)$\\
          
         \cline{2-7}
         
          &\dpy& $0.00(0.00, 0.00)$ & $66(66, 67)$ & $0.17(0.17, 0.18)$&$ 0.08(0.08, 0.08)$&$0.03(0.03, 0.03) $\\
         \cline{2-7}
          
          &\dpm& $0.11(0.08, 0.15)$ & $60(60, 60)$ & $0.13(0.12, 0.13)$&$0.01(0.00, 0.01)$&$0.00(0.00, 0.00)$\\
          \bottomrule
    \end{tabular}
    }
    \label{tab:centralized-all-e05}
\end{table}

\begin{figure}[]
\hspace*{-0.5cm}\begin{minipage}[c]{0.56\columnwidth}
        \centering
        \includegraphics[width=\linewidth]{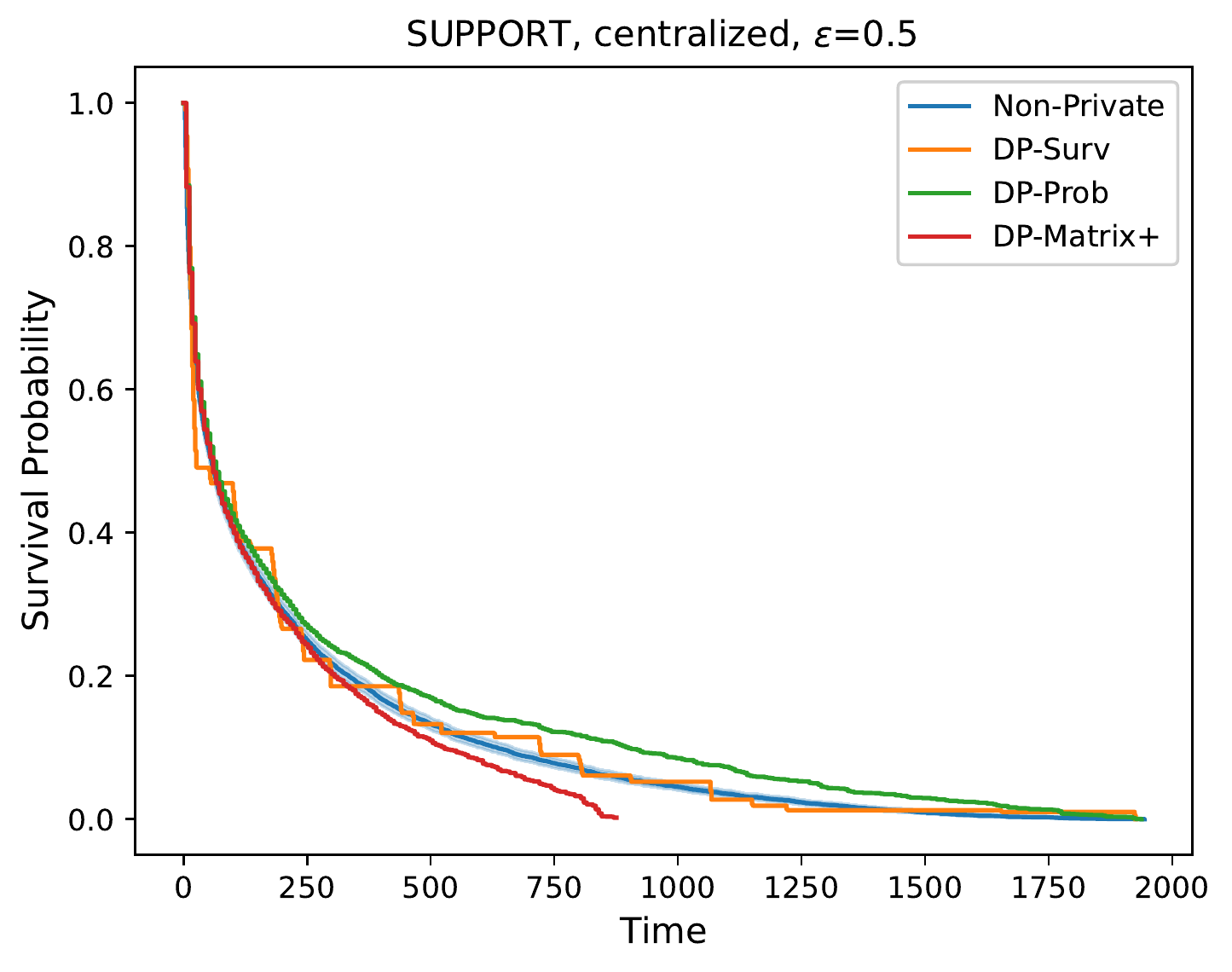}
\end{minipage}%
\hspace*{-0.1cm}
\begin{minipage}[c]{0.56\columnwidth}
        \centering
        \includegraphics[width=\linewidth]{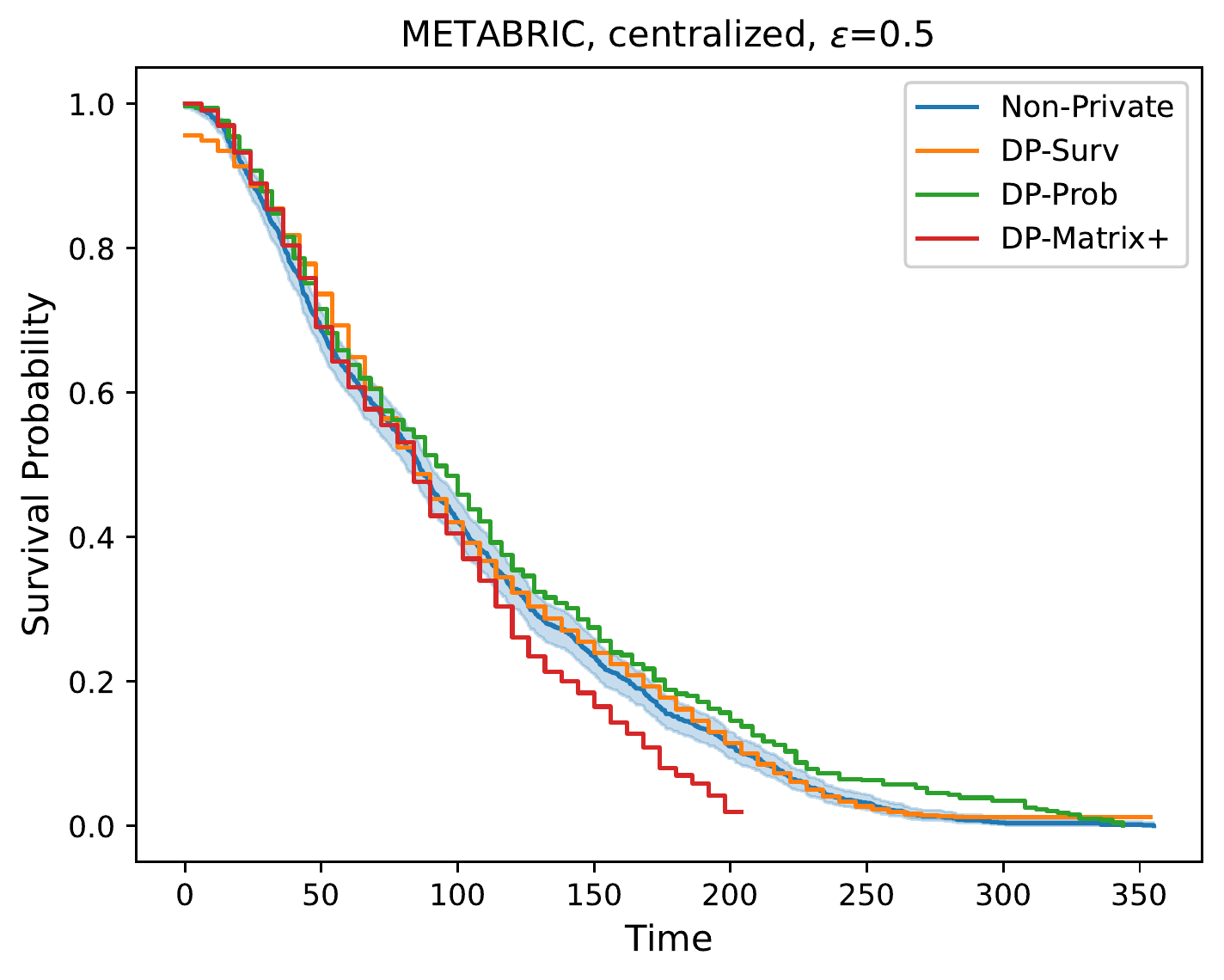}
\end{minipage}
\caption{Comparison of all the DP methods in a centralized setting, for one random run of the DP algorithms. The blue shaded region shows the confidence area of the non-private dataset.}
\label{fig:centralized-e05}
\end{figure}

\myparagraph{\dps Performance} Our \dps method shows consistent performance across multiple runs of the DP algorithm for all datasets with very tight $95\%$ confidence interval. Its $p-$value is significantly above the commonly-acceptable statistical difference of $0.05$, and the samples' confidence interval of the median as well as survival percentages fall within the confidence interval for all the 3 non-private datasets in all cases. 

\myparagraph{\dpy Performance} Our second method also shows stability in confidence intervals for metric means. However, it underperforms with the METABRIC and SUPPORT datasets, particularly for the SUPPORT dataset, which has an early sharp decline in survival rate (this can be seen better in Figure~\ref{fig:centralized-e05}). This results in poor performance in accurately matching metrics, especially since the median is in a very sensitive area where many events occur simultaneously. This issue seems to affect \dpy more negatively compared to \dps which shows a good approximation even for the median. It is also observed that \dpy tends to overestimate survival percentages. This phenomenon arises because across all datasets (most notably in SUPPORT), there are several time bins at the onset of the study with very high probability estimates (indicative of high event rates), and numerous bins subsequently exhibiting near zero probability values. So when the Laplace noise is added to the probability estimator vector, numerous negative values are manifested in the latter part of the study. The postprocessing step of line 3 in Algorithm~\ref{algo:dpprob}, eliminates these negative values, followed by the rescaling of the entire vector using an overestimated sum value (line 4 of Algorithm~\ref{algo:dpprob}). This procedure results in the positive noisy event probabilities being smaller than anticipated, thereby leading to overesitmated survival rates for these periods.  

\myparagraph{\dpm Performance} \dpm also struggles to match the metrics for METABRIC and SUPPORT. We especially see a degradation in the performance of \dpm towards the endpoint of the study at $0.75\% T_{\max}$ in all datasets. The reason is probably the postprocessing step as described in Equation~\ref{eq:dpmpp} (corresponding to line 5 of Algorithm 1 in the original paper~\cite{gondara2020differentially}), where the noisy at-risk group is calculated based on the previous step. If we restrict the noisy death numbers to be positive, the noise causes the risk set to drop faster than the original dataset and deplete quickly towards the end of the study.


\subsection{Collaboration}
\label{sec:exp-collab}
We next evaluate the performance of our DP methods in constructing a private KM curve using data from multiple collaborating sites. Our objective is to achieve a KM curve that closely approximates the one derived from aggregated data (centralized setup) while maintaining an acceptable level of privacy.

According to our overall workflow as shown in Figure~\ref{fig:overall-graph}, to utilize \dps, \dpy or \dpm, we can take one of the 7 possible routes for this privacy-preserving collaboration:
\begin{itemize}
    \item \dps pooled (path \textbf{A}): \dps is applied locally by each client, then private surrogate datasets are generated and shared with the central server. The pooled collection of all private datasets is used to construct the final KM curve.
    \item \dps averaged $\s'$ (path \textbf{B}): \dps is applied locally, local private KM curves $\s'$ are shared, an average is taken over them and then a global private surrogate dataset is generated to calculate the metrics.
    \item \dps averaged $\yh'$ (path \textbf{C}): \dps is applied locally, local private probability function is calculated and shared, an average is taken over these local $\yh'$s and a global private surrogate dataset is generated using the average. The final private and global KM estimator is calculated using this surrogate dataset. 
    
    \item \dpy pooled (path \textbf{D}): \dpy is applied locally by each client, private local surrogate datasets are generated and shared with the central server. The pooled collection of all private datasets is used to construct the final KM curve and metrics.
    \item \dpy averaged $\s'$ (path \textbf{E}): \dpy is applied locally, local private KM curves $\s'$ are calculated as a function of local $\yh'$ by clients and shared, an average is taken over them and then a global private surrogate dataset is generated to calculate the metrics and construct the final global and private KM curve. 
    \item \dpy averaged $\yh'$ (path \textbf{F}): \dpy is applied locally, local private probability function is calculated and shared, an average is taken over these local $\yh'$s and a global private surrogate dataset is generated using the average. This dataset is then used to calculate the metrics for the global, private KM estimator. 
    \item \dpm pooled (path \textbf{M}): \dpm is applied locally by clients, based on the noisy count numbers private local datasets are generated and shared with the central server. The pooled private datasets are used to construct the final KM curve.  
\end{itemize}

\myparagraph{Setup} We consider the case where 10 clients collaborate to jointly build a KM estimator over the collection of their datasets. Each dataset is shuffled and split into 10 parts with equal number of datapoints for each client. To assess the stability of our DP methods, we conduct 100 independent runs and report the $95\%$ confidence intervals for the mean of the \pv, median, and survival percentages at $t=\{0.25T_{\max}, 0.5T_{\max}, 0.75T_{\max}\}$, using bootstrapping. Given the increased complexity of collaborative learning compared to centralized learning, we perform all experiments with $\varepsilon=\{1, 3, 5\}$. Here, we analyze the results for $\varepsilon=1$, with the complete results and analysis for $\varepsilon=\{3, 5\}$ included in Appendix~\ref{sec:app-collaboration}. 

\myparagraph{Sensitivity calculation and privacy} A key challenge when employing differential privacy across multiple sites is determining the appropriate sensitivity and privacy budget. To address this, we standardize all hyperparameters, including the dataset's maximum time, discretization bin size $b$ and the fraction of DCT coefficients $k$ across all collaborating clients. Assuming bounded differential privacy, the local size of the uncensored datasets $\nb_k$ is considered public, allowing us to use this number to calculate the sensitivity of the locally added noise.

Table~\ref{tab:even-e1}, shows the results for the privacy budget $\varepsilon=1$ and the 7 possible paths. For DP methods, we report the $95\%$ confidence interval of the means of the metrics in parentheses. For the non-private centralized dataset we also report the confidence intervals in parentheses. 
\begin{table*}[!ht]
    \centering
        \caption{Collaboration with even data split for $e=1$}
    \scalebox{0.80}{
    \begin{tabular}{c|c|c|c|c|c|c|c}
    \toprule
         & & &$p$- value& median survival time& $25\%$ $T_{\text{max}}$ & $50\%$ $T_{\text{max}}$ & $75\%$ $T_{\text{max}}$ \\
         \midrule
          \multirow{7}{*}{\rotatebox[origin=c]{90}{GBSG}}& centralized, non-private  & &-& $24 (22; 25)$ & $0.58 (0.55;0.60)$&$0.24 (0.22;0.26) $&$0.08 (0.07;0.10) $\\
         \cline{2-8}
         
         && pooled&$0.17(0.12, 0.22)$ & $24(24, 24)$ & $0.57(0.57, 0.58)$ & $0.25(0.24, 0.25)$ &$0.09(0.08, 0.09)$ \\
         &\dps ($\varepsilon=1$)&averaged $\s'$& $0.22(0.16, 0.27)$& $24(24, 25)$&$0.58(0.58, 0.59)$&$0.25(0.24, 0.25)$&$0.09(0.08, 0.09)$\\
         &&averaged $\yh'$&$0.17(0.12, 0.21)$&$24(24, 24)$&$0.58(0.57, 0.58)$&$0.25(0.25, 0.26)$&$0.09(0.09, 0.09)$\\
         \cline{2-8}
          
         && pooled&$0(0.00, 0.00)$ & $30(29, 30)$ & $0.65(0.64, 0.65)$ & $0.36(0.35, 0.36)$ &$0.16(0.16, 0.16)$ \\
         &\dpy ($\varepsilon=1$)&pveraged $\s'$&$0(0.00, 0.00)$&$30(29, 30)$&$0.65(0.64, 0.65)$&$0.35(0.35, 0.36)$&$0.16(0.16, 0.17)$\\
         &&pveraged $\yh'$&$0(0.00, 0.00)$&$30(30, 30)$&$0.65(0.65, 0.65)$&$0.36(0.35, 0.36)$&$0.16(0.16, 0.16)$\\
         \cline{2-8}
         
         &\dpm ($\varepsilon=1$)&pooled&$0(0.00, 0.00)$&$20(20, 21)$&$0.50(0.49, 0.50)$&$0.06(0.05, 0.06)$&$0.02(0.02, 0.02)$\\

    \bottomrule
    \toprule
         \multirow{7}{*}{\rotatebox[origin=c]{90}{METABRIC}}&centralized, non-private  &&-& $86 (81; 90)$ & $0.49 (0.46;0.51)$&$0.16 (0.14;0.18) $&$0.02 (0.01;0.03) $\\
         \cline{2-8}
          
         && pooled&$0.11(0.07, 0.14)$ & $85(84, 86)$ & $0.49(0.48, 0.49)$ & $0.18(0.18, 0.18)$ &$0.04(0.04, 0.04)$ \\
         &\dps ($\varepsilon=1$)&averaged $\s'$&$0.07(0.03, 0.10)$&$85(84, 86)$&$0.49(0.48, 0.49)$&$0.18(0.18, 0.19)$&$0.04(0.04, 0.04)$\\
         &&averaged $\yh'$&$0.07(0.04, 0.10)$&$85(84, 86)$&$0.49(0.48, 0.49)$&$0.18(0.18, 0.18)$&$0.04(0.04, 0.04)$\\
         \cline{2-8}
          
         && pooled&$0(0.00, 0.00)$ & $110(109, 110)$ & $0.60(0.60, 0.60)$ & $0.29(0.29, 0.30)$ &$0.12(0.11, 0.12)$ \\
         &\dpy ($\varepsilon=1$)&averaged $\s'$&$0(0.00, 0.00)$&$110(109, 111)$&$0.60(0.60, 0.60)$&$0.30(0.30, 0.30)$&$0.12(0.12, 0.12)$\\
         &&averaged $\yh'$&$0(0.00, 0.00)$&$110(109, 110)$&$0.60(0.60, 0.60)$&$0.29(0.29, 0.30)$&$0.12(0.11, 0.12)$\\
         \cline{2-8}
         
         &\dpm ($\varepsilon=1$)&pooled&$0(0.00, 0.01)$&$79(78, 80)$&$0.45(0.44, 0.46)$&$0.05(0.04, 0.05)$&$0.02(0.02, 0.03)$\\

    \bottomrule
    \toprule
          \multirow{7}{*}{\rotatebox[origin=c]{90}{SUPPORT}}&centralized, non-private  &&-& $57 (53; 61)$ & $0.14 (0.13;0.15)$&$0.05 (0.04;0.05) $&$0.01 (0.01;0.01) $\\
         \cline{2-8}
         
         && pooled&$0.05(0.02, 0.08)$ & $66(64, 69)$ & $0.15(0.15, 0.15)$ & $0.06(0.05, 0.06)$ &$0.02(0.02, 0.02)$ \\
         &\dps ($\varepsilon=1$)&averaged $\s'$&$0.05(0.02, 0.08)$&$66(64, 69)$&$0.15(0.14, 0.15)$&$0.06(0.05, 0.06)$&$0.02(0.02, 0.02)$\\
         &&averaged $\yh'$&$0.09(0.04, 0.14)$&$68(65, 70)$&$0.15(0.14, 0.15)$&$0.06(0.05, 0.06)$&$0.02(0.02, 0.02)$\\
         \cline{2-8}
         
         && pooled&$0(0.00, 0.00)$ & $551(548, 554)$ & $0.53(0.53, 0.53)$ & $0.34(0.34, 0.34)$ &$0.17(0.17, 0.17)$ \\
         &\dpy ($\varepsilon=1$)&averaged $\s'$&$0(0.00, 0.00)$&$575(572, 578)$&$0.54(0.54, 0.54)$&$0.35(0.35, 0.35)$&$0.17(0.17, 0.17)$\\
         &&averaged $\yh'$&$0(0.00, 0.00)$&$576(574, 579)$&$0.54(0.54, 0.54)$&$0.35(0.35, 0.35)$&$0.17(0.17, 0.17)$\\
         \cline{2-8}
         
         &\dpm ($\varepsilon=1$)&pooled&$0(0.00, 0.00)$&$54(53, 55)$&$0.01(0.01, 0.01)$&$0.01(0.00, 0.01)$&$0.01(0.00, 0.01)$\\
         
          \bottomrule
    \end{tabular}
    }
    \label{tab:even-e1}
\end{table*}

\myparagraph{Performance of \dps-Based Methods} At first glance, we observe that for all datasets, the DP-based methods consistently achieve an acceptable mean \pv, adhering to the common significance level of $0.05$. For the GBSG and METABRIC datasets, the estimated private median times and their confidence intervals lie within the confidence intervals of the non-DP centralized estimates. Similarly, the estimated private survival percentages for GBSG fall within the non-DP confidence interval. For METABRIC and SUPPORT, the confidence intervals of the estimated survival percentages deviate by at most $1\%$ from those of the non-private dataset. The private median estimation for SUPPORT deviates by at most 7 units of time from the non-private confidence interval. This deviation is attributed to the median of this rapidly declining survival dataset being in a sensitive region with a steep slope of change over time.

An interesting observation when comparing paths A, B, and C is their comparable performance. For the mean survival percentages, the difference between these three paths does not exceed $1\%$ in any dataset. Similarly, the difference in the estimated mean median is at most 2 units of time across all datasets. This indicates that averaging the private survival functions or the private probability functions is a viable solution for collaboration schemes, closely approximating the outcome of sharing private datasets. This significant finding provides clients with the flexibility to choose their preferred path for jointly building a model.


Similar to the centralized application of \dps, we observe very stable results between multiple runs of the algorithm. The $95\%$ confidence intervals for the mean survival percentages show a maximum difference of $1\%$ from the mean value. The median's confidence interval is at most 3 units of time away from the estimated mean. Overall, we witness robust and stable performance across the three \dps-based paths.  

\myparagraph{Performance of \dpy-Based Methods} For this privacy regime, the \dpy-based paths do not perform as well as the \dps method according to \pv. This indicates that the \dpy method is more sensitive to the amount of DP-noise added for a specific $\varepsilon$ level, compared to \dps. 

Compared to \dps-based methods, we see more deviation between different runs of the DP algorithm for paths D, E, and F, particularly noticeable in the median estimate of SUPPORT. This is again due to the more sensitive response of \dpy to noise. Additionally, we observe the same issue of overestimating survival percentages, as we described in the centralized experiments in Section~\ref{sec:exp-centralized}.

\myparagraph{Performance of \dpm-Based Methods} We observe that \dpm fails in \pv for all dataset in this stringent privacy regime. It underestimates the mean of the median for all datasets. It also suffers from the same problem of under estimating the survival percentages, especially towards the end of the study, as explained in Section~\ref{sec:exp-centralized}. This issue is especially noticeable in SUPPORT where all the estimated mean survival percentages fall to 0.01 from $25\%\tm$ time point.

\subsection{Collaboration: A Broader View}
\label{sec:exp-collab-uneven}

\begin{table*}[ht]
    \centering
        \caption{Collaboration with uneven data split with one site receiving either $50\%$ or $5\%$ of all of the data, for $e=1$ and $\varepsilon=1$}
    \scalebox{0.80}{
    \begin{tabular}{c|c|c|c|c|c|c|c}
    \toprule
         & & &$p$- value& median survival time& $25\%$ $T_{\text{max}}$ & $50\%$ $T_{\text{max}}$ & $75\%$ $T_{\text{max}}$ \\
         \hline
          \multirow{7}{*}{\rotatebox[origin=c]{90}{GBSG}}& centralized, non-private  & &-& $24 (22; 25)$ & $0.58 (0.55;0.60)$&$0.24 (0.22;0.26) $&$0.08 (0.07;0.10) $\\
         \cline{2-8}
         
         && pooled&$0.19(0.14, 0.24)$ & $24(24, 25)$ & $0.58(0.58, 0.58)$ & $0.25(0.25, 0.26)$ &$0.09(0.09, 0.09)$ \\
         &\dps ($\varepsilon=1$)&averaged $\s'$& $0.17(0.12, 0.21)$& $24(24, 25)$&$0.58(0.57, 0.58)$&$0.25(0.25, 0.26)$&$0.09(0.09, 0.09)$\\
         &minority has $50\%$&averaged $\yh'$&$0.19(0.14, 0.23)$&$24(24, 25)$&$0.58(0.57, 0.58)$&$0.25(0.24, 0.25)$&$0.09(0.09, 0.09)$\\
         \cline{2-8}
         && pooled&$0.17(0.12, 0.22)$ & $24(24, 25)$ & $0.58(0.57, 0.58)$ & $0.25(0.25, 0.25)$ &$0.09(0.09, 0.09)$ \\
         &\dps ($\varepsilon=1$)&averaged $\s'$& $0.13(0.09, 0.16)$& $24(24, 24)$&$0.58(0.58, 0.59)$&$0.25(0.25, 0.25)$&$0.09(0.09, 0.10)$\\
         &minority has $5\%$&averaged $\yh'$&$0.18(0.13, 0.22)$&$24(24, 25)$&$0.58(0.58, 0.59)$&$0.25(0.24, 0.25)$&$0.09(0.08, 0.09)$\\

    \bottomrule
    \toprule
         \multirow{7}{*}{\rotatebox[origin=c]{90}{METABRIC}}&centralized, non-private  &&-& $86 (81; 90)$ & $0.49 (0.46;0.51)$&$0.16 (0.14;0.18) $&$0.02 (0.01;0.03) $\\
         \cline{2-8}
          
         && pooled&$0.08(0.04, 0.11)$ & $85(84, 86)$ & $0.49(0.48, 0.49)$ & $0.18(0.18, 0.19)$ &$0.04(0.04, 0.04)$ \\
         &\dps ($\varepsilon=1$)&averaged $\s'$&$0.07(0.03, 0.09)$&$85(84, 85)$&$0.49(0.48, 0.49)$&$0.19(0.18, 0.19)$&$0.04(0.04, 0.04)$\\
         &minority has $50\%$&averaged $\yh'$&$0.08(0.05, 0.11)$&$85(84, 85)$&$0.49(0.48, 0.49)$&$0.18(0.18, 0.19)$&$0.04(0.03, 0.04)$\\
         \cline{2-8}
         && pooled&$0.06(0.03, 0.08)$ & $86(85, 87)$ & $0.49(0.49, 0.50)$ & $0.18(0.18, 0.19)$ &$0.04(0.04, 0.04)$ \\
         &\dps ($\varepsilon=1$)&averaged $\s'$&$0.07(0.04, 0.10)$&$85(84, 86)$&$0.49(0.48, 0.49)$&$0.18(0.18, 0.19)$&$0.04(0.04, 0.04)$\\
         &minority has $5\%$&averaged $\yh'$&$0.08(0.05, 0.11)$&$85(85, 86)$&$0.49(0.49, 0.50)$&$0.18(0.18, 0.18)$&$0.04(0.04, 0.04)$\\

    \bottomrule
    \toprule
          \multirow{7}{*}{\rotatebox[origin=c]{90}{SUPPORT}}&centralized, non-private  &&-& $57 (53; 61)$ & $0.14 (0.13;0.15)$&$0.05 (0.04;0.05) $&$0.01 (0.01;0.01) $\\
         \cline{2-8}
         
         && pooled&$0.05(0.01, 0.07)$ & $68(66, 70)$ & $0.15(0.14, 0.15)$ & $0.06(0.05, 0.06)$ &$0.02(0.02, 0.02)$ \\
         &\dps ($\varepsilon=1$)&averaged $\s'$&$0.04(0.02, 0.05)$&$71(68, 74)$&$0.15(0.14, 0.15)$&$0.05(0.05, 0.06)$&$0.02(0.02, 0.02)$\\
         &minority has $50\%$&averaged $\yh'$&$0.05(0.01, 0.09)$&$68(66, 71)$&$0.15(0.15, 0.15)$&$0.06(0.05, 0.06)$&$0.02(0.02, 0.02)$\\
         \cline{2-8}
         && pooled&$0.04(0.01, 0.05)$ & $67(66, 71)$ & $0.15(0.14, 0.15)$ & $0.05(0.05, 0.06)$ &$0.02(0.02, 0.02)$ \\
         &\dps ($\varepsilon=1$)&averaged $\s'$&$0.07(0.03, 0.10)$&$65(62, 68)$&$0.15(0.14, 0.15)$&$0.06(0.05, 0.06)$&$0.02(0.02, 0.02)$\\
         &minority has $5\%$&averaged $\yh'$&$0.05(0.02, 0.07)$&$65(63, 68)$&$0.15(0.15, 0.15)$&$0.06(0.05, 0.06)$&$0.02(0.02, 0.02)$\\
         
          \bottomrule
    \end{tabular}
    }
    \label{tab:uneven-e1}
\end{table*}
So far we only studied the case of even split of data, where each client has the same number of data points as the others. We now would like to explore more realistic and challenging scenarios, including uneven data distribution (similar to previous works on distributed medical data~\cite[e.g.][]{duan2018odal}). For this reason, we examine two cases: a) one client has $50\%$ of all data, and b) one client has only $5\%$ of the data. 

Our results in the previous section show that the \dps-based paths (A, B, and C) work best for collaboration with an even split of data and \dpy and \dpm underperform in comparison. For this reason, in this section we only inspect the \dps-based paths. 

\myparagraph{Setup} We again consider that we have 10 collaborating clients. Data is first shuffled and split between these clients. One \textit{minority} client receives either $5\%$ or $50\%$ of the total amount of data and the rest is evenly shared between the 9 remaining participants. To ensure the stability of our DP method, we perform 100 random runs for each setting of our algorithms and report the mean of the metrics along with the $95\%$ confidence interval of the mean, determined through bootstrapping of the samples. In a similar fashion to the even split of the data, we explore the privacy regime $\varepsilon=\{1, 3, 5\}$. Here, we provide the results for $\varepsilon=1$, but all the results and analysis for $\varepsilon=\{3, 5\}$ are provided in Appendix~\ref{sec:app-collaboration}. 

\myparagraph{Results} Table~\ref{tab:uneven-e1} presents the results for all \dps-based paths under both data splits: minority client receiving either $5\%$ or $50\%$. Our methods demonstrate stability across multiple algorithm runs, with confidence intervals for survival percentages differing by at most $1\%$ from the mean. The confidence intervals for medians show a maximum deviation of 4 units of time from the mean, with the largest interval observed for the SUPPORT dataset. An intriguing observation is the consistent performance of our \dps-based pipeline under different data splits. The private mean survival percentages exhibit differences of at most $1\%$ among paths A, B, and C across all three datasets. Similar results are noted for the estimated private mean of the median in GBSG and METABRIC. The median estimate for SUPPORT proves more challenging due to its location in a high-slope region of the curve. A significant finding is that the estimated mean survival percentages deviate by at most $1\%$ from the confidence interval of the non-private centralized dataset. Additionally, the estimated median falls completely within the confidence interval of the non-DP dataset for GBSG and METABRIC, while for SUPPORT, it deviates by at most 10 units of time. \textbf{This is an important finding and shows that without prior knowledge about the accuracy of the local estimator, there is always an incentive for individual data holders to collaborate for a better estimation of the KM curves. Given that we are applying tight privacy guarantees, the privacy of the datasets of these individual collaborators will not be compromised.} 

\myparagraph{Summary of Our Findings}
Through our experiments we found out that:
\begin{enumerate}
    \item Our surrogate dataset generation method is a reliable way to generate a surrogate dataset that match the performance of the real dataset, with access to only the probability mass function of the data.
    \item Our \dps method shows a near perfect performance in a centralized setting and for very low privacy budgets. 
    \item Our \dps-based collaboration paths consistently demonstrate comparability, stability, and accuracy in estimating Kaplan-Meier curves, closely aligning with the non-DP centralized setting where all data are assumed to be stored on a central server.
    \item \dps-based paths can successfully be used for uneven data splits and offer a strong incentive for collaboration among multiple data centers.  
\end{enumerate}

\section{Related Work}
\label{sec:related-work}

The power of Kaplan-Meier estimators, especially for medical applications, lies in the fact that they are non parametic models and can directly be constructed from the data and readily used to draw conclusions. Therefore, these are widely used in the medical domain for treatment assessment~\cite[e.g.][]{castro1995results, rossi2012carboplatin}, gene expression affect on survival~\cite[e.g.][]{mihaly2013meta, glinsky2004gene}, etc. 

Survival datasets are usually distributed among multiple data collectors such as hospitals or banks. To construct more accurate Kaplan-Meier estimators access to more data and thus a collaboration between these centers is necessary.  
In many applications, and especially the medical survival analysis, these data contain sensitive information about the individuals and protection of privacy of these individuals is a matter of utmost concern. Naturally, there are many privacy regulations~\cite[e.g.][]{politou2018forgetting, voigt2017eu, beverage1976privacy, hodge1999legal} that prohibit the sharing of raw data with other centers. Attempts to overcome this issue and to construct a KM estimator in collaboration with multiple centers have mostly focused on secure multi party computation (SMPC)\cite{froelicher2021truly, vogelsang2020secure, von2021privacy} of KM curves based on secure calculation of statistics needed to construct the estimator. However, there are many issues with this approach. Firstly, SMPC schemes do not scale well to larger settings: the cost of computation and communication usually grows very fast. Second, even after using a secure scheme, there is still privacy risks for the dataset when the summary statistics are shared publicly. An outside adversary can still perform attacks such as re-identification~\cite{el2011systematic, rocher2019estimating} or inference~\cite{backes2016membership, homer2008resolving} on summary statistics. 

A practical and strong method to guarantee the privacy of the dataset is using differential privacy~\cite{Dwork2014book}. In contrast to SMPC, differential privacy by definition has the power to neutralize adversarial attacks. DP can be applied either directly on the dataset or functions of the dataset. One way to incorporate DP in the Kaplan-Meier estimation is to add Laplace noise to the number counts in the survival dataset~\cite{gondara2020differentially}. This method is restrictive, because always an access to the number counts at specific times of events is required. It also does not offer privacy for the times of events and these will still be published with no DP randomness applied to them. 

In our paper, we take advantage of the probability density estimator~\cite{kvamme2019time, lee2018deephit}, which is an alternative statistic, closely related to the Kaplan-Meier function, to construct surrogate datasets solely based on KM function or probability function. This allows us to offer DP methods that are directly applicable on these two functions and readily converting between summary statistics and (surrogate) dataset. Our first DP method which is inspired by~\cite{kerkouche2021compression, rastogi2010differentially}, tweaks the KM function in its discrete cosine space. Our second DP method tweaks the probability function. By sampling these two functions in their time dimension, we are able to offer privacy on the times of events. Our methods show improvement in the privacy budget ($\varepsilon$) spending for the same utility compared to the previously-suggested method~\cite{goel2010understanding} and this allows us to expand our methods to a collaborative setting. 
 



\section{Conclusion and Future Work}
\label{sec:conclusion}




Kaplan-Meier (KM) curves are valuable tools, especially in the medical domain, but achieving higher accuracy often requires larger datasets. Collaborative learning combined with differential privacy (DP) is a promising approach to balance privacy concerns while effectively utilizing diverse data sources for the calculation of the KM curve.   

In this work, we take a broad view on different representations of survival statistics and leverage these different functions to apply differential privacy in different stages of survival data processing. We also present a synthetic data generation technique that facilitates conversion between these different representations. This helps us to apply differential privacy in an effective and straightforward way to survival information with no need to have access to the dataset.

With this broader point of view on different representations, we are able to suggest multiple different routes that a system of collaborating clients can utilize to achieve a global private KM estimation. We show that our methods are robust against different distributions of data among dataholders and how this can motivate small as well as big data centers to join our private, collaborative scheme.

\subsection{On Censored Data and our DP Methods}
\label{sec:discussion-censoring}
Censoring occurs when a data point exits the study without experiencing the event of interest by the end of the observation period. Despite their incomplete status, these points are included in survival analysis in hopes of gleaning insights from the fact that they did not experience the event up to the point of censoring, particularly in the absence of extensive datasets.

However, their inclusion can introduce biases in survival curves due to assumptions about dropout reasons that may not always be accurate~\cite{ranganathan2012censoring, leung1997censoring, coemans2022bias, piovani2021pitfalls}. Moreover, uncertainty surrounds when or if these individuals will eventually experience the event, whether days or years later. For a visual comparison of the impact of censored points on KM curves across our datasets, see Appendix~\ref{sec:app-smooth}.

Our proposed DP algorithms, \dps and \dpy, offer collaborators flexibility by directly modifying $\s$ or $\yh$ functions, rather than manipulating raw number counts as in \dpmo. Many widely used survival analysis packages (e.g., lifelines\footnote{\url{https://lifelines.readthedocs.io/}}, pycox \footnote{\url{https://github.com/havakv/pycox}}, scikit-survival\footnote{\url{https://github.com/sebp/scikit-survival}}) require data in the form of $D=\{t^i, e^i\}_{i=1}^N$,  rather than counts. Therefore, using \dpmo requires frequent conversions between survival data and counts, complicating the process. Furthermore, it is customary for medical centers to only publish survival curves across the entire population as opposed to number counts at each distinct time. This means that the privacy provider might not even have access to the dataset at all. For these reasons, it is crucial to explore DP methods that can directly and efficiently handle KM and probability estimators.

Moreover, our \dps shows superior performance in the collaborative setting for uncensored data, outperforming both \dpy and \dpmo in all datasets and for all metrics (Table~\ref{tab:even-e1}). By offering these privacy-preserving solutions, we encourage data owners to learn a more explainable and bias-free estimator through collaboration and to solve the issue of limited data access. We showed that averaging private KM estimators works really well: for uneven data splits, the mean median is at most 10 units of time different from the non DP baseline and the mean survival percentages are at most $1\%$ out of the confidence interval of the non DP baseline (Table~\ref{tab:uneven-e1}). So, if the duration of these studies were later extended, the points that later experience the event of interest (censored in the first study) could be used to calculate a new KM estimator and then shared with the central server for an updated average. 

However, if including censored points is absolutely necessary, our improved and more private algorithm of \dpm (compared to the original method of \dpmo) can be utilized following our recommended paths and collaboration strategies, as demonstrated in Figure~\ref{fig:overall-graph}. 

Since our \dps-based methods perform really well for noncensoring datasets, we think that improving its sensitivity for when there are censored points, or using the general ideas from \dps to build a more stable and versatile private estimator is an important future research direction.

\bibliographystyle{ACM-Reference-Format}
\bibliography{references}


\begin{thebibliography}{67}


\ifx \showCODEN    \undefined \def \showCODEN     #1{\unskip}     \fi
\ifx \showDOI      \undefined \def \showDOI       #1{#1}\fi
\ifx \showISBNx    \undefined \def \showISBNx     #1{\unskip}     \fi
\ifx \showISBNxiii \undefined \def \showISBNxiii  #1{\unskip}     \fi
\ifx \showISSN     \undefined \def \showISSN      #1{\unskip}     \fi
\ifx \showLCCN     \undefined \def \showLCCN      #1{\unskip}     \fi
\ifx \shownote     \undefined \def \shownote      #1{#1}          \fi
\ifx \showarticletitle \undefined \def \showarticletitle #1{#1}   \fi
\ifx \showURL      \undefined \def \showURL       {\relax}        \fi
\providecommand\bibfield[2]{#2}
\providecommand\bibinfo[2]{#2}
\providecommand\natexlab[1]{#1}
\providecommand\showeprint[2][]{arXiv:#2}

\bibitem[GDP(2016)]%
        {GDPR}
 \bibinfo{year}{2016}\natexlab{}.
\newblock \bibinfo{title}{General Data Protection Regulation}.
\newblock \bibinfo{howpublished}{\url{https://eur-lex.europa.eu/eli/reg/2016/679/oj}}.
\newblock


\bibitem[SUP(2022)]%
        {SUP}
 \bibinfo{year}{2022}\natexlab{}.
\newblock \bibinfo{title}{Health data in the workplace}.
\newblock \bibinfo{howpublished}{\url{https://edps.europa.eu/data-protection/data-protection/reference-library/health-data-workplace_en}}.
\newblock


\bibitem[Abadi et~al\mbox{.}(2016)]%
        {abadi2016deep}
\bibfield{author}{\bibinfo{person}{Martin Abadi}, \bibinfo{person}{Andy Chu}, \bibinfo{person}{Ian Goodfellow}, \bibinfo{person}{H~Brendan McMahan}, \bibinfo{person}{Ilya Mironov}, \bibinfo{person}{Kunal Talwar}, {and} \bibinfo{person}{Li Zhang}.} \bibinfo{year}{2016}\natexlab{}.
\newblock \showarticletitle{Deep learning with differential privacy}. In \bibinfo{booktitle}{\emph{Proceedings of the 2016 ACM SIGSAC conference on computer and communications security}}. \bibinfo{pages}{308--318}.
\newblock


\bibitem[Backes et~al\mbox{.}(2016)]%
        {backes2016membership}
\bibfield{author}{\bibinfo{person}{Michael Backes}, \bibinfo{person}{Pascal Berrang}, \bibinfo{person}{Mathias Humbert}, {and} \bibinfo{person}{Praveen Manoharan}.} \bibinfo{year}{2016}\natexlab{}.
\newblock \showarticletitle{Membership privacy in MicroRNA-based studies}. In \bibinfo{booktitle}{\emph{Proceedings of the 2016 ACM SIGSAC Conference on Computer and Communications Security}}. \bibinfo{pages}{319--330}.
\newblock


\bibitem[Baesens et~al\mbox{.}(2005)]%
        {baesens2005neural}
\bibfield{author}{\bibinfo{person}{Bart Baesens}, \bibinfo{person}{Tony Van~Gestel}, \bibinfo{person}{Maria Stepanova}, \bibinfo{person}{Dirk Van~den Poel}, {and} \bibinfo{person}{Jan Vanthienen}.} \bibinfo{year}{2005}\natexlab{}.
\newblock \showarticletitle{Neural network survival analysis for personal loan data}.
\newblock \bibinfo{journal}{\emph{Journal of the Operational Research Society}} \bibinfo{volume}{56}, \bibinfo{number}{9} (\bibinfo{year}{2005}), \bibinfo{pages}{1089--1098}.
\newblock


\bibitem[Beverage(1976)]%
        {beverage1976privacy}
\bibfield{author}{\bibinfo{person}{James Beverage}.} \bibinfo{year}{1976}\natexlab{}.
\newblock \showarticletitle{The Privacy Act of 1974: an overview}.
\newblock \bibinfo{journal}{\emph{Duke law journal}} \bibinfo{volume}{1976}, \bibinfo{number}{2} (\bibinfo{year}{1976}), \bibinfo{pages}{301--329}.
\newblock


\bibitem[Bouliotis and Billingham(2011)]%
        {bouliotis2011crossing}
\bibfield{author}{\bibinfo{person}{George Bouliotis} {and} \bibinfo{person}{Lucinda Billingham}.} \bibinfo{year}{2011}\natexlab{}.
\newblock \showarticletitle{Crossing survival curves: alternatives to the log-rank test}.
\newblock \bibinfo{journal}{\emph{Trials}} \bibinfo{volume}{12}, \bibinfo{number}{1} (\bibinfo{year}{2011}), \bibinfo{pages}{1--1}.
\newblock


\bibitem[Brenner(2002)]%
        {brenner2002long}
\bibfield{author}{\bibinfo{person}{Hermann Brenner}.} \bibinfo{year}{2002}\natexlab{}.
\newblock \showarticletitle{Long-term survival rates of cancer patients achieved by the end of the 20th century: a period analysis}.
\newblock \bibinfo{journal}{\emph{The Lancet}} \bibinfo{volume}{360}, \bibinfo{number}{9340} (\bibinfo{year}{2002}), \bibinfo{pages}{1131--1135}.
\newblock


\bibitem[Castro(1995)]%
        {castro1995results}
\bibfield{author}{\bibinfo{person}{Joseph~R Castro}.} \bibinfo{year}{1995}\natexlab{}.
\newblock \showarticletitle{Results of heavy ion radiotherapy}.
\newblock \bibinfo{journal}{\emph{Radiation and environmental biophysics}} \bibinfo{volume}{34}, \bibinfo{number}{1} (\bibinfo{year}{1995}), \bibinfo{pages}{45--48}.
\newblock


\bibitem[Chakravarti(1989)]%
        {chakravarti1989isotonic}
\bibfield{author}{\bibinfo{person}{Nilotpal Chakravarti}.} \bibinfo{year}{1989}\natexlab{}.
\newblock \showarticletitle{Isotonic median regression: a linear programming approach}.
\newblock \bibinfo{journal}{\emph{Mathematics of operations research}} \bibinfo{volume}{14}, \bibinfo{number}{2} (\bibinfo{year}{1989}), \bibinfo{pages}{303--308}.
\newblock


\bibitem[Coemans et~al\mbox{.}(2022)]%
        {coemans2022bias}
\bibfield{author}{\bibinfo{person}{Maarten Coemans}, \bibinfo{person}{Geert Verbeke}, \bibinfo{person}{Bernd D{\"o}hler}, \bibinfo{person}{Caner S{\"u}sal}, {and} \bibinfo{person}{Maarten Naesens}.} \bibinfo{year}{2022}\natexlab{}.
\newblock \showarticletitle{Bias by censoring for competing events in survival analysis}.
\newblock \bibinfo{journal}{\emph{bmj}}  \bibinfo{volume}{378} (\bibinfo{year}{2022}).
\newblock


\bibitem[Curtis et~al\mbox{.}(2012)]%
        {curtis2012genomic}
\bibfield{author}{\bibinfo{person}{Christina Curtis}, \bibinfo{person}{Sohrab~P Shah}, \bibinfo{person}{Suet-Feung Chin}, \bibinfo{person}{Gulisa Turashvili}, \bibinfo{person}{Oscar~M Rueda}, \bibinfo{person}{Mark~J Dunning}, \bibinfo{person}{Doug Speed}, \bibinfo{person}{Andy~G Lynch}, \bibinfo{person}{Shamith Samarajiwa}, \bibinfo{person}{Yinyin Yuan}, {et~al\mbox{.}}} \bibinfo{year}{2012}\natexlab{}.
\newblock \showarticletitle{The genomic and transcriptomic architecture of 2,000 breast tumours reveals novel subgroups}.
\newblock \bibinfo{journal}{\emph{Nature}} \bibinfo{volume}{486}, \bibinfo{number}{7403} (\bibinfo{year}{2012}), \bibinfo{pages}{346--352}.
\newblock


\bibitem[Damien~Desfontaines(2021)]%
        {desfontainesblog20211001}
\bibfield{author}{\bibinfo{person}{Daniel Simmons-Marengo Damien~Desfontaines}.} \bibinfo{year}{2021}\natexlab{}.
\newblock \bibinfo{title}{A list of real-world uses of differential privacy}.
\newblock \bibinfo{howpublished}{\url{https://desfontain.es/privacy/real-world-differential-privacy.html}}.
\newblock
\newblock
\shownote{Ted is writing things (personal blog)}.


\bibitem[De~Leeuw(1977)]%
        {de1977correctness}
\bibfield{author}{\bibinfo{person}{Jan De~Leeuw}.} \bibinfo{year}{1977}\natexlab{}.
\newblock \showarticletitle{Correctness of Kruskal's algorithms for monotone regression with ties}.
\newblock \bibinfo{journal}{\emph{Psychometrika}}  \bibinfo{volume}{42} (\bibinfo{year}{1977}), \bibinfo{pages}{141--144}.
\newblock


\bibitem[Desfontaines and Simmons-Marengo(2021)]%
        {desfontainesblog20210616}
\bibfield{author}{\bibinfo{person}{Damien Desfontaines} {and} \bibinfo{person}{Daniel Simmons-Marengo}.} \bibinfo{year}{2021}\natexlab{}.
\newblock \bibinfo{title}{Getting more useful results with differential privacy}.
\newblock \bibinfo{howpublished}{\url{https://desfontain.es/privacy/more-useful-results-dp.html}}.
\newblock
\newblock
\shownote{Ted is writing things (personal blog)}.


\bibitem[Di~Leo and Sardanelli(2020)]%
        {di2020statistical}
\bibfield{author}{\bibinfo{person}{Giovanni Di~Leo} {and} \bibinfo{person}{Francesco Sardanelli}.} \bibinfo{year}{2020}\natexlab{}.
\newblock \showarticletitle{Statistical significance: p value, 0.05 threshold, and applications to radiomics—reasons for a conservative approach}.
\newblock \bibinfo{journal}{\emph{European radiology experimental}} \bibinfo{volume}{4}, \bibinfo{number}{1} (\bibinfo{year}{2020}), \bibinfo{pages}{1--8}.
\newblock


\bibitem[Dirick et~al\mbox{.}(2017)]%
        {dirick2017time}
\bibfield{author}{\bibinfo{person}{Lore Dirick}, \bibinfo{person}{Gerda Claeskens}, {and} \bibinfo{person}{Bart Baesens}.} \bibinfo{year}{2017}\natexlab{}.
\newblock \showarticletitle{Time to default in credit scoring using survival analysis: a benchmark study}.
\newblock \bibinfo{journal}{\emph{Journal of the Operational Research Society}} \bibinfo{volume}{68}, \bibinfo{number}{6} (\bibinfo{year}{2017}), \bibinfo{pages}{652--665}.
\newblock


\bibitem[Duan et~al\mbox{.}(2018)]%
        {duan2018odal}
\bibfield{author}{\bibinfo{person}{Rui Duan}, \bibinfo{person}{Mary~Regina Boland}, \bibinfo{person}{Jason~H Moore}, {and} \bibinfo{person}{Yong Chen}.} \bibinfo{year}{2018}\natexlab{}.
\newblock \showarticletitle{ODAL: A one-shot distributed algorithm to perform logistic regressions on electronic health records data from multiple clinical sites}. In \bibinfo{booktitle}{\emph{BIOCOMPUTING 2019: Proceedings of the Pacific Symposium}}. World Scientific, \bibinfo{pages}{30--41}.
\newblock


\bibitem[Dwork and Roth(2014)]%
        {Dwork2014book}
\bibfield{author}{\bibinfo{person}{Cynthia Dwork} {and} \bibinfo{person}{Aaron Roth}.} \bibinfo{year}{2014}\natexlab{}.
\newblock \showarticletitle{{The Algorithmic Foundations of Differential Privacy}}.
\newblock \bibinfo{journal}{\emph{Foundations and Trends in Theoretical Computer Science}} \bibinfo{volume}{9}, \bibinfo{number}{3--4} (\bibinfo{year}{2014}).
\newblock


\bibitem[Efron and RJ(1993)]%
        {efron1993introduction}
\bibfield{author}{\bibinfo{person}{B Efron} {and} \bibinfo{person}{Tibshirani RJ}.} \bibinfo{year}{1993}\natexlab{}.
\newblock \showarticletitle{An introduction to the bootstrap. Chapman and Hall, New York, NY}.
\newblock \bibinfo{journal}{\emph{Farrell, J., Johnston, M. and Twynam, D.(1998),‘‘Volunteer motivation, satisfaction, and management at an elite sporting competition’’, Journal of Sport Management}}  \bibinfo{volume}{12} (\bibinfo{year}{1993}), \bibinfo{pages}{288--300}.
\newblock


\bibitem[El~Emam et~al\mbox{.}(2011)]%
        {el2011systematic}
\bibfield{author}{\bibinfo{person}{Khaled El~Emam}, \bibinfo{person}{Elizabeth Jonker}, \bibinfo{person}{Luk Arbuckle}, {and} \bibinfo{person}{Bradley Malin}.} \bibinfo{year}{2011}\natexlab{}.
\newblock \showarticletitle{A systematic review of re-identification attacks on health data}.
\newblock \bibinfo{journal}{\emph{PloS one}} \bibinfo{volume}{6}, \bibinfo{number}{12} (\bibinfo{year}{2011}), \bibinfo{pages}{e28071}.
\newblock


\bibitem[Foekens et~al\mbox{.}(2000)]%
        {foekens2000urokinase}
\bibfield{author}{\bibinfo{person}{John~A Foekens}, \bibinfo{person}{Harry~A Peters}, \bibinfo{person}{Maxime~P Look}, \bibinfo{person}{Henk Portengen}, \bibinfo{person}{Manfred Schmitt}, \bibinfo{person}{Michael~D Kramer}, \bibinfo{person}{Nils Br{\"u}nner}, \bibinfo{person}{Fritz J{\"a}nicke}, \bibinfo{person}{Marion~E Meijer-van Gelder}, \bibinfo{person}{Sonja~C Henzen-Logmans}, {et~al\mbox{.}}} \bibinfo{year}{2000}\natexlab{}.
\newblock \showarticletitle{The urokinase system of plasminogen activation and prognosis in 2780 breast cancer patients}.
\newblock \bibinfo{journal}{\emph{Cancer research}} \bibinfo{volume}{60}, \bibinfo{number}{3} (\bibinfo{year}{2000}), \bibinfo{pages}{636--643}.
\newblock


\bibitem[Froelicher et~al\mbox{.}(2021)]%
        {froelicher2021truly}
\bibfield{author}{\bibinfo{person}{David Froelicher}, \bibinfo{person}{Juan~R Troncoso-Pastoriza}, \bibinfo{person}{Jean~Louis Raisaro}, \bibinfo{person}{Michel~A Cuendet}, \bibinfo{person}{Joao~Sa Sousa}, \bibinfo{person}{Hyunghoon Cho}, \bibinfo{person}{Bonnie Berger}, \bibinfo{person}{Jacques Fellay}, {and} \bibinfo{person}{Jean-Pierre Hubaux}.} \bibinfo{year}{2021}\natexlab{}.
\newblock \showarticletitle{Truly privacy-preserving federated analytics for precision medicine with multiparty homomorphic encryption}.
\newblock \bibinfo{journal}{\emph{Nature communications}} \bibinfo{volume}{12}, \bibinfo{number}{1} (\bibinfo{year}{2021}), \bibinfo{pages}{1--10}.
\newblock


\bibitem[Glinsky et~al\mbox{.}(2004)]%
        {glinsky2004gene}
\bibfield{author}{\bibinfo{person}{Gennadi~V Glinsky}, \bibinfo{person}{Anna~B Glinskii}, \bibinfo{person}{Andrew~J Stephenson}, \bibinfo{person}{Robert~M Hoffman}, \bibinfo{person}{William~L Gerald}, {et~al\mbox{.}}} \bibinfo{year}{2004}\natexlab{}.
\newblock \showarticletitle{Gene expression profiling predicts clinical outcome of prostate cancer}.
\newblock \bibinfo{journal}{\emph{The Journal of clinical investigation}} \bibinfo{volume}{113}, \bibinfo{number}{6} (\bibinfo{year}{2004}), \bibinfo{pages}{913--923}.
\newblock


\bibitem[Goel et~al\mbox{.}(2010)]%
        {goel2010understanding}
\bibfield{author}{\bibinfo{person}{Manish~Kumar Goel}, \bibinfo{person}{Pardeep Khanna}, {and} \bibinfo{person}{Jugal Kishore}.} \bibinfo{year}{2010}\natexlab{}.
\newblock \showarticletitle{Understanding survival analysis: Kaplan-Meier estimate}.
\newblock \bibinfo{journal}{\emph{International journal of Ayurveda research}} \bibinfo{volume}{1}, \bibinfo{number}{4} (\bibinfo{year}{2010}), \bibinfo{pages}{274}.
\newblock


\bibitem[Goldhirsch et~al\mbox{.}(1989)]%
        {goldhirsch1989costs}
\bibfield{author}{\bibinfo{person}{ARRPA Goldhirsch}, \bibinfo{person}{Richard~D Gelber}, \bibinfo{person}{R~John Simes}, \bibinfo{person}{Paul Glasziou}, {and} \bibinfo{person}{Alan~S Coates}.} \bibinfo{year}{1989}\natexlab{}.
\newblock \showarticletitle{Costs and benefits of adjuvant therapy in breast cancer: a quality-adjusted survival analysis.}
\newblock \bibinfo{journal}{\emph{Journal of Clinical Oncology}} \bibinfo{volume}{7}, \bibinfo{number}{1} (\bibinfo{year}{1989}), \bibinfo{pages}{36--44}.
\newblock


\bibitem[Gondara and Wang(2020)]%
        {gondara2020differentially}
\bibfield{author}{\bibinfo{person}{Lovedeep Gondara} {and} \bibinfo{person}{Ke Wang}.} \bibinfo{year}{2020}\natexlab{}.
\newblock \showarticletitle{Differentially Private Survival Function Estimation}. In \bibinfo{booktitle}{\emph{Machine Learning for Healthcare Conference}}. PMLR, \bibinfo{pages}{271--291}.
\newblock


\bibitem[Greenwood et~al\mbox{.}(1926)]%
        {greenwood1926report}
\bibfield{author}{\bibinfo{person}{Major Greenwood} {et~al\mbox{.}}} \bibinfo{year}{1926}\natexlab{}.
\newblock \showarticletitle{A report on the natural duration of cancer.}
\newblock \bibinfo{journal}{\emph{A Report on the Natural Duration of Cancer.}} \bibinfo{number}{33} (\bibinfo{year}{1926}).
\newblock


\bibitem[Guyot et~al\mbox{.}(2012)]%
        {guyot2012enhanced}
\bibfield{author}{\bibinfo{person}{Patricia Guyot}, \bibinfo{person}{AE Ades}, \bibinfo{person}{Mario~JNM Ouwens}, {and} \bibinfo{person}{Nicky~J Welton}.} \bibinfo{year}{2012}\natexlab{}.
\newblock \showarticletitle{Enhanced secondary analysis of survival data: reconstructing the data from published Kaplan-Meier survival curves}.
\newblock \bibinfo{journal}{\emph{BMC medical research methodology}} \bibinfo{volume}{12}, \bibinfo{number}{1} (\bibinfo{year}{2012}), \bibinfo{pages}{1--13}.
\newblock


\bibitem[Hay et~al\mbox{.}(2009)]%
        {hay2009boosting}
\bibfield{author}{\bibinfo{person}{Michael Hay}, \bibinfo{person}{Vibhor Rastogi}, \bibinfo{person}{Gerome Miklau}, {and} \bibinfo{person}{Dan Suciu}.} \bibinfo{year}{2009}\natexlab{}.
\newblock \showarticletitle{Boosting the accuracy of differentially-private histograms through consistency}.
\newblock \bibinfo{journal}{\emph{arXiv preprint arXiv:0904.0942}} (\bibinfo{year}{2009}).
\newblock


\bibitem[Helwig(2017)]%
        {bootstrap}
\bibfield{author}{\bibinfo{person}{Nathaniel~E. Helwig}.} \bibinfo{year}{2017}\natexlab{}.
\newblock \bibinfo{title}{Bootstrap Confidence Intervals}.
\newblock
\newblock
\urldef\tempurl%
\url{http://users.stat.umn.edu/~helwig/notes/bootci-Notes.pdf}
\showURL{%
\tempurl}


\bibitem[Hern{\'a}ndez and Weiss(1996)]%
        {hernandez1996first}
\bibfield{author}{\bibinfo{person}{Eugenio Hern{\'a}ndez} {and} \bibinfo{person}{Guido Weiss}.} \bibinfo{year}{1996}\natexlab{}.
\newblock \bibinfo{booktitle}{\emph{A first course on wavelets}}.
\newblock \bibinfo{publisher}{CRC press}.
\newblock


\bibitem[Hill and Miller(2011)]%
        {hill2011combine}
\bibfield{author}{\bibinfo{person}{Theodore~P Hill} {and} \bibinfo{person}{Jack Miller}.} \bibinfo{year}{2011}\natexlab{}.
\newblock \showarticletitle{How to combine independent data sets for the same quantity}.
\newblock \bibinfo{journal}{\emph{Chaos: An Interdisciplinary Journal of Nonlinear Science}} \bibinfo{volume}{21}, \bibinfo{number}{3} (\bibinfo{year}{2011}), \bibinfo{pages}{033102}.
\newblock


\bibitem[Hodge~Jr et~al\mbox{.}(1999)]%
        {hodge1999legal}
\bibfield{author}{\bibinfo{person}{James~G Hodge~Jr}, \bibinfo{person}{Lawrence~O Gostin}, {and} \bibinfo{person}{Peter~D Jacobson}.} \bibinfo{year}{1999}\natexlab{}.
\newblock \showarticletitle{Legal issues concerning electronic health information: privacy, quality, and liability}.
\newblock \bibinfo{journal}{\emph{Jama}} \bibinfo{volume}{282}, \bibinfo{number}{15} (\bibinfo{year}{1999}), \bibinfo{pages}{1466--1471}.
\newblock


\bibitem[Homer et~al\mbox{.}(2008)]%
        {homer2008resolving}
\bibfield{author}{\bibinfo{person}{Nils Homer}, \bibinfo{person}{Szabolcs Szelinger}, \bibinfo{person}{Margot Redman}, \bibinfo{person}{David Duggan}, \bibinfo{person}{Waibhav Tembe}, \bibinfo{person}{Jill Muehling}, \bibinfo{person}{John~V Pearson}, \bibinfo{person}{Dietrich~A Stephan}, \bibinfo{person}{Stanley~F Nelson}, {and} \bibinfo{person}{David~W Craig}.} \bibinfo{year}{2008}\natexlab{}.
\newblock \showarticletitle{Resolving individuals contributing trace amounts of DNA to highly complex mixtures using high-density SNP genotyping microarrays}.
\newblock \bibinfo{journal}{\emph{PLoS genetics}} \bibinfo{volume}{4}, \bibinfo{number}{8} (\bibinfo{year}{2008}), \bibinfo{pages}{e1000167}.
\newblock


\bibitem[Hsu et~al\mbox{.}(2014)]%
        {hsu2014differential}
\bibfield{author}{\bibinfo{person}{Justin Hsu}, \bibinfo{person}{Marco Gaboardi}, \bibinfo{person}{Andreas Haeberlen}, \bibinfo{person}{Sanjeev Khanna}, \bibinfo{person}{Arjun Narayan}, \bibinfo{person}{Benjamin~C Pierce}, {and} \bibinfo{person}{Aaron Roth}.} \bibinfo{year}{2014}\natexlab{}.
\newblock \showarticletitle{Differential privacy: An economic method for choosing epsilon}. In \bibinfo{booktitle}{\emph{2014 IEEE 27th Computer Security Foundations Symposium}}. IEEE, \bibinfo{pages}{398--410}.
\newblock


\bibitem[Ioannidis(2019)]%
        {10.1001/jama.2019.4582}
\bibfield{author}{\bibinfo{person}{John P.~A. Ioannidis}.} \bibinfo{year}{2019}\natexlab{}.
\newblock \showarticletitle{{The Importance of Predefined Rules and Prespecified Statistical Analyses: Do Not Abandon Significance}}.
\newblock \bibinfo{journal}{\emph{JAMA}} \bibinfo{volume}{321}, \bibinfo{number}{21} (\bibinfo{date}{06} \bibinfo{year}{2019}), \bibinfo{pages}{2067--2068}.
\newblock
\showISSN{0098-7484}
\urldef\tempurl%
\url{https://doi.org/10.1001/jama.2019.4582}
\showDOI{\tempurl}


\bibitem[Kaplan and Meier(1958)]%
        {kaplan1958nonparametric}
\bibfield{author}{\bibinfo{person}{Edward~L Kaplan} {and} \bibinfo{person}{Paul Meier}.} \bibinfo{year}{1958}\natexlab{}.
\newblock \showarticletitle{Nonparametric estimation from incomplete observations}.
\newblock \bibinfo{journal}{\emph{Journal of the American statistical association}} \bibinfo{volume}{53}, \bibinfo{number}{282} (\bibinfo{year}{1958}), \bibinfo{pages}{457--481}.
\newblock


\bibitem[Katzman et~al\mbox{.}(2018)]%
        {katzman2018deepsurv}
\bibfield{author}{\bibinfo{person}{Jared~L Katzman}, \bibinfo{person}{Uri Shaham}, \bibinfo{person}{Alexander Cloninger}, \bibinfo{person}{Jonathan Bates}, \bibinfo{person}{Tingting Jiang}, {and} \bibinfo{person}{Yuval Kluger}.} \bibinfo{year}{2018}\natexlab{}.
\newblock \showarticletitle{DeepSurv: personalized treatment recommender system using a Cox proportional hazards deep neural network}.
\newblock \bibinfo{journal}{\emph{BMC medical research methodology}} \bibinfo{volume}{18}, \bibinfo{number}{1} (\bibinfo{year}{2018}), \bibinfo{pages}{1--12}.
\newblock


\bibitem[Kerkouche et~al\mbox{.}(2021)]%
        {kerkouche2021compression}
\bibfield{author}{\bibinfo{person}{Raouf Kerkouche}, \bibinfo{person}{Gergely {\'A}cs}, \bibinfo{person}{Claude Castelluccia}, {and} \bibinfo{person}{Pierre Genev{\`e}s}.} \bibinfo{year}{2021}\natexlab{}.
\newblock \showarticletitle{Compression boosts differentially private federated learning}. In \bibinfo{booktitle}{\emph{2021 IEEE European Symposium on Security and Privacy (EuroS\&P)}}. IEEE, \bibinfo{pages}{304--318}.
\newblock


\bibitem[Kleinbaum et~al\mbox{.}(2012)]%
        {kleinbaum2012survival}
\bibfield{author}{\bibinfo{person}{David~G Kleinbaum}, \bibinfo{person}{Mitchel Klein}, {et~al\mbox{.}}} \bibinfo{year}{2012}\natexlab{}.
\newblock \bibinfo{booktitle}{\emph{Survival analysis: a self-learning text}}. Vol.~\bibinfo{volume}{3}.
\newblock \bibinfo{publisher}{Springer}.
\newblock


\bibitem[Knaus et~al\mbox{.}(1995)]%
        {knaus1995support}
\bibfield{author}{\bibinfo{person}{William~A Knaus}, \bibinfo{person}{Frank~E Harrell}, \bibinfo{person}{Joanne Lynn}, \bibinfo{person}{Lee Goldman}, \bibinfo{person}{Russell~S Phillips}, \bibinfo{person}{Alfred~F Connors}, \bibinfo{person}{Neal~V Dawson}, \bibinfo{person}{William~J Fulkerson}, \bibinfo{person}{Robert~M Califf}, \bibinfo{person}{Norman Desbiens}, {et~al\mbox{.}}} \bibinfo{year}{1995}\natexlab{}.
\newblock \showarticletitle{The SUPPORT prognostic model: Objective estimates of survival for seriously ill hospitalized adults}.
\newblock \bibinfo{journal}{\emph{Annals of internal medicine}} \bibinfo{volume}{122}, \bibinfo{number}{3} (\bibinfo{year}{1995}), \bibinfo{pages}{191--203}.
\newblock


\bibitem[Kvamme and Borgan(2021)]%
        {kvamme2021continuous}
\bibfield{author}{\bibinfo{person}{H{\aa}vard Kvamme} {and} \bibinfo{person}{{\O}rnulf Borgan}.} \bibinfo{year}{2021}\natexlab{}.
\newblock \showarticletitle{Continuous and discrete-time survival prediction with neural networks}.
\newblock \bibinfo{journal}{\emph{Lifetime Data Analysis}} \bibinfo{volume}{27}, \bibinfo{number}{4} (\bibinfo{year}{2021}), \bibinfo{pages}{710--736}.
\newblock


\bibitem[Kvamme et~al\mbox{.}(2019)]%
        {kvamme2019time}
\bibfield{author}{\bibinfo{person}{H{\aa}vard Kvamme}, \bibinfo{person}{{\O}rnulf Borgan}, {and} \bibinfo{person}{Ida Scheel}.} \bibinfo{year}{2019}\natexlab{}.
\newblock \showarticletitle{Time-to-event prediction with neural networks and Cox regression}.
\newblock \bibinfo{journal}{\emph{arXiv preprint arXiv:1907.00825}} (\bibinfo{year}{2019}).
\newblock


\bibitem[Lee et~al\mbox{.}(2018)]%
        {lee2018deephit}
\bibfield{author}{\bibinfo{person}{Changhee Lee}, \bibinfo{person}{William Zame}, \bibinfo{person}{Jinsung Yoon}, {and} \bibinfo{person}{Mihaela van~der Schaar}.} \bibinfo{year}{2018}\natexlab{}.
\newblock \showarticletitle{Deephit: A deep learning approach to survival analysis with competing risks}. In \bibinfo{booktitle}{\emph{Proceedings of the AAAI Conference on Artificial Intelligence}}, Vol.~\bibinfo{volume}{32}.
\newblock


\bibitem[Leung et~al\mbox{.}(1997)]%
        {leung1997censoring}
\bibfield{author}{\bibinfo{person}{Kwan-Moon Leung}, \bibinfo{person}{Robert~M Elashoff}, {and} \bibinfo{person}{Abdelmonem~A Afifi}.} \bibinfo{year}{1997}\natexlab{}.
\newblock \showarticletitle{Censoring issues in survival analysis}.
\newblock \bibinfo{journal}{\emph{Annual review of public health}} \bibinfo{volume}{18}, \bibinfo{number}{1} (\bibinfo{year}{1997}), \bibinfo{pages}{83--104}.
\newblock


\bibitem[Lu(2002)]%
        {lu2002predicting}
\bibfield{author}{\bibinfo{person}{Junxiang Lu}.} \bibinfo{year}{2002}\natexlab{}.
\newblock \showarticletitle{Predicting customer churn in the telecommunications industry----An application of survival analysis modeling using SAS}. In \bibinfo{booktitle}{\emph{SAS User Group International (SUGI27) Online Proceedings}}, Vol.~\bibinfo{volume}{114}.
\newblock


\bibitem[Makhoul(1980)]%
        {makhoul1980fast}
\bibfield{author}{\bibinfo{person}{John Makhoul}.} \bibinfo{year}{1980}\natexlab{}.
\newblock \showarticletitle{A fast cosine transform in one and two dimensions}.
\newblock \bibinfo{journal}{\emph{IEEE Transactions on Acoustics, Speech, and Signal Processing}} \bibinfo{volume}{28}, \bibinfo{number}{1} (\bibinfo{year}{1980}), \bibinfo{pages}{27--34}.
\newblock


\bibitem[Mantel et~al\mbox{.}(1966)]%
        {mantel1966evaluation}
\bibfield{author}{\bibinfo{person}{Nathan Mantel} {et~al\mbox{.}}} \bibinfo{year}{1966}\natexlab{}.
\newblock \showarticletitle{Evaluation of survival data and two new rank order statistics arising in its consideration}.
\newblock \bibinfo{journal}{\emph{Cancer Chemother Rep}} \bibinfo{volume}{50}, \bibinfo{number}{3} (\bibinfo{year}{1966}), \bibinfo{pages}{163--170}.
\newblock


\bibitem[Mih{\'a}ly et~al\mbox{.}(2013)]%
        {mihaly2013meta}
\bibfield{author}{\bibinfo{person}{Zsuzsanna Mih{\'a}ly}, \bibinfo{person}{M{\'a}t{\'e} Kormos}, \bibinfo{person}{Andr{\'a}s L{\'a}nczky}, \bibinfo{person}{Magdolna Dank}, \bibinfo{person}{Jan Budczies}, \bibinfo{person}{Marcell~A Sz{\'a}sz}, {and} \bibinfo{person}{Bal{\'a}zs Gy{\H{o}}rffy}.} \bibinfo{year}{2013}\natexlab{}.
\newblock \showarticletitle{A meta-analysis of gene expression-based biomarkers predicting outcome after tamoxifen treatment in breast cancer}.
\newblock \bibinfo{journal}{\emph{Breast cancer research and treatment}} \bibinfo{volume}{140}, \bibinfo{number}{2} (\bibinfo{year}{2013}), \bibinfo{pages}{219--232}.
\newblock


\bibitem[Nissim et~al\mbox{.}(2007)]%
        {nissim2007smooth}
\bibfield{author}{\bibinfo{person}{Kobbi Nissim}, \bibinfo{person}{Sofya Raskhodnikova}, {and} \bibinfo{person}{Adam Smith}.} \bibinfo{year}{2007}\natexlab{}.
\newblock \showarticletitle{Smooth sensitivity and sampling in private data analysis}. In \bibinfo{booktitle}{\emph{Proceedings of the thirty-ninth annual ACM symposium on Theory of computing}}. \bibinfo{pages}{75--84}.
\newblock


\bibitem[Piovani et~al\mbox{.}(2021)]%
        {piovani2021pitfalls}
\bibfield{author}{\bibinfo{person}{Daniele Piovani}, \bibinfo{person}{Georgios K~Nikolopoulos}, {and} \bibinfo{person}{Stefanos Bonovas}.} \bibinfo{year}{2021}\natexlab{}.
\newblock \showarticletitle{Pitfalls and perils of survival analysis under incorrect assumptions: the case of COVID-19 data}.
\newblock \bibinfo{journal}{\emph{Biomedica}}  \bibinfo{volume}{41} (\bibinfo{year}{2021}), \bibinfo{pages}{21--28}.
\newblock


\bibitem[Politou et~al\mbox{.}(2018)]%
        {politou2018forgetting}
\bibfield{author}{\bibinfo{person}{Eugenia Politou}, \bibinfo{person}{Efthimios Alepis}, {and} \bibinfo{person}{Constantinos Patsakis}.} \bibinfo{year}{2018}\natexlab{}.
\newblock \showarticletitle{Forgetting personal data and revoking consent under the GDPR: Challenges and proposed solutions}.
\newblock \bibinfo{journal}{\emph{Journal of cybersecurity}} \bibinfo{volume}{4}, \bibinfo{number}{1} (\bibinfo{year}{2018}), \bibinfo{pages}{tyy001}.
\newblock


\bibitem[Ponomareva et~al\mbox{.}(2023)]%
        {ponomareva2023dp}
\bibfield{author}{\bibinfo{person}{Natalia Ponomareva}, \bibinfo{person}{Hussein Hazimeh}, \bibinfo{person}{Alex Kurakin}, \bibinfo{person}{Zheng Xu}, \bibinfo{person}{Carson Denison}, \bibinfo{person}{H~Brendan McMahan}, \bibinfo{person}{Sergei Vassilvitskii}, \bibinfo{person}{Steve Chien}, {and} \bibinfo{person}{Abhradeep Thakurta}.} \bibinfo{year}{2023}\natexlab{}.
\newblock \showarticletitle{How to dp-fy ml: A practical guide to machine learning with differential privacy}.
\newblock \bibinfo{journal}{\emph{arXiv preprint arXiv:2303.00654}} (\bibinfo{year}{2023}).
\newblock


\bibitem[Ranganathan et~al\mbox{.}(2012)]%
        {ranganathan2012censoring}
\bibfield{author}{\bibinfo{person}{Priya Ranganathan}, \bibinfo{person}{CS Pramesh}, {et~al\mbox{.}}} \bibinfo{year}{2012}\natexlab{}.
\newblock \showarticletitle{Censoring in survival analysis: potential for bias}.
\newblock \bibinfo{journal}{\emph{Perspectives in clinical research}} \bibinfo{volume}{3}, \bibinfo{number}{1} (\bibinfo{year}{2012}), \bibinfo{pages}{40}.
\newblock


\bibitem[Rastogi and Nath(2010)]%
        {rastogi2010differentially}
\bibfield{author}{\bibinfo{person}{Vibhor Rastogi} {and} \bibinfo{person}{Suman Nath}.} \bibinfo{year}{2010}\natexlab{}.
\newblock \showarticletitle{Differentially private aggregation of distributed time-series with transformation and encryption}. In \bibinfo{booktitle}{\emph{Proceedings of the 2010 ACM SIGMOD International Conference on Management of data}}. \bibinfo{pages}{735--746}.
\newblock


\bibitem[Rocher et~al\mbox{.}(2019)]%
        {rocher2019estimating}
\bibfield{author}{\bibinfo{person}{Luc Rocher}, \bibinfo{person}{Julien~M Hendrickx}, {and} \bibinfo{person}{Yves-Alexandre De~Montjoye}.} \bibinfo{year}{2019}\natexlab{}.
\newblock \showarticletitle{Estimating the success of re-identifications in incomplete datasets using generative models}.
\newblock \bibinfo{journal}{\emph{Nature communications}} \bibinfo{volume}{10}, \bibinfo{number}{1} (\bibinfo{year}{2019}), \bibinfo{pages}{1--9}.
\newblock


\bibitem[Rossi et~al\mbox{.}(2012)]%
        {rossi2012carboplatin}
\bibfield{author}{\bibinfo{person}{Antonio Rossi}, \bibinfo{person}{Massimo Di~Maio}, \bibinfo{person}{Paolo Chiodini}, \bibinfo{person}{Robin~Michael Rudd}, \bibinfo{person}{Hiroaki Okamoto}, \bibinfo{person}{Dimosthenis~Vasilios Skarlos}, \bibinfo{person}{M Fruh}, \bibinfo{person}{Wendi Qian}, \bibinfo{person}{Tomohide Tamura}, \bibinfo{person}{Epaminondas Samantas}, {et~al\mbox{.}}} \bibinfo{year}{2012}\natexlab{}.
\newblock \showarticletitle{Carboplatin-or cisplatin-based chemotherapy in first-line treatment of small-cell lung cancer: the COCIS meta-analysis of individual patient data}.
\newblock \bibinfo{journal}{\emph{Database of Abstracts of Reviews of Effects (DARE): Quality-assessed Reviews [Internet]}} (\bibinfo{year}{2012}).
\newblock


\bibitem[Roth(2022)]%
        {roth2022private}
\bibfield{author}{\bibinfo{person}{Edo Roth}.} \bibinfo{year}{2022}\natexlab{}.
\newblock \showarticletitle{Private Federated Analytics At Scale}.
\newblock  (\bibinfo{year}{2022}).
\newblock


\bibitem[Sawyer(2003)]%
        {sawyer2003greenwood}
\bibfield{author}{\bibinfo{person}{S Sawyer}.} \bibinfo{year}{2003}\natexlab{}.
\newblock \showarticletitle{The greenwood and exponential greenwood confidence intervals in survival analysis}.
\newblock \bibinfo{journal}{\emph{Applied survival analysis: regression modeling of time to event data}} (\bibinfo{year}{2003}), \bibinfo{pages}{1--14}.
\newblock


\bibitem[Schumacher et~al\mbox{.}(1994)]%
        {schumacher1994randomized}
\bibfield{author}{\bibinfo{person}{M Schumacher}, \bibinfo{person}{G Bastert}, \bibinfo{person}{H Bojar}, \bibinfo{person}{K H{\"u}bner}, \bibinfo{person}{M Olschewski}, \bibinfo{person}{W Sauerbrei}, \bibinfo{person}{C Schmoor}, \bibinfo{person}{C Beyerle}, \bibinfo{person}{RL Neumann}, {and} \bibinfo{person}{HF Rauschecker}.} \bibinfo{year}{1994}\natexlab{}.
\newblock \showarticletitle{Randomized 2 x 2 trial evaluating hormonal treatment and the duration of chemotherapy in node-positive breast cancer patients. German Breast Cancer Study Group.}
\newblock \bibinfo{journal}{\emph{Journal of Clinical Oncology}} \bibinfo{volume}{12}, \bibinfo{number}{10} (\bibinfo{year}{1994}), \bibinfo{pages}{2086--2093}.
\newblock


\bibitem[Strang(1999)]%
        {strang1999discrete}
\bibfield{author}{\bibinfo{person}{Gilbert Strang}.} \bibinfo{year}{1999}\natexlab{}.
\newblock \showarticletitle{The discrete cosine transform}.
\newblock \bibinfo{journal}{\emph{SIAM review}} \bibinfo{volume}{41}, \bibinfo{number}{1} (\bibinfo{year}{1999}), \bibinfo{pages}{135--147}.
\newblock


\bibitem[Vogelsang et~al\mbox{.}(2020)]%
        {vogelsang2020secure}
\bibfield{author}{\bibinfo{person}{Lennart Vogelsang}, \bibinfo{person}{Moritz Lehne}, \bibinfo{person}{Phillipp Schoppmann}, \bibinfo{person}{Fabian Prasser}, \bibinfo{person}{Sylvia Thun}, \bibinfo{person}{Bj{\"o}rn Scheuermann}, {and} \bibinfo{person}{Josef Schepers}.} \bibinfo{year}{2020}\natexlab{}.
\newblock \showarticletitle{A Secure Multi-Party Computation Protocol for Time-To-Event Analyses}.
\newblock In \bibinfo{booktitle}{\emph{Digital Personalized Health and Medicine}}. \bibinfo{publisher}{IOS Press}, \bibinfo{pages}{8--12}.
\newblock


\bibitem[Voigt and Von~dem Bussche(2017)]%
        {voigt2017eu}
\bibfield{author}{\bibinfo{person}{Paul Voigt} {and} \bibinfo{person}{Axel Von~dem Bussche}.} \bibinfo{year}{2017}\natexlab{}.
\newblock \showarticletitle{The eu general data protection regulation (gdpr)}.
\newblock \bibinfo{journal}{\emph{A Practical Guide, 1st Ed., Cham: Springer International Publishing}} \bibinfo{volume}{10}, \bibinfo{number}{3152676} (\bibinfo{year}{2017}), \bibinfo{pages}{10--5555}.
\newblock


\bibitem[von Maltitz et~al\mbox{.}(2021)]%
        {von2021privacy}
\bibfield{author}{\bibinfo{person}{Marcel von Maltitz}, \bibinfo{person}{Hendrik Ballhausen}, \bibinfo{person}{David Kaul}, \bibinfo{person}{Daniel~F Fleischmann}, \bibinfo{person}{Maximilian Niyazi}, \bibinfo{person}{Claus Belka}, {and} \bibinfo{person}{Georg Carle}.} \bibinfo{year}{2021}\natexlab{}.
\newblock \showarticletitle{A Privacy-Preserving Log-Rank Test for the Kaplan-Meier Estimator With Secure Multiparty Computation: Algorithm Development and Validation}.
\newblock \bibinfo{journal}{\emph{JMIR medical informatics}} \bibinfo{volume}{9}, \bibinfo{number}{1} (\bibinfo{year}{2021}), \bibinfo{pages}{e22158}.
\newblock


\bibitem[Wei and Royston(2017)]%
        {wei2017reconstructing}
\bibfield{author}{\bibinfo{person}{Yinghui Wei} {and} \bibinfo{person}{Patrick Royston}.} \bibinfo{year}{2017}\natexlab{}.
\newblock \showarticletitle{Reconstructing time-to-event data from published Kaplan--Meier curves}.
\newblock \bibinfo{journal}{\emph{The Stata Journal}} \bibinfo{volume}{17}, \bibinfo{number}{4} (\bibinfo{year}{2017}), \bibinfo{pages}{786--802}.
\newblock


\bibitem[Xi et~al\mbox{.}(2022)]%
        {xi2022review}
\bibfield{author}{\bibinfo{person}{Peng Xi}, \bibinfo{person}{Xinglong Zhang}, \bibinfo{person}{Lian Wang}, \bibinfo{person}{Wenjuan Liu}, {and} \bibinfo{person}{Shaoliang Peng}.} \bibinfo{year}{2022}\natexlab{}.
\newblock \showarticletitle{A review of Blockchain-based secure sharing of healthcare data}.
\newblock \bibinfo{journal}{\emph{Applied Sciences}} \bibinfo{volume}{12}, \bibinfo{number}{15} (\bibinfo{year}{2022}), \bibinfo{pages}{7912}.
\newblock


\end{thebibliography}

%
\appendix
\section{Supplementary Information Appendix}

\subsection{Discrete Cosine Transform (DCT)}
\label{sec:appendix-dct}
Discrete cosine transform (DCT) is a change of basis for a finite-dimensional signal of the form $\mathbf{X} = [x_0, ...,x_{N-1}]$ and can be described by a linear and invertible function $f: \mathbb{R}^N \rightarrow \mathbb{R}^N$. The new "bases" for this transformation are in the form of cosine functions with different oscillating frequencies. So, by DCT, we transform a vector of $\mathbf{X} = [x_0, ...,x_{N-1}]$ to another vector $\mathbf{Y} = [y_0, ...,y_{N-1}]$ with the same number of components. Formally, DCT is defined as:
\begin{alignat}{3}
&&& y_k = c_{k, N}\sum_{n=0}^{N-1} x_n \cos{(\frac{k\pi(2n+1)}{2N})}\label{eq:dct-def}
\end{alignat}
with $c_{0, N} = \sqrt{\frac{1}{N}}$ and $c_{k,N}=\sqrt{\frac{2}{N}} \ \forall k>0$. The new elements $y_k$ are the projection of the original vector $\mathbf{X}$ onto the cosine bases, so these can be seen as the "coefficients" of the vector $\mathbf{X}$ in the space spanned by the new bases. It can be proved that these bases are orthonormal~\cite{strang1999discrete, hernandez1996first}, that is, their $l_2$ norm is 1 and that they are all mutually orthogonal. A general property of changing bases with orthonormal transformations is that the $l_2$ norm is invariant under these transformations. So, for a properly scaled DCT transformation, the $l_2$ norm is preserved. The inverse cosine transform can be obtained through:
\begin{alignat}{3}
\label{eq:idct-def}
    &&& x_n = c_{k,N}\sum_{k=0}^{N-1} y_k \cos(\frac{k\pi(2n+1)}{2N})
\end{alignat}
This is the original vector $\mathbf{X}$, now decomposed onto the new orthonormal cosine basis. 
\subsection{Sensitivity of $\s(t)$}
\label{sec:appendix-proof-s}

To apply differential privacy on the Kaplan-Meier estimator $\s(T)$, we consider the notion of neighboring datasets where one dataset is obtained by \textit{changing} one data point in the other dataset (bounded differential privacy). So, both our neighboring datasets have the same number of total data points $N$. Here, we measure the value of the KM estimator in equidistant time intervals $t_i \in \{t_0 = 0, t_1 = b, ..., T_{\max} \}$ with a fixed distance of $b$ to simultaneously guarantee the privacy of times of events. In this case, the most significant effect that one data point might have on $\s(t)$ is obtained by changing the event time for a point that experiences the event $e=1$ at time $T_{\max}$ to $t_1$. We assume that this is true since, as described, the effect of censored data on the calculated value of the Kaplan-Meier curve is minimal (only appearing in the risk set of Equation~\ref{eq:kmestimator}) compared to points that experience the event of interest. And also since the effect of an event happening at time $t$ demonstrates itself in all the calculated KM values of the following time-steps $\geq t$ (see Equation~\ref{eq:kmestimator}). Note that neighboring datasets will have the same number of data points $N$ and also the same number of events $\sum_{i=1}^{\tm} d_i$ and the same number of censored points $\sum_{i=1}^{\tm}c_i$.

\myparagraph{Sensitivity with no censoring in dataset} To develop an initial intuition, we start with the case where no censoring data is present in the dataset, as this will be the easiest to bound for sensitivity. Assuming that $\s(t)$ is measured on a dataset with a total of $N$ datapoints with \textbf{no} censored data, and $S'(t)$ is measured for the neighboring dataset obtained by moving the time of the event of a point to $t_1$, we have:

\begin{alignat}{3}
\s_1& =&& \frac{N - d_1}{N},\qquad \s'_1 = \frac{N-d_1-1}{N}\nonumber\\ &\Rightarrow&& \s_1 - \s'_1 = \frac{1}{N}\\
\s_2& =&& \s_1 . \frac{r_2 - d_2}{r_2} = \frac{N- d_1}{N} . \frac{N - d_1 - d_2}{N - d_1} \nonumber  \\
&=&&\frac{N-d_1 - d_2}{N}\\
\s'_2& =&& \s'_1 .\frac{r'_2 - d_2}{r'_2} = \frac{N - d_1 -1 }{N} . \frac{N - d_1 -1 - d_2}{N - d_1 -1}\nonumber \\
&=&& \frac{N - d_1 -1 - d_2}{N}\\
&\Rightarrow&& \s_2 - \s'_2 = \frac{1}{N}
\end{alignat}
Where $\s_i, d_i$ and $r_i$ denote the KM estimator, number of events and risk set measured at time $t_i$, respectively. The third line is derived from the definition of $\displaystyle{\s(t) = \s(t-1)\times \frac{r_t - d_t}{r_t}}$. Now we hypothesize that for the $k$-th term of the survival function we have: 
\begin{alignat}{3}
\s_k& =&& \frac{N - d_1 - d_2 - ... - d_k}{N}\\
\s'_k& =&& \frac{N - d_1 -1 - d_2 - ... - d_k}{N}
\end{alignat}
and show that for the $k+1$-th term we have:
\begin{alignat}{3}
\s_{k+1}& =&& \s_k \times \frac{r_{k+1} - d_{k+1}}{r_{k+1}} \nonumber\\
&=&& \frac{N - d_1 - ... - d_k}{N}. \frac{N-d_1 -... - d_{k+1}}{N- d_1 - ... - d_k}\nonumber\\
&=&&\frac{N -d_1 - ... - d_{k+1}}{N}\label{eq:app-st}\\
\s'_{k+1} &=&& \s'_k\times \frac{r'_{k+1} - d_{k+1}}{r'_{k+1}} \nonumber\\
&=&& \frac{N-d_1 - 1 -... - d_k}{N}. \frac{N-d_1-1 - ... - d_{k+1}}{N-d_1 -1 - ...- d_k}\nonumber\\
&=&& \frac{N - d_1 - 1- ... - d_{k+1}}{N}\label{eq:app-stp}
\end{alignat}
Therefore, we prove by induction that our hypothesis is correct. Now we can calculate the difference of any term $k$ between $S$ and $S'$:
\begin{alignat}{3}
\s_k - \s'_k &=&& \frac{N - d_1 - ... - d_k}{N} - \frac{N - d_1 - 1 - ...- d_k}{N} \nonumber\\
&=&& \frac{1}{N}
\end{alignat}
For the last term, there is also a difference between $\s_{\tm}$ and $\s'_{\tm}$, since now for $\s'$ a datapoint experiencing the event is missing:
\begin{alignat}{3}
&&&\s_{\tm} - \s'_{\tm} = \frac{N-d_1-...d_{\tm}}{N} -\nonumber\\ &&&\frac{N-d_1-1-...-(d_{\tm}-1)}{N} = \frac{0}{N}\nonumber
\end{alignat}
So over the time-span of the study the total difference of $\s$ and $\s'$ would be:
\begin{alignat}{3}
&&&\Delta_1\s_{\text{no censor}} = \|\s - \s'\|_1 = \frac{T-1}{N} \label{eq:app-sl1-nocensor-proof}\\
&&&\Delta_2\s_{\text{no censor}} =\|\s - \s'\|_2 = \frac{\sqrt{T-1}}{N}\label{eq:app-sl2-nocensor-proof}
\end{alignat}
Where $T$ is the number of time bins, i.e. $T = T_{\max}/b$. 

\myparagraph{Sensitivity with censoring in dataset} Now we move to the more general case of datasets with censored points. We first calculate the sensitivity for when neighboring datasets are obtained by changing a point $x = \{t=\tm, e=1\}$ in $D$ to $x' =\{t=1, e=1\}$ in $D'$. Later, we also calculate the sensitivities for the cases of neighboring datasets being obtained by changing $x=\{t=1, e=1\}$ to $x'=\{t=1, e=0\}$ as well as $x=\{t=\tm, e=1\}$ to $x'=\{t=1, e=0\}$ and show that our assumed neighboring case of changing $x = \{t=\tm, e=1\}$ to $x' =\{t=1, e=1\}$ results in the largest sensitivity. 

\begin{alignat}{3}
\s_t=& \frac{N - d_1}{N} \frac{N-d_1-c_1-d_2}{N - d_1 - c_1}\times ...&& \nonumber\\
&\times\frac{N - d_1 -...d_{t-1} - c_1 -....c_{t-1} - d_t}{N-d_1 -...d_{t-1} -c_1 ...-c_{t-1}}&&\nonumber\\
&\label{eq:def-st}&&\\
\s'_t=& \frac{N - d_1 -1}{N} \frac{N-1-d_1-c_1-d_2}{N -1- d_1 - c_1}\times...&& \nonumber\\
&\times\frac{N - 1 -d_1 -...d_{t-1} - c_1 -....c_{t-1} - d_t}{N-1-d_1 -...d_{t-1} -c_1 ...-c_{t-1}}&&\nonumber\\
&\label{eq:def-sprimet}&&
\end{alignat}

\begin{lemma} $\forall A, B, c\in \mathbb{N}$ if ($B, B-c>0 \land A, c\geq 0 \land B\geq A$) :
\label{eq:app-upper-bound}
\begin{alignat}{3}
&\frac{A - c}{B-c}\leq \frac{A}{B}&&\nonumber
\end{alignat}
\end{lemma}
\begin{proof}
$AB - Bc \leq AB - Ac \Leftrightarrow Bc \geq Ac \Leftrightarrow B\geq A$\nonumber
\end{proof}
\begin{lemma}
\label{eq:app-lower-bound}
$\forall A, B, c\in \mathbb{N}$ if $B, B-c>0 \land A, c\geq 0 $:
\begin{alignat}{3}
\frac{A}{B} \leq \frac{A}{B-c}\nonumber
\end{alignat}
\end{lemma}
\begin{proof}
$AB - Ac \leq AB \Leftrightarrow Ac\geq 0 $\nonumber
\end{proof}

So we can upper and lower bound each individual $\s_t$ and $\s'_t$ by either adding $\sum_{i=1}^{t'-1} c_i$ to both numerator and denominator of each fraction which represents the new term at time $t'$ (for the upper bound) or adding $c_{t'-1}$ to only the denominator of each fraction at time $t'$ in the Equation \ref{eq:def-st} and \ref{eq:def-sprimet} (for lower bound), that is:
\begin{alignat}{3}
&&&\frac{N-d_1}{N}\frac{N-d_1-c_1-d_2}{N-d_1-c_1+c_1}\times...\nonumber\\
&&&\frac{N-d_1-...d_{t-1}-c_1-...c_{t-1}-d_t}{N-d_1-...d_{t-1}-c_1-...c_{t-1}+c_{t-1}}\leq \s_t\leq\nonumber\\
&&&\frac{N-d_1}{N}\frac{N-d_1-c_1+c_1-d_2}{N-d_1-c_1+c_1}\times...\nonumber\\
&&&\frac{N-d_1-...d_{t-1}-c_1-...c_{t-1}+c_1+...+c_{t-1}-d_t}{N-d_1-...d_{t-1}-c_1-...c_{t-1}+c_1+...c_{t-1}}\nonumber
\end{alignat}
and we can also bound $\s'_t$ in the same way. So finally, after cancelling out the consecutive numerators and denominators we are left with: 
\begin{alignat}{3}
&&&\frac{N - d_1...-d_t - c_1 ... -c_{t-1}}{N} \leq \s_t\leq \frac{N-d_1...-d_t}{N}
\label{eq:app-bounds-st-case1}
\end{alignat}
\begin{alignat}{3}
&&&\frac{N - 1 - d_1 ...- d_t - c_1 ... -c_{t-1}}{N} \leq \s'_t \leq \frac{N-1-d_1...-d_t}{N}
\label{eq:app-bounds-spt-case1}
\end{alignat}
We know that $\s_t \geq \s'_t$ for all times, because $\s'$ has experienced one extra event $e=1$ in the first time step, so the value of $\s_t - \s'_t \geq 0$ at all times. For an upper bound we can subtract the lower bound of $\s'_t$ from the upper bound of $\s_t$:
\begin{alignat}{3}
&&\s_t - \s'_t \leq& \frac{N-d_1...-d_{t}}{N} - \frac{N-1-d_1 ...-d_t - c_1...-c_{t-1}}{N}\nonumber\\
&&\Rightarrow&\s_t-\s'_t = \begin{cases}
\frac{1}{N} & t=1\\
 \frac{1 + c_1 ...+c_{t-1}}{N} & 1 < t < \tm \\
 \frac{c_1+ ...+c_{t-1}}{N}& t=\tm
\end{cases}
\end{alignat}
if we define $C = \sum^{T_{\max}-1}_{i=1}c_i$, we can find a general upper bound over the whole time-span of the study:
{\small
\begin{alignat}{3}
&&&\Delta_1\s_{\text{w censor}}=\|\s - \s'\|_1=\frac{1}{N}+\frac{1+c_1}{N}+...\frac{c_1+...c_{\tm -1}}{N} \nonumber\\
&&&\leq \frac{TC}{N}\label{eq:app-sl1-wcensor} \\
&&&\Delta_2\s_{\text{w censor}}=\|\s-\s'\|_2\nonumber\\
&&&=\sqrt{(\frac{1}{N})\strut^2+(\frac{1+c_1}{N})\strut^2+...(\frac{c_1+...c_{\tm-1}}{N})\strut^2}\nonumber\\
&&&\leq\sqrt{(\frac{C}{N})\strut^2+...+(\frac{C}{N})\strut^2}=\frac{\sqrt{T}C}{N}\label{eq:app-sl2-wcensor}
\end{alignat}
}%
Where $T$ is the number of time-bins, i.e. $T = T_{\max}/b$. Note that this is in line with our proof for the case of no censoring according to Equations~\ref{eq:app-sl1-nocensor-proof} and~\ref{eq:app-sl2-nocensor-proof}. 


Now, let us look at other possible neighboring datasets obtained by changing one point. We calculate the sensitivity for $x=\{t=1, e=1\}$ in $D$ to $x'=\{t=1, e=0\}$ in $D'$. Here, we have:

\begin{alignat}{3}
\s_t=& \frac{N - d_1}{N} \frac{N-d_1-c_1-d_2}{N - d_1 - c_1}\times ...&& \nonumber\\
&\times\frac{N - d_1 -...d_{t-1} - c_1 -....c_{t-1} - d_t}{N-d_1 -...d_{t-1} -c_1 ...-c_{t-1}}&&\nonumber\\
\s'_t=& \frac{N - (d_1 -1)}{N} \frac{N-(d_1-1)-(c_1+1)-d_2}{N - (d_1 -1) -(c_1+1)}\times...&& \nonumber\\
&\times\frac{N - (d_1-1) -...d_{t-1} - (c_1+1) -....c_{t-1} - d_t}{N-(d_1 -1)-...d_{t-1} -(c_1+1) ...-c_{t-1}}&&\nonumber\\
=&\frac{N-d_1+1}{N}\frac{N-d_1-c_1-d_2}{N-d_1-c_1}\times ...&&\nonumber\\
&\times\frac{N - d_1 -...d_{t-1} - c_1 -....c_{t-1} - d_t}{N-d_1 -...d_{t-1} -c_1 ...-c_{t-1}}\nonumber&&
\end{alignat}
We can calculate the upper and lower bounds for these survival functions according to Lemma~\ref{eq:app-upper-bound} and Lemma~\ref{eq:app-lower-bound}:
\begin{alignat}{3}
\label{eq:app-bounds-st-case2}
&&&\frac{N - d_1...-d_t - c_1 ... -c_{t-1}}{N} \leq \s_t\leq \frac{N-d_1...-d_t}{N}
\end{alignat}
\begin{alignat}{3}
\label{eq:app-bounds-spt-case2}
&&&\frac{N - d_1 ...- d_t - c_1 ... -c_{t-1}}{N} \leq \s'_t \leq \frac{N+1-d_1...-d_t}{N}
\end{alignat}

But we also see that except the first term, the rest of the terms are the same between $\s$ and $\s'$:
\begin{alignat}{3}
    \s_t =& \frac{N-d_1}{N}\times A_t, \qquad\s'_t =& \frac{N-d_1+1}{N}\times A_t\nonumber&&\\
    \rightarrow &\s'_t - \s_t = A_t \times \frac{1}{N}&&
    \label{eq:app-case2-diff-s}
\end{alignat}
where $A_t = a_2\times a_3 \times...a_t$ and  $\forall t: 0\leq a_t\leq 1$, so $0\leq A_t\leq 1 \rightarrow 0\leq|\s'_t - \s_t|\leq \frac{1}{N}$ and we have:
\begin{alignat}{3}
    &&&\Delta_1\s_{\text{w censor}} =\|\s' -\s\|_1 \leq \frac{T}{N}\\
    &&&\Delta_2\s_{\text{w censor}} =\|\s' -\s\|_2 \leq \frac{\sqrt{T}}{N}
    \end{alignat}
The same steps can also be applied to neighboring datasets constructed by changing $x=\{t=1, e=0\}$ in $D$ to $x=\{t=1, e=1\}$ in $D'$, to obtain the same bounds. 

Finally, we look at the case of changing $x=\{t=\tm, e=1\}$ in $D$ to $x=\{t=1, e=0\}$ in $D'$:  

\begin{alignat}{3}
\s_t=& \frac{N - d_1}{N} \frac{N-d_1-c_1-d_2}{N - d_1 - c_1}\times ...&& \nonumber\\
&\times\frac{N - d_1 -...d_{t-1} - c_1 -....c_{t-1} - d_t}{N-d_1 -...d_{t-1} -c_1 ...-c_{t-1}}&&\nonumber\\
\s'_t=& \frac{N - d_1}{N} \frac{N-d_1-(c_1+1)-d_2}{N- d_1 - (c_1+1)}\times...&& \nonumber\\
&\times\frac{N -d_1 -...d_{t-1} - (c_1+1) -....c_{t-1} - d_t}{N-d_1 -...d_{t-1} -(c_1+1) ...-c_{t-1}}&&\nonumber
\end{alignat}
We can again use Lemma~\ref{eq:app-upper-bound} and Lemma~\ref{eq:app-lower-bound} to upper and lower bound our $\s_t$ and $\s'_t$ terms:
\begin{alignat}{3}
\label{eq:app-bounds-st-case3}
&&&\frac{N - d_1...-d_t - c_1 ... -c_{t-1}}{N} \leq \s_t \leq \frac{N-d_1...-d_t}{N}
\end{alignat}
\begin{alignat}{3}
\label{eq:app-bounds-spt-case3}
&&&\frac{N - d_1 ...- d_t - (c_1+1) ... -c_{t-1}}{N} \leq \s'_t \leq \frac{N-d_1...-d_t}{N}
\end{alignat}
again, we know that $\s' \leq \s$ because it experiences one event of $e=0$ in the first time step compared to $\s$. So to upper bound the difference of these two functions, we can subtract the lower bound of $\s'_t$ from the upper bound of $\s_t$:

\begin{alignat}{3}
&&&\s_t-\s'_t = \begin{cases}
0 & t=1\\
 \frac{1 + c_1 ...+c_{t-1}}{N} & 1 < t < \tm \\
 \frac{c_1+ ...+c_{t-1}}{N}& t=\tm
\end{cases}
\end{alignat}
We see that this, in the worst case, is equivalent to our bounds found for the first case of neighboring datasets, shown in Equations~\ref{eq:app-sl1-wcensor} and~\ref{eq:app-sl2-wcensor}. So over all possible neighboring datasets we have the following sensitivities for $\s$:
{\small
\begin{alignat}{3}
&&&\Delta_1\s = \{(C=0)\rightarrow \frac{T-1}{N}, (C\neq0)\rightarrow \frac{TC}{N}\}\label{eq:app-sl1-proof}\\
&&&\Delta_2\s = \{(C=0)\rightarrow\frac{\sqrt{T-1}}{N}, (C\neq0)\rightarrow \frac{\sqrt{T}C}{N}\label{eq:app-sl2-proof}\}
\end{alignat}
}%

\myparagraph{Discussion about the sensitivity in datasets with censoring} As we discussed in Section~\ref{sec:dp-definition}, throughout the paper, we assume \textit{bounded} differential privacy. This means that the number of data points $N$ in the neighboring datasets is equal. This implies that this parameter, $N$ is public and can be shared externally without breaking the guarantees of differential privacy. $T$, the number of time bins is a hyperparameter that is not an intrinsic property of each dataset and can be selected and set publicly (as explained in e.g. Appendix D of~\cite{abadi2016deep}). However, as seen in Equations~\ref{eq:app-sl1-proof} and~\ref{eq:app-sl2-proof}, in the case of censoring in the datasets, the sensitivities depend on an additional parameter $C$, which is an inherent property of the dataset. This means that this parameter is privacy sensitive and should not be shared or used publicly. Therefore, in the most correct way to utilize differential privacy, we cannot use these sensitivities for the case of datasets with censored points. One option is to consider the worst-case scenario of $C=N$, however, this sensitivity would be so large that all the useful information of the signal would be destroyed by the DP noise. Another more elegant option, is to use \textit{smooth sensitivity}~\cite{nissim2007smooth}, and consider all the possible neighboring datasets to the actual dataset that we work with. The calculation of smooth sensitivity is usually complicated, computationally difficult and out of the scope of this paper. Unfortunately, with our current framework of directly bounding the $\s$ function, we could not achieve reasonable sensitivities for the censoring case and we defer this to future work.

\subsection{Sensitivity of $\yh(t)$}
\label{sec:appendix-proof-yh}
In this section we derive the sensitivity of $\yh$. The neighboring datasets are defined the same way as in our proof for sensitivity of $\s$ in Section~\ref{sec:appendix-proof-s}: $D'$ is derived from $D$ by changing the time of event of a point that experiences the event at the end of study $\tm$ to $t=1$, the first time bin in the study. We once more assume that the events are read in equidistant time intervals $t_i \in \{t_0 = 0, t_1 = b, ..., T_{\max} \}$ with a fixed bin size of $b$. We first derive the sensitivity in the absence of censored data and then proceed to the more general case with censored data in the dataset. 

\myparagraph{Sensitivity with no censoring in dataset} We denote the probability mass function of the dataset $D$ as $\yh$ and the probability mass function of its neighboring dataset $D'$ as $\yh'$. According to Equation~\ref{eq:s2y} we have $\yh_t=\s_{t-1} - \s_t$ and $\yh'_t=\s'_{t-1}-\s'_t$, where $\yh_t(\yh'_t)$ and $\s_t(\s'_t)$ indicate the value of the probability mass function and the KM estimator at times $t$, respectively. So:
\begin{alignat}{3}
\label{eq:app-yh-nocens-1}
&\yh_1 =&& \s_0 - \s_1 =1 - \frac{N-d_1}{N}\nonumber\\
&\yh'_1 =&& \s_0' - \s_1' =1 - \frac{N-d_1-1}{N}\nonumber\\
&\Rightarrow&& \yh'_1 - \yh_1 = \frac{1}{N}
\end{alignat}
According to Equations~\ref{eq:app-st} and~\ref{eq:app-stp}, we have:
\begin{alignat}{3}
&\yh_t =&& \frac{N-d_1-...d_{t-1}}{N} - \frac{N-d_1-...d_t}{N} = \frac{d_t}{N}\nonumber\\
&\yh'_t =&& \frac{N-d_1-1-...d_{t-1}}{N} - \frac{N-1- d_1-...d_t}{N} = \frac{d_t}{N}\nonumber\\
&\Rightarrow&& \yh'_t - \yh_t = 0 
\end{alignat}
There is also a difference in the last term due to changing the event of one datapoint from $t=\tm$ to $t=1$:
\begin{alignat}{3}
&\yh_{\tm} =&& \frac{N-d_1-...d_{\tm-1}}{N} - \frac{N-d_1-...d_{\tm}}{N} = \frac{d_{\tm}}{N}\nonumber\\
&\yh'_{\tm} =&& \frac{N-d_1-1-...d_{\tm-1}}{N} -\nonumber\\ 
&&&\frac{N- d_1-1-...-(d_{\tm}-1)}{N} = \frac{d_{\tm}-1}{N}\nonumber\\
&\Rightarrow&& \yh'_t - \yh_t = \frac{-1}{N} 
\end{alignat}
This means that the sensitivity of the probability mass function for datasets with no censoring will be:
\begin{alignat}{3}
\label{eq:app-proof-y-nocens}
&&&\Delta_{1}\yh_{\text{no censor}} = \|\yh' - \yh\|_1 = \frac{2}{N}\\
&&&\Delta_{2}\yh_{\text{no censor}} = \|\yh' - \yh\|_2 = \frac{\sqrt{2}}{N}
\end{alignat}

\myparagraph{Sensitivity with censoring in dataset} Now we derive the sensitivity for the general case of datasets also containing censored datapoints. First, we consider the case of changing a point $x = \{t=\tm, e=1\}$ in $D$ to $x' = \{t=1, e=1\}$ in $D'$.  
\begin{alignat}{3}
&&&\yh_t = \s_{t-1}- \s_{t}= \s_{t-1}-\nonumber \\
&&&\s_{t-1}\frac{N-d_1-...d_{t-1}-c_1-...c_{t-1}-d_t}{N-d_1-...d_{t-1}-c_1-...c_{t-1}}=\nonumber\\
&&&\s_{t-1}[1-\frac{N-d_1-...d_{t-1}-c_1-...c_{t-1}-d_t}{N-d_1-...d_{t-1}-c_1-...c_{t-1}}]=\nonumber\\
&&&\s_{t-1}[\frac{d_t}{N-d_1-...d_{t-1}-c_1-...c_{t-1}}]
\end{alignat}
Where for the first line we use the definition of $S_t = S_{t-1}\times\frac{r_t-d_t}{r_t}$. In the same way, we can show that:
\begin{alignat}{3}
&&&\yh'_{t} = \s'_{t-1} - \s'_{t}\nonumber\\
&&&= \begin{cases}
\frac{d_1+1}{N} & t=1\\
\s'_{t-1}[\frac{d_t}{N-1-d_1-...d_{t-1}-c_1-...c_{t-1}}]&1<t<\tm\\
\s'_{\tm-1}[\frac{d_t-1}{N-1-d_1-...d_{t-1}-c_1-...c_{t-1}}] & t=\tm
\end{cases}
\end{alignat}
here the case of $t=1$ and $t=\tm$ are different because we add and subtract one from $d_t$ at these times respectively. Now we can use Lemma~\ref{eq:app-upper-bound} and~\ref{eq:app-lower-bound} and Inequalities~\ref{eq:app-bounds-st-case1} and~\ref{eq:app-bounds-spt-case1} and the same trick we used for $\s$ to lower and upper bound $\yh_t$ and $\yh'_t$ for $1<t<\tm$:
\begin{alignat}{3}
&&\frac{d_t}{N}\leq\yh_t\leq\frac{d_t+\sum_{i=1}^{t-1}c_i}{N}\\
&&\frac{d_t}{N}\leq\yh'_t\leq\frac{d_t+\sum_{i=1}^{t-1}c_i}{N}
\end{alignat}
and again to find an upper bound on the difference of $|\yh'_t - \yh_t|$ we can subtract the lower bound of $\yh_t$ from the upper bound of $\yh'_t$ (or vice versa):
\begin{alignat}{3}
&&&|\yh'_t - \yh_t|=\nonumber\\
&&&\begin{cases}
\frac{1}{N} & t=1\\
\frac{d_t+\sum_{i=1}^{t-1}c_i}{N} - \frac{d_t}{N} = \frac{\sum_{i=1}^{t-1}c_i}{N}&1<t<\tm\\
\frac{d_{\tm}-1+\sum_{i=1}^{\tm-1}c_i}{N} - \frac{d_{\tm}}{N} \\= \frac{-1+\sum_{i=1}^{\tm-1}c_i}{N}& t=\tm
\end{cases}
\end{alignat}
Note that this is again consistent with our results for the no censoring case. By defining $C = \sum_{i=1}^{T_{\max}-1} c_i$, we can construct a general sensitivity over the whole time-frame of the study:
\begin{alignat}{3}
\label{eq:app-ybound-case1-l1}
&&&\Delta_1 \yh_{\text{w censor}} = \|\yh' - \yh\|_1 = \frac{1}{N}+\frac{c_1}{N}+...\frac{c_1+...c_{\tm-1}-1}{N}\nonumber\\
&&&\leq\frac{TC}{N}
\end{alignat}
\begin{alignat}{3}
\label{eq:app-ybound-case1-l2}
&&&\Delta_2 \yh _{\text{w censor}}=\|\yh' - \yh\|_2 \nonumber\\
&&&= \sqrt{(\frac{1}{N})\strut^2+(\frac{c_1}{N})\strut^2+...(\frac{c_1+...c_{\tm-1}}{N})\strut^2} \nonumber\\
&&&\leq\sqrt{(\frac{C}{N})\strut^2+...+(\frac{C}{N})\strut^2} =\frac{\sqrt{T}C}{N}
\end{alignat}
Where $T$ is the number of time bins, i.e. $T = T_{\max}/b$.

Now let us look at the case of neighboring datasets for when we change a point $x=\{t=1, e=1\}$ in $D$ to $x=\{t=1, e=0\}$ in $D'$. 
For the first time step we have: 
\begin{alignat}{3}
&\yh_1 =&& \s_0 - \s_1 =1 - \frac{N-d_1}{N}=\frac{d_1}{N}\nonumber\\
&\yh'_1 =&& \s_0' - \s_1' =1 - \frac{N-(d_1-1)}{N}=\frac{d_1-1}{N}\nonumber\\
&\Rightarrow&& \yh_1 - \yh'_1 = \frac{1}{N}
\end{alignat}
And for other times we have:
\begin{alignat}{3}
&&&\yh_t = \s_{t-1}[\frac{d_t}{N-d_1...-d_{t-1}-c_1...-c_{t-1}}]\nonumber\\
&&&\yh'_t = \s'_{t-1}[\frac{d_t}{N-d_1...-d_{t-1}-c_1...-c_{t-1}}]\nonumber    
\end{alignat}
Since $\s'_1 > \s_1$, and the rest of the multiplicative terms are always identical between $\s_t$ and $\s'_t$, we have $\forall t>1 : \yh'_t> \yh_t$, and the difference is:
\begin{alignat}{3}
&&& \yh'_t - \yh_t = [\frac{d_t}{N-d_1...-d_{t-1}-c_1...-c_{t-1}}](\s'_t -\s_t)
\end{alignat}
by using Equation~\ref{eq:app-case2-diff-s} and Lemma~\ref{eq:app-upper-bound} we have:
\begin{alignat}{3}
&&&0\leq \yh'_t - \yh_t \leq \frac{1}{N}\frac{d_t}{N-d_1...-d_{t-1}-c_1...-c_{t-1}}\\
&&&0 \leq\yh'_t - \yh_t \leq \frac{1}{N}\frac{d_t+d_1...d_{t-1}+c_1...+c_{t-1}}{N}\\
&&& 0 \leq\yh'_t - \yh_t \leq \frac{1}{N}\frac{N}{N}=\frac{1}{N}
\end{alignat}
the last line comes from the fact that we know that the upper bound for $d_1+d_2...+d_t+c_1...+c_{t-1}\leq N$. So over the complete time-frame of the study we have:
\begin{alignat}{3}
&&&\Delta_1\yh_{\text{w censor}} = \|\yh' - \yh\|_1 = \frac{T}{N}\\
&&&\Delta_2\yh_{\text{w censor}} = \|\yh' - \yh\|_2 = \frac{\sqrt{T}}{N}
\end{alignat}
which is a smaller bound compared to the one found for the first case in Equations~\ref{eq:app-ybound-case1-l1} and~\ref{eq:app-ybound-case1-l2}. Similarly, we can prove that the sensitivity for the reverse case of changing a point $x=\{t=1, e=0\}$ in $D$ to $x=\{t=1, e=1\}$ in $D'$ also results in a bound still smaller than Equations~\ref{eq:app-ybound-case1-l1} and~\ref{eq:app-ybound-case1-l2}. 

Finally, let us look at the case of neighboring datasets constructed by changing $x=\{t=\tm, e=1\}$ in $D$ to $x=\{t=1, e=0\}$ in $D'$ for $1<t<\tm$:
\begin{alignat}{3}
&&& \yh_t = \s_{t-1} - \s_{t} = \s_{t-1}[\frac{d_t}{N-d_1...-d_{t-1}-c_1...-c_{t-1}}]\\
&&& \yh'_t = \s'_{t-1} - \s'_{t} = \s'_{t-1}[\frac{d_t}{N-d_1...-d_{t-1}-(c+1)...-c_{t-1}}]
\end{alignat}

we once more use Lemma~\ref{eq:app-upper-bound} and Lemma~\ref{eq:app-lower-bound} with Inequalities~\ref{eq:app-bounds-st-case3} and~\ref{eq:app-bounds-spt-case3} to upper and lower bound these probability estimators: 

\begin{alignat}{3}
&&&\frac{d_t}{N}\leq\yh_t\leq\frac{d_t+\sum_{i=1}^{t-1}c_i}{N}\\
&&&\frac{d_t}{N}\leq\yh'_t\leq\frac{1 + d_t+\sum_{i=1}^{t-1}c_i}{N}
\end{alignat}

To upper bound the different $|\yh'_t - \yh_t|$, we subtract the lower bound of $\yh_t$ from the upper bound of $\yh'_{t}$: 
\begin{alignat}{3}
&&&|\yh'_t - \yh_t|=\nonumber\\
&&&\begin{cases}
0 & t=1\\
\frac{1+d_t+\sum_{i=1}^{t-1}c_i}{N} - \frac{d_t}{N} = \frac{1+\sum_{i=1}^{t-1}c_i}{N}&1<t<\tm\\
\frac{d_{\tm}-1+1 +\sum_{i=1}^{\tm-1}c_i}{N} - \frac{d_{\tm}}{N} \\= \frac{\sum_{i=1}^{\tm-1}c_i}{N}& t=\tm
\end{cases}
\end{alignat}
So over the complete time-span of the dataset we have:
\begin{alignat}{3}
&&&\Delta_1\yh_{\text{w censor}} = \|\yh' - \yh\|_1 \leq \frac{TC}{N}\\
&&&\Delta_2\yh_{\text{w censor}} = \|\yh' - \yh\|_2 \leq \frac{\sqrt{T}C}{N}
\end{alignat}

So, finally we can write the most general case of sensitivity for the probability mass function as:
{\small
\begin{alignat}{3}
&&&\Delta_1 \yh = \{(C=0)\rightarrow \frac{2}{N},(C\neq0)\rightarrow \frac{TC}{N}\}\label{eq:app-y1-sens-final}\\
&&&\Delta_2 \yh = \{(C=0)\rightarrow\frac{\sqrt{2}}{N},(C\neq0)\rightarrow\frac{\sqrt{T}C}{N}\}\label{eq:app-y2-sens-final}
\end{alignat}
}

\myparagraph{Discussion about sensitivities for datasets with censoring} As we discussed in Section~\ref{sec:dp-definition} and also in the previous section, the choice of bounded differential privacy allows us to treat the number of data points in the dataset, $N$, as a non-private parameter. Again, $T$, the number of time bins, is a hyperparameter that is not intrinsic to the dataset. These type of hyperparameters can be chosen and set, publicly as explained in Appendix D of~\cite{abadi2016deep}. However, we see that the sensitivity of the probability function when censoring is present in the dataset is again dependent on the total number of censored points $C$. However, $C$ is a property of the dataset and by including it in the sensitivities, we cannot have the usual DP guarantees anymore. We defer finding a theoretically correct bound for this case to future work.


\subsection{Datasets}
\label{sec:app-dataset}
For our experiments, we use the following real-world medical datasets:

\textbf{Rotterdam and German Breast Cancer Study Group (GBSG):} Contains data from 2,232 breast cancer patients from the Rotterdam tumor bank~\cite{foekens2000urokinase} and the German Breast Cancer Study Group (GBSG)~\cite{schumacher1994randomized}. $960 
 (43\%)$ patients are censored. The data is pre-processed similar to~\cite{katzman2018deepsurv} with a maximum survival duration of 87 months.  

\textbf{The Molecular Taxonomy of Breast Cancer International Consortium (METABRIC):} This dataset contains gene and protein expressions of 1904 individuals~\cite{curtis2012genomic}. We use a dataset prepared  similar to~\cite{katzman2018deepsurv}. The maximum duration of the study is 355 months ($\sim30$ years), 801 $(42\%$ of total) patients were right-censored and 1103 $(58\%$ of total) were followed until death.

\textbf{Study to Understand Prognoses Preferences Outcomes
and Risks of Treatment (SUPPORT):} This dataset consists of 8873 seriously ill adults~\cite{knaus1995support}. The dataset has a maximum survival time of 2029 days $(\sim5.6)$ years and $32\% (2839)$ of the data is right-censored. We use a pre-processed version according to~\cite{katzman2018deepsurv}.

\subsection{Metrics}
\label{sec:app-metrics}
\myparagraph{Logrank Test}
Consider that we would like to compare the survival distribution of two populations $j = \{1, 2\}$, and the combined data over these two populations has $t=\{1,..., T\}$ distinct events times. Here, the null hypothesis is that the two populations have the same survival distribution. We define $d_{t,j}$ as the number of events observed in group $j$ at time $t$, and $d_t= d_{t, 1} + d_{t, 2}$ as the total events at time $t$. If we consider $r_{t, j}$ as the number at risk in group $j$ at time $t$ and $r_t = r_{t, 1} + r_{t, 2}$ as the total number at risk at time $t$, the expected number of events for group $j$ at time $t$ under the null hypothesis would be $E_{t, j} = r_{t, j}\frac{d_t}{r_t}$. 

Using these notations, we can construct test statistics for population 1 (without loss of generality) under the null hypothesis as:
\begin{eqnarray}
Z = \frac{\sum_{t=1}^T (d_{t, 1} - E_{t, 1})}{\sqrt{\sum_{t=1}^T V_{t,1}}}
\end{eqnarray}
where $V_{t, 1}$ is the variance in group 1 at time $t$. Under the assumption that $d_{t,1}$ have a hypergeometric distribution, the variance is defined as:
\begin{eqnarray}
V_{t, 1} &=& E_{t, 1} (\frac{r_t - d_t}{r_t}) (\frac{r_t - r_{t,1}}{r_t - 1})\nonumber\\
&=& \frac{r_{t,1}r_{t,2}d_t(r_t-d_t)}{r_t^2 (r_t-1)}
\end{eqnarray}
By the central limit theorem $Z \sim \mathcal{N}(0, 1)$, where $\mathcal{N}(0,1)$ is a Gaussian probability with mean 0 and variance 1. Based on this approximation, the value of $Z$ can be compared with the tails of the standard Gaussian distribution to obtain the $p-$value of the null hypothesis.

\myparagraph{Confidence Intervals}
The variance of the KM estimator according to Greenwood's formula~\cite{greenwood1926report} is:
\begin{eqnarray}
\vh(t) =\sigh^2(t)= \s^2(t)\sum_{t'\leq t} \frac{d_{t'}}{r_{t'}(r_{t'}-d_{t'})}
\end{eqnarray}
Once more, for large samples, the Kaplan-Meier curve evaluated at time $t$ is assumed to be normally distributed and the $100(1-\alpha)\%$ confidence interval (CI) can be obtained as:
\begin{eqnarray}
\s(t) \pm z_{1-\alpha/2}\sigh(t)
\label{eq:ci-normal}
\end{eqnarray}
where $z_{1-\alpha/2}$ is the $1-\alpha/2$ fractile of the standard normal distribution. This assumption can be improved for smaller sample size, using Greenwood's exponential \textit{log-log} formula~\cite{sawyer2003greenwood}:
\begin{eqnarray}
\s(t)^{\exp(\pm z_{1-\alpha/2}\sigh(t)/[\s(t)\ln \s(t)])}
\end{eqnarray}

For some of our experiments, we also introduce a measure of median that makes comparison between different datasets easy. We call this the \textit{calibrated median difference} or cmd:
\begin{eqnarray}
    \text{cmd}_{n, b} = \frac{\mathopen|\text{median}_{n, b} - \text{median}_{\text{original}}\mathclose|}{\text{median}_{\text{original}}}
\end{eqnarray}

\subsection{Construction of Surrogate Datasets}
\label{sec:app-exp-surrogates}
Our metrics are dependent on access to individual data points and their times of event in each dataset. This problem motivated us to develop our surrogate dataset generation Algorithm~\ref{algo:surrogate} and it is a very important aspect of all of our experiments, since no matter which route we take in Figure~\ref{fig:overall-graph} of our workflow, we eventually need to reconstruct surrogate datasets to be able to calculate the performance metrics for our methods. In this section, we study the performance of our surrogate dataset generation method in a centralized setting and without any privacy-preserving mechanism. 

\myparagraph{Parameters} By inspecting Algorithm~\ref{algo:surrogate}, we see that $n$, the total number of points we choose to populate a KM curve with, is one of the hyperparameters that we need to optimize. We also discuss in Section~\ref{sec:method} that \dpm and our \dps and \dpy methods are based on equidistant time of event discretization. Since we would like to remain consistent among all DP methods, we choose an equidistant grid of size $b$ for times of events for all our experiments from this point on. Now $b$ is also a hyperparameter that can change the results and needs optimization.

\myparagraph{Setup} For all our datasets we first produce the discretized KM function $\s$ and its corresponding probability function $\yh$ based on our bin length $b$, then run our surrogate generation algorithm with $n = \{0.5\nb, \nb, 2\nb\}$ data points, where $\bar{N}$ is the total number of uncensored data points in each dataset. For the discretization step, we choose $b= \{1, 2, 4, 6\}$, which is measured in months for \textcolor{magenta}{METABRIC} and \textcolor{blue}{GBSG} and in days for \textcolor{teal}{SUPPORT}. The reason we choose a relatively smaller binning size for \textcolor{teal}{SUPPORT} is that the study is done on a shorter time frame compared to the other two datasets, and the initial drop in the value of survival function is much more drastic compared to the other two dataset, with a median survival time of only 57 days for $e=1$ points. In comparison, \textcolor{magenta}{METABRIC} and \textcolor{blue}{GBSG} have median survival time of 86 months and 24 months for $e=1$ datapoints, respectively. 

In Tables~\ref{tab:surrogates-median} and ~\ref{tab:surrogates-pvalue} we report the calibrated median difference (cmd) and $p-$values to the original dataset for \textcolor{teal}{SUPPORT}, \textcolor{blue}{GBSG}, and \textcolor{magenta}{METABRIC}. The \pv shows if a surrogate dataset is statistically dissimilar to the original dataset and higher values of it are preferred. The cmd shows how far from the real median the reconstructed median is, and lower values of this parameter are desired. 

\myparagraph{Effect of binning size} Based on both cmd and $p-$value smaller binning lengths of $b=\{1, 2\}$ work best for \textcolor{teal}{SUPPORT} and \textcolor{blue}{GBSG}. For \textcolor{blue}{GBSG}, the effect of discretization only (before construction of the surrogate dataset) for $b>2$ is enough to drop the \pv between the discretized KM estimator and the original curve below significant level. The same effect happens for \textcolor{teal}{SUPPORT} for $b>4$. \textcolor{magenta}{METABRIC} seems to be more robust to discretization and we observe acceptable results for surrogate dataset generation. In general, for all the datasets we start to see a degradation in the performance for larger binning lengths. 
\myparagraph{Effect of number of points} We also observe that our surrogate generation method is very robust against changes in $n$, the number of points used to construct the surrogate dataset. We expect the median to be constant with respect to the number of datapoints, as long as we have enough datapoints to populate all the bins over the time of study. However, \pv can be less robust, as it is directly calculated on data points and here, number of points we choose to calculate it with is important. But we observe that \pv is also always above the significance level of 0.05 for small enough binning size and enough number of points to successfully recreate the KM function. 

\begin{table*}
\caption{Calibrated difference to the median of the original dataset for reconstructed surrogate datasets for \textcolor{teal}{SUPPORT}, \textcolor{blue}{GBSG} and \textcolor{magenta}{METABRIC} only for event $e=1$. A lower value shows a more accurate reconstruction. The parameter $n$ is the total number of data points in each dataset and $b$ is the time discretization length.}
  \label{tab:surrogates-median}
\centering
{\begin{tabular}{c|c|c|c|c}
\toprule
    $n$&$b=1$ &$ b=2$& $b=4$& $b=6$\\
    \hline
    $2\nb$& \textcolor{teal}{0.00}, \textcolor{blue}{0.042}, \textcolor{magenta}{0.000}&\textcolor{teal}{0.018}, \textcolor{blue}{0.083}, \textcolor{magenta}{0.000}& \textcolor{teal}{0.053}, \textcolor{blue}{0.167}, \textcolor{magenta}{0.023}& \textcolor{teal}{0.053}, \textcolor{blue}{0.25}, \textcolor{magenta}{0.047}\\
    \hline
    $\nb$ & \textcolor{teal}{0.00}, \textcolor{blue}{0.042}, \textcolor{magenta}{0.000}&\textcolor{teal}{0.018}, \textcolor{blue}{0.083}, \textcolor{magenta}{0.000}& \textcolor{teal}{0.053}, \textcolor{blue}{0.167}, \textcolor{magenta}{0.023}& \textcolor{teal}{0.053}, \textcolor{blue}{0.25}, \textcolor{magenta}{0.047}\\
    \hline
    $0.5\nb$& \textcolor{teal}{0.00}, \textcolor{blue}{0.042}, \textcolor{magenta}{0.047}& \textcolor{teal}{0.018}, \textcolor{blue}{0.083}, \textcolor{magenta}{0.023}&  \textcolor{teal}{0.053}, \textcolor{blue}{0.167}, \textcolor{magenta}{0.023}& \textcolor{teal}{0.053}, \textcolor{blue}{0.25}, \textcolor{magenta}{0.047}\\
    \bottomrule
\end{tabular}}
\end{table*}

\begin{table*}
\caption{p-value between the surrogate and original datasets for \textcolor{teal}{SUPPORT}, \textcolor{blue}{GBSG} and \textcolor{magenta}{METABRIC} only for event $e=1$. Higher values are preferable, and small values show a statistically significant separation between the reconstructed dataset and the original. The parameter $n$ is the total number of data points in each dataset and $b$ is the time discretization length.}
  \label{tab:surrogates-pvalue}
\centering
{\begin{tabular}{c|c|c|c|c}
\toprule
    $n$&$b=1$&$ b=2$& $b=4$& $b=6$\\
    \hline
    $2\nb$&\textcolor{teal}{1.00}, \textcolor{blue}{0.21}, \textcolor{magenta}{0.71}& \textcolor{teal}{0.53}, \textcolor{blue}{0.03}, \textcolor{magenta}{0.52}& \textcolor{teal}{0.10}, \textcolor{blue}{0.00}, \textcolor{magenta}{0.23}& \textcolor{teal}{0.01}, \textcolor{blue}{0.00}, \textcolor{magenta}{0.08}\\
    \hline
    $\nb$ &\textcolor{teal}{1.00}, \textcolor{blue}{0.33}, \textcolor{magenta}{0.78}& \textcolor{teal}{0.63}, \textcolor{blue}{0.08}, \textcolor{magenta}{0.62}& \textcolor{teal}{0.20}, \textcolor{blue}{0.00}, \textcolor{magenta}{0.35}& \textcolor{teal}{0.04}, \textcolor{blue}{0.00}, \textcolor{magenta}{0.18}\\
    \hline
    $0.5\nb$ &\textcolor{teal}{0.02}, \textcolor{blue}{0.27}, \textcolor{magenta}{0.04}&\textcolor{teal}{0.14}, \textcolor{blue}{0.10}, \textcolor{magenta}{0.16}& \textcolor{teal}{0.20}, \textcolor{blue}{0.00}, \textcolor{magenta}{0.16}& \textcolor{teal}{0.11}, \textcolor{blue}{0.00}, \textcolor{magenta}{0.11}\\
    \bottomrule
\end{tabular}}
\end{table*}

\subsection{Hyperparameter Selection}
\label{sec:app-exp-hyperparameter}
Our differentially private methods depend on a few hyperparameters. These can influence the performance of our methods. In this section, we strive to evaluate the effect of these parameters on our methods in a centralized setting, where all the data is accessible centrally. We again run all the experiments on the noncensored portion of our datasets. Later, we will generalize the results of this section to run experiments in a collaborative setting. In the following, we will go through each DP method and explain which parameters are important for each method and how we choose the optimal values.

\myparagraph{\dps} According to Section~\ref{sec:dpsurvdef} and Theorem~\ref{theorem:dctws}, the sensitivity of our \dps method scales like
\begin{eqnarray}
    \Delta_1 D^k \propto \sqrt{k} \sqrt{T-1} = \sqrt{k}\sqrt{(\tm/b) -1}
\end{eqnarray}

where $T_{\max}$ is the maximum duration in the study, $b$ is the discretization binning size and $k$ is the number of the first coefficients chosen from the discrete cosine transform (DCT). So for a smaller sensitivity of \dps and, therefore, a better expected utility, we strive to choose the smallest first $k$ coefficient of the DCT and the largest binning size $b$ possible. This is a classic privacy/performance trade-off problem, because the more coefficients $k$ we take and the smaller our discretization step $b$, the more accurate our reconstructed private survival function becomes. However, these increase the sensitivity and more noise should be added for the same level of privacy guarantee $\varepsilon$.

Based on our experiments in the previous section and the fact that in general our surrogate generation algorithm works best for small bin sizes, we pick $b=\{1, 2, 4, 6\}$ and to make the surrogate datasets from noisy survival functions, we set $n = \nb$ where $\nb$ is the total number of uncensored data points for each dataset. 

Table~\ref{tab:central-DPS-parameter} shows the calibrated median difference (cmd) and \pv between the original and the reconstructed noisy survival function, for $\varepsilon = 1.0$ and averaged over 100 runs of the algorithm for different fractions of the total coefficients of the DCT, $k$. We choose $\varepsilon=1.0$ because it gives tight theoretical guarantees for privacy in a centralized setting~\cite{ponomareva2023dp,desfontainesblog20211001, hsu2014differential, roth2022private}.

\begin{table*}
\caption{Calibrated median difference and $p$-value between the original dataset and the private dataset for $e=1$, obtianed by \dps with $\varepsilon=1$ for \textcolor{teal}{SUPPORT}, \textcolor{blue}{GBSG} and \textcolor{magenta}{METABRIC}. All results are averaged over five independent runs of the DP algorithm. The arrows show whether a lower or higher value indicates better utility of our DP method.}
  \label{tab:central-DPS-parameter}
\centering
{\begin{tabular}{c|c|c|c|c || c|c|c|c}
\toprule
    & \multicolumn{4}{c||}{cmd $\downarrow$}& \multicolumn{4}{c}{$p-$value $\uparrow$}\\
    \hline
    $k$&$b=1$ &$ b=2$& $b=4$& $b=6$ &$b=1$ &$b=2$&$b=4$& $b=6$\\
    \hline
    $5\%$ &\textcolor{teal}{0.08}, \textcolor{blue}{0.04}, \textcolor{magenta}{0.04}& \textcolor{teal}{0.11}, \textcolor{blue}{0.33}, \textcolor{magenta}{0.02}& \textcolor{teal}{0.26}, \textcolor{blue}{1.00}, \textcolor{magenta}{0.07}&\textcolor{teal}{0.58}, \textcolor{blue}{1.0}, \textcolor{magenta}{0.19} & \textcolor{teal}{0.48}, \textcolor{blue}{0.55}, \textcolor{magenta}{0.56}&\textcolor{teal}{0.53}, \textcolor{blue}{0.04}, \textcolor{magenta}{0.57}& \textcolor{teal}{0.67}, \textcolor{blue}{0}, \textcolor{magenta}{0.11}&\textcolor{teal}{0.72}, \textcolor{blue}{0}, \textcolor{magenta}{0.38} \\
    \hline
    $10\%$ &\textcolor{teal}{0.17}, \textcolor{blue}{0.01}, \textcolor{magenta}{0.07}&\textcolor{teal}{0.07}, \textcolor{blue}{0.08}, \textcolor{magenta}{0.04}& \textcolor{teal}{0.11}, \textcolor{blue}{0.33}, \textcolor{magenta}{0.02}&\textcolor{teal}{0.10}, \textcolor{blue}{1.0}, \textcolor{magenta}{0.02}& \textcolor{teal}{0.40}, \textcolor{blue}{0.41}, \textcolor{magenta}{0.47}&\textcolor{teal}{0.37}, \textcolor{blue}{0.16}, \textcolor{magenta}{0.51}& \textcolor{teal}{0.29}, \textcolor{blue}{0}, \textcolor{magenta}{0.36}&\textcolor{teal}{0.26}, \textcolor{blue}{0}, \textcolor{magenta}{0.26}\\
    \hline
    $15\%$&\textcolor{teal}{0.19}, \textcolor{blue}{0.02}, \textcolor{magenta}{0.08} &\textcolor{teal}{0.16}, \textcolor{blue}{0.01}, \textcolor{magenta}{0.05}& \textcolor{teal}{0.14}, \textcolor{blue}{0.17}, \textcolor{magenta}{0.02}&\textcolor{teal}{0.10}, \textcolor{blue}{0.5}, \textcolor{magenta}{0.04}&\textcolor{teal}{0.30}, \textcolor{blue}{0.40}, \textcolor{magenta}{0.34} &\textcolor{teal}{0.27}, \textcolor{blue}{0.13}, \textcolor{magenta}{0.43}& \textcolor{teal}{0.20}, \textcolor{blue}{0}, \textcolor{magenta}{0.43}&\textcolor{teal}{0.11}, \textcolor{blue}{0}, \textcolor{magenta}{0.21}\\
    \hline
    $20\%$ &\textcolor{teal}{0.18}, \textcolor{blue}{0.02}, \textcolor{magenta}{0.08}& \textcolor{teal}{0.17}, \textcolor{blue}{0.02}, \textcolor{magenta}{0.06}& \textcolor{teal}{0.09}, \textcolor{blue}{0.17}, \textcolor{magenta}{0.05}&\textcolor{teal}{0.16}, \textcolor{blue}{0.3}, \textcolor{magenta}{0.05}&\textcolor{teal}{0.26}, \textcolor{blue}{0.38}, \textcolor{magenta}{0.32} &\textcolor{teal}{0.25}, \textcolor{blue}{0.14}, \textcolor{magenta}{0.41}& \textcolor{teal}{0.17}, \textcolor{blue}{0}, \textcolor{magenta}{0.32}&\textcolor{teal}{0.06}, \textcolor{blue}{0}, \textcolor{magenta}{0.25}\\
    \bottomrule
\end{tabular}}
\end{table*}

\begin{table*}
\caption{Calibrated median difference and $p$-value between the original dataset and the private dataset for $e=1$ obtianed by \dpy with $\varepsilon=1$ for \textcolor{teal}{SUPPORT}, \textcolor{blue}{GBSG} and \textcolor{magenta}{METABRIC}. All results are averaged over 5 independent runs of the DP algorithm. The arrows show if a lower or a higher value indicates better utility of our DP method.}
  \label{tab:central-DPy-parameter}
\centering
{\begin{tabular}{c|c|c|c || c|c|c|c}
\toprule
    \multicolumn{4}{c||}{cmd $\downarrow$}& \multicolumn{4}{c}{$p-$value $\uparrow$}\\
    \hline
    $b=1$&$ b=2$& $b=4$ & $b=6$& $b=1$&$b=2$&$b=4$& $b=6$\\
    \hline
      \textcolor{teal}{0.90}, \textcolor{blue}{0.03}, \textcolor{magenta}{0.13}&\textcolor{teal}{0.28}, \textcolor{blue}{0.05}, \textcolor{magenta}{0.04}& \textcolor{teal}{0.10}, \textcolor{blue}{0.09}, \textcolor{magenta}{0.02}& \textcolor{teal}{0.06}, \textcolor{blue}{0.16}, \textcolor{magenta}{0.05} & \textcolor{teal}{0.00}, \textcolor{blue}{0.32}, \textcolor{magenta}{0.00}& \textcolor{teal}{0.00}, \textcolor{blue}{0.13}, \textcolor{magenta}{0.02}& \textcolor{teal}{0.00}, \textcolor{blue}{0.00}, \textcolor{magenta}{0.11}&\textcolor{teal}{0.00}, \textcolor{blue}{0.00}, \textcolor{magenta}{0.08}\\
    \bottomrule
\end{tabular}}
\end{table*}

By comparing the cmd and \pv and striving to choose the lowest possible $k$ and the highest $b$, which return a reasonable performance, we choose: for \textcolor{teal}{SUPPORT} $\{b = 2, k = 10\% \} $, for \textcolor{blue}{GBSG} $\{b=1, k=10\%\}$ and for \textcolor{magenta}{METABRIC} $\{b=6, k=10\%\}$. From now on, we will always use these parameters for our experiments involving \dps. 

\myparagraph{\dpy} As explained in Section~\ref{sec:dpprobdef} and Theorem~\ref{theorem:ysensitivity}, for \dpy and $L_1$ sensitivity, the only important hyperparameter is the discretization grid size of the duration. With bigger binning size, we add noise to a more aggregated function of the data and thus achieve a better level of privacy with less severe adverse effect of noise on utility~\cite{desfontainesblog20210616}. However, as seen in Section~\ref{sec:app-exp-surrogates}, larger bin size degrades the surrogate dataset generation.  

To study this effect, we ran \dpy with $\varepsilon=1.0$ for bin sizes $b = \{1, 2, 4, 6\}$. The averaged cmd and \pv over 100 runs of the algorithm are shown in Table~\ref{tab:central-DPy-parameter}. We again use $n = n_{\text{tot}}$ points to generate the surrogate datasets. Although sometimes \pv falls below the significance level, for example for \textcolor{teal}{SUPPORT}, we still choose the binning size based on the best value of cmd, as \pv alone is not sufficient to measure performance. Based on these metrics, we choose: $b=6$ for \textcolor{teal}{SUPPORT}, $b=2$ for \textcolor{blue}{GBSG} and $b=4$ for \textcolor{magenta}{METABRIC}. These binning sizes will be used for our \dpy method for all forthcoming experiments. 

\myparagraph{\dpm} As explained in Section~\ref{sec:dpmdef}, the original algorithm of DP-Matrix~\cite{gondara2020differentially} uses no binning, furthermore, no post-processing or pre-processing is mentioned. We explain that we fix this issue in our improved version \dpm. This will be in favor of the performance of this algorithm, because again the noise will be added to aggregated data, which increase the value-to-noise level and thus improve utility~\cite{desfontainesblog20210616}. We also add post-processing steps to ensure that the noisy number of at risk group $r'_t$ does not become negative at any step and that the algorithm halts once all the datapoints have experienced an event. Now the only hyperparameter that \dpm depends on is the binning size $b$.  

Table~\ref{tab:central-DPm-parameter} shows the cmd and \pv averaged over 100 runs of \dpm on our datasets. The same as \dpy we strive to find the largest binning size that returns good utility. Based on the metrics, we choose $b=6$ for \textcolor{teal}{SUPPORT}, $b=2$ for \textcolor{blue}{GBSG} and $b=6$ for \textcolor{magenta}{METABRIC}. These binning sizes will be used for our \dpy method for all forthcoming experiments.

\begin{table*}
\caption{Calibrated median difference and $p$-value between the original dataset and the private dataset for $e=1$ obtianed by \dpm with $\varepsilon=1$ for \textcolor{teal}{SUPPORT}, \textcolor{blue}{GBSG} and \textcolor{magenta}{METABRIC}. All results are averaged over 5 independent runs of the DP algorithm. The arrows show if a lower or a higher value indicates better utility of our DP method.}
  \label{tab:central-DPm-parameter}
\centering
{\begin{tabular}{c|c|c|c || c|c|c|c}
\toprule
    \multicolumn{4}{c||}{cmd $\downarrow$}& \multicolumn{4}{c}{$p-$value $\uparrow$}\\
    \hline
    $b=1$&$ b=2$& $b=4$ & $b=6$& $b=1$&$b=2$&$b=4$& $b=6$\\
    \hline
      \textcolor{teal}{0.02}, \textcolor{blue}{0.03}, \textcolor{magenta}{0.06}&\textcolor{teal}{0.02}, \textcolor{blue}{0.03}, \textcolor{magenta}{0.04}& \textcolor{teal}{0.04}, \textcolor{blue}{0.08}, \textcolor{magenta}{0.03}& \textcolor{teal}{0.05}, \textcolor{blue}{0.16}, \textcolor{magenta}{0.05} & \textcolor{teal}{0.00}, \textcolor{blue}{0.18}, \textcolor{magenta}{0.00}& \textcolor{teal}{0.00}, \textcolor{blue}{0.47}, \textcolor{magenta}{0.06}& \textcolor{teal}{0.04}, \textcolor{blue}{0.08}, \textcolor{magenta}{0.28}&\textcolor{teal}{0.38}, \textcolor{blue}{0.00}, \textcolor{magenta}{0.39}\\
    \bottomrule
\end{tabular}}
\end{table*}
\subsection{Centralized Performance of DP Methods}
\label{sec:app-centralizedDP}

Here, we present the results for our centralized experiments, for the privacy budget $\varepsilon = 1.0$. All parameters and procedures are as explained in Section~\ref{sec:exp-centralized}.  

Table~\ref{tab:centralized-all-e1} shows the performance for all DP methods and all datasets. The mean of the metrics (i.e., the \pv, median and survival percentages at $t=\{0.25T_{\max}, 0.5T_{\max}, 0.75T_{\max}\}$) over 100 random runs of our DP algorithm and their calculated $95\%$ confidence interval in parentheses are reported.  

Here we again observe that \dps performs the best, with the means and their confidence intervals always contained within the confidence interval of the original datasets, for all datasets and all metrics. The second best method with respect to \pv is \dpm. However, we still observe the issue of underprediction of survival percentages towards the end point of the studies, in particular in METABRIC and SUPPORT. Our \dpy method, although lacking in \pv performance, keeps a good performance for all the other metrics, especially for GBSG and METABRIC, where all the means of the metrics and their confidence intervals are within the confidence interval of the original non-private KM curves. 

We also show the visual results for one random run of the algorithms in Figure~\ref{fig:centralized-e1}.  

\begin{figure*}[ht]
\centering
\begin{minipage}[l]{0.68\columnwidth}
        \centering
        \includegraphics[width=\linewidth]{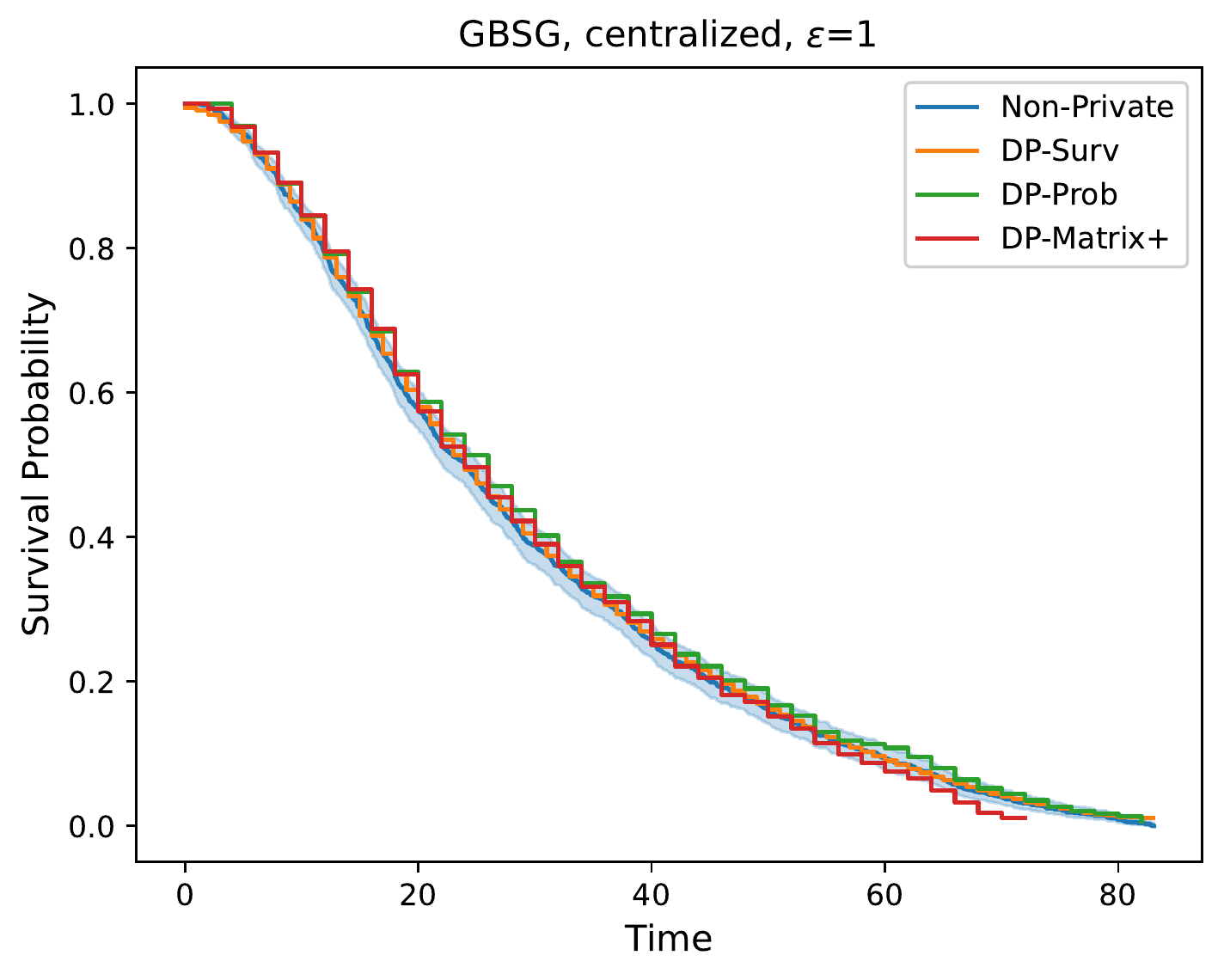}
\end{minipage}
\begin{minipage}[l]{0.68\columnwidth}
        \centering
        \includegraphics[width=\linewidth]{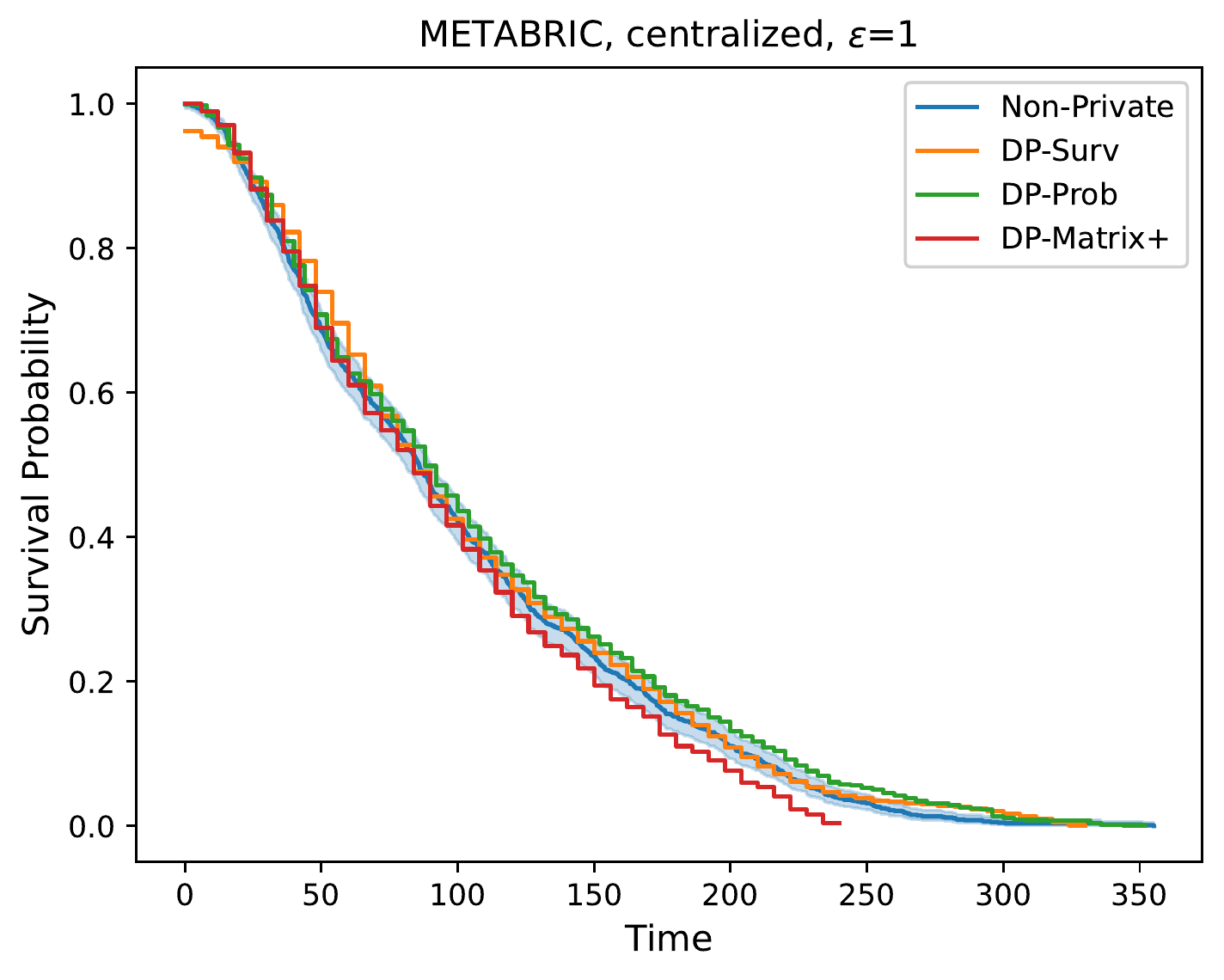}
\end{minipage}
\begin{minipage}[l]{0.68\columnwidth}
        \centering
        \includegraphics[width=\linewidth]{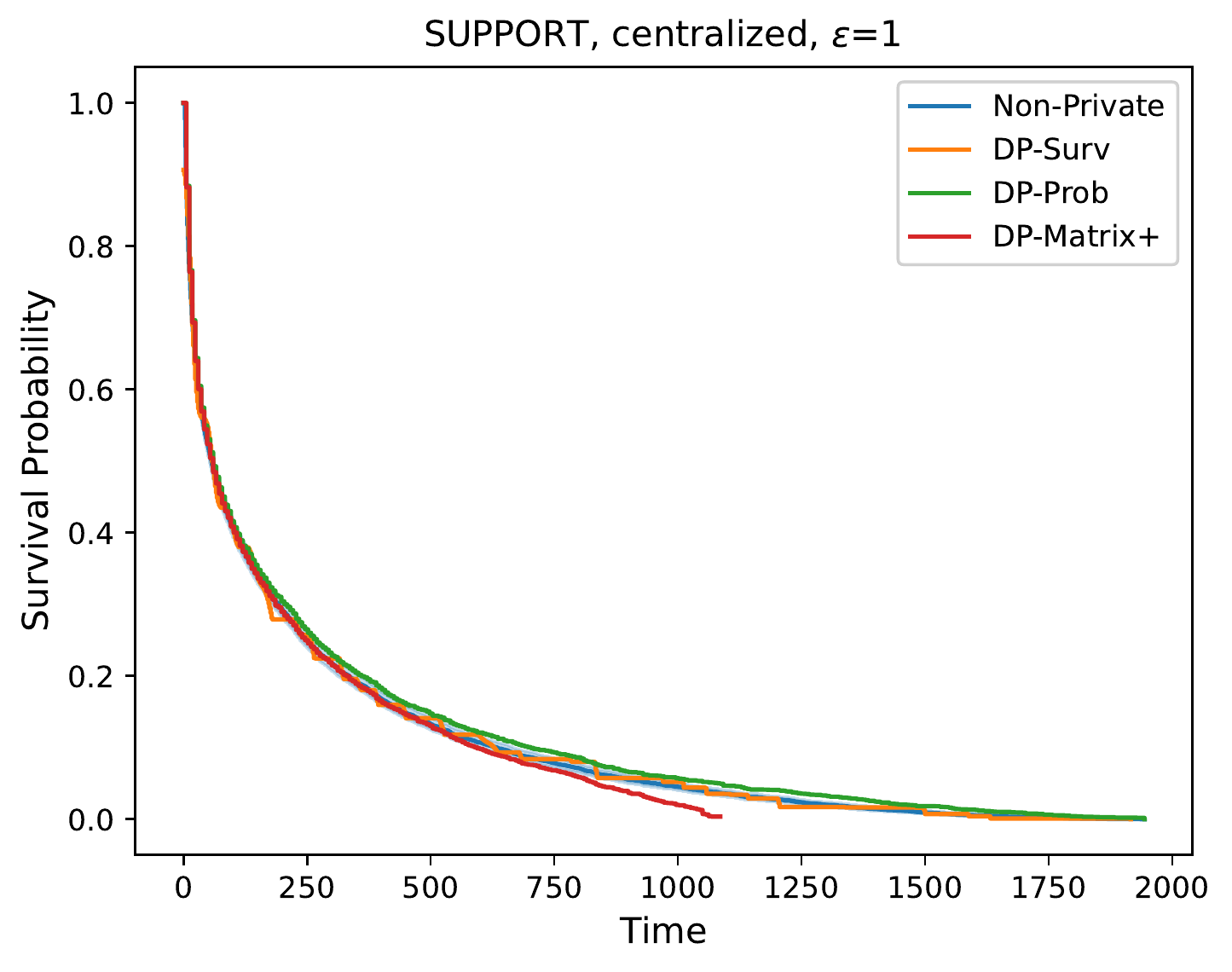}
\end{minipage}
\caption{Comparison of all the DP methods in a centralized setting, for $\varepsilon=1.0$ and one random run of the DP algorithms. The blue shaded region shows the confidence area of the non-private dataset.}
\label{fig:centralized-e1}
\end{figure*}

\begin{table*}[!ht]
\centering
\caption{Performance of the DP methods in the centralized setting for event $e=1$.}
    \hspace*{-0.6cm}\scalebox{1}
    {
    \begin{tabular}{c|c|c|c|c|c|c}
    \toprule
         & & $p$- value& median& $25\%$ $T_{\text{max}}$ & $50\%$ $T_{\text{max}}$ & $75\%$ $T_{\text{max}}$ \\
         \midrule
          \multirow{4}{*}{\rotatebox[origin=c]{90}{\small{GBSG}}}&centralized, non-DP &-& $24 (22; 25)$ & $0.58 (0.55;0.60)$&$0.24 (0.22;0.26) $&$0.08 (0.07;0.11) $\\
         \cline{2-7}
         
          &\dps ($\varepsilon=1$)& $0.39(0.35, 0.44)$ & $24(24, 24)$ & $0.58(0.58, 0.58)$ & $0.25(0.24, 0.25)$& $0.08(0.08, 0.08)$\\
          
         \cline{2-7}
         
          &\dpy ($\varepsilon=1$)& $0.33(0.28, 0.39)$ & $25(25, 25)$ & $0.58(0.57, 0.58)$&$ 0.26(0.26, 0.26)$&$0.08(0.08, 0.09) $\\
         \cline{2-7}
         
          &\dpm ($\varepsilon=1$)& $0.47(0.41, 0.52)$ & $25(24, 25)$ & $0.57(0.57, 0.58)$&$0.25(0.25, 0.25)$&$0.07(0.06, 0.07)$\\
         
          \bottomrule
          \toprule
          
          \multirow{4}{*}{\rotatebox[origin=c]{90}{\small{METABRIC}}}&centralized, non-DP &-& $86 (81; 90)$ & $0.49 (0.46;0.51)$&$0.16 (0.14;0.18) $&$0.02 (0.01;0.03) $\\
         \cline{2-7}

          &\dps ($\varepsilon=1$)& $0.24(0.20, 0.26)$&$84(84, 85)$&$0.49(0.49, 0.49)$ & $0.18(0.18, 0.18)$& $0.02(0.02, 0.02)$\\
          
         \cline{2-7}
         
          &\dpy ($\varepsilon=1$)& $0.07(0.05, 0.09)$ & $89(88, 89)$&$0.49(0.49, 0.50)$&$ 0.17(0.17, 0.18)$&$0.03(0.03, 0.03) $\\

         \cline{2-7}
         
          &\dpm ($\varepsilon=1$)& $0.32(0.15, 0.46)$ & $89(88, 91)$ & $0.51(0.50, 0.51)$&$0.15(0.14, 0.16)$&$0.01(0.00, 0.01)$\\
          \bottomrule
          \toprule
          
          \multirow{4}{*}{\rotatebox[origin=c]{90}{\small{SUPPORT}}}&centralized, non-DP &-& $57 (53; 61)$ & $0.14 (0.13;0.15)$&$0.05 (0.04;0.05) $&$0.01 (0.01;0.01) $\\
         \cline{2-7}

          &\dps ($\varepsilon=1$)& $0.42(0.36, 0.48)$ & $59 (58, 60)$ & $0.14(0.13, 0.14)$ & $0.05(0.05, 0.05)$& $0.01(0.01, 0.01)$\\
          
         \cline{2-7}
         
          &\dpy ($\varepsilon=1$)& $0.00(0.00, 0.00)$ & $60(60, 61)$ & $0.15(0.15, 0.15)$&$ 0.06(0.06, 0.06)$&$0.02(0.02, 0.02) $\\
         \cline{2-7}
          
          &\dpm ($\varepsilon=1$)& $0.37(0.26, 0.47)$ & $60(60, 60)$ & $0.13(0.13, 0.13)$&$0.03(0.02, 0.03)$&$0.00(0.00, 0.00)$\\
          \bottomrule
    \end{tabular}
    }
    \label{tab:centralized-all-e1}
\end{table*}

\subsection{Collaboration}
\label{sec:app-collaboration}

In this section, we expand on our results from Section~\ref{sec:exp-collab} and Section~\ref{sec:exp-collab-uneven}. The setup is exactly the same as explained, but here we show the results for other values of the privacy budget $\varepsilon$. 
\subsubsection{Even Split of Data}
In the paper we analyze the results of collaboration with even splitting of data (where each site has the same amount of data) for the tight privacy budget of $\varepsilon$. Here, we also include the privacy regimes of $\varepsilon=\{3, 5\}$. 

Table~\ref{tab:even-e3} and Table~\ref{tab:even-e5} summarize the results for our DP methods across 10 collaborating sites. The $95\%$ confidence interval of the mean of the metrics over 100 random runs of the DP algorithm is shown in parentheses. 

\myparagraph{Performance of \dps-Based Methods} Our \dps method performs really well for these privacy regimes, as expected. We observe that for values of $\varepsilon$ we have a mean \pv that is above the significance level of 0.05 for all datasets. Similar to the tighter privacy regime of $\varepsilon=1$, here we also observe consistent performance between multiple runs of the algorithm and also between different paths A, B and C. The estimated private survival percentages always fall within the confidence interval of the non-DP datasets. The estimated median and its confidence interval is also always accurate. Here we see that raising the value of $\varepsilon$ to 3 is already enough to solve the issue with median estimation for SUPPORT which is a challenging dataset.

\myparagraph{Performance of \dpy-Based Methods} For these privacy budgets, the \dpy-based paths still do not work as well as the \dps method according to \pv.  

In this lower privacy regime, however, we see more stable behavior of \dpy-based paths, among different runs of the algorithm as well as between different paths D, E, and F. We still observe the problem of overestimation of survival percentages, as we described in the centralized experiments in Section~\ref{sec:exp-centralized}.

\myparagraph{Performance of \dpm-Based Methods} For $\varepsilon=3$, the \pv of this path falls below the significance level of 0.05 for GBSG and SUPPORT and for $\varepsilon=5$, it still falls below the significance level for SUPPORT. It shows stable behavior between multiple runs. However, it still suffers from the problem of under estimating the survival percentages, especially towards the end of the study, for all datasets and for both privacy values, as explained in Section~\ref{sec:exp-centralized}.

\begin{table*}[!ht]
    \centering
        \caption{Collaboration with even data split for $e=1$}
    \scalebox{0.9}{
    \begin{tabular}{c|c|c|c|c|c|c|c}
    \toprule
         & & &$p$- value& median survival time& $25\%$ $T_{\text{max}}$ & $50\%$ $T_{\text{max}}$ & $75\%$ $T_{\text{max}}$ \\
         \midrule
          \multirow{7}{*}{\rotatebox[origin=c]{90}{GBSG}}& centralized, non-private  & &-& $24 (22; 25)$ & $0.58 (0.55;0.60)$&$0.24 (0.22;0.26) $&$0.08 (0.07;0.10) $\\
         \cline{2-8}
         
         && pooled&$0.29(0.25, 0.33)$ & $24(24, 24)$ & $0.58(0.58, 0.58)$ & $0.24(0.24, 0.25)$ &$0.08(0.08, 0.09)$ \\
         &\dps ($\varepsilon=3$)&Averaged $\s'$& $0.25(0.21, 0.29)$& $24(24, 24)$&$0.58(0.58, 0.58)$&$0.25(0.24, 0.25)$&$0.08(0.08, 0.08)$\\
         &&Averaged $\yh'$&$0.24(0.20, 0.27)$&$24(24, 24)$&$0.58(0.58, 0.58)$&$0.25(0.24, 0.25)$&$0.08(0.08, 0.09)$\\
         \cline{2-8}
          
         && pooled&$0.00(0.00, 0.00)$ & $26(26, 26)$ & $0.60(0.60, 0.60)$ & $0.28(0.28, 0.28)$ &$0.11(0.11, 0.11)$ \\
         &\dpy ($\varepsilon=3$)&Averaged $\s'$&$0.00(0.00, 0.00)$&$26(26, 26)$&$0.60(0.60, 0.60)$&$0.28(0.28, 0.28)$&$0.11(0.11, 0.11)$\\
         &&Averaged $\yh'$&$0.00(0.00, 0.00)$&$26(26, 26)$&$0.60(0.60, 0.60)$&$0.28(0.28, 0.28)$&$0.11(0.11, 0.12)$\\
         \cline{2-8}
         
         &\dpm ($\varepsilon=3$)&pooled&$0.01(0.0, 0.01)$&$23(23, 23)$&$0.56(0.56, 0.57)$&$0.18(0.18, 0.18)$&$0.02(0.02, 0.02)$\\

    \bottomrule
    \toprule
         \multirow{7}{*}{\rotatebox[origin=c]{90}{METABRIC}}&centralized, non-private  &&-& $86 (81; 90)$ & $0.49 (0.46;0.51)$&$0.16 (0.14;0.18) $&$0.02 (0.01;0.03) $\\
         \cline{2-8}
          
         && pooled&$0.24(0.2, 0.28)$ & $84(84, 84)$ & $0.49(0.48, 0.49)$ & $0.17(0.17, 0.18)$ &$0.02(0.02, 0.02)$ \\
         &\dps ($\varepsilon=3$)&Averaged $\s'$&$0.14(0.12, 0.16)$&$85(84, 85)$&$0.49(0.49, 0.49)$&$0.18(0.18, 0.18)$&$0.03(0.02, 0.03)$\\
         &&Averaged $\yh'$&$0.12(0.10, 0.14)$&$84(84, 84)$&$0.49(0.49, 0.49)$&$0.18(0.18, 0.18)$&$0.03(0.02, 0.03)$\\
         \cline{2-8}
          
         && pooled&$0(0.00, 0.00)$ & $92(92, 93)$ & $0.54(0.53, 0.54)$ & $0.21(0.21, 0.21)$ &$0.05(0.05, 0.06)$ \\
         &\dpy ($\varepsilon=3$)&Averaged $\s'$&$0(0.00, 0.00)$&$93(93, 94)$&$0.54(0.54, 0.54)$&$0.21(0.21, 0.21)$&$0.06(0.06, 0.06)$\\
         &&Averaged $\yh'$&$0(0.00, 0.00)$&$93(93, 94)$&$0.54(0.54, 0.54)$&$0.21(0.21, 0.21)$&$0.06(0.06, 0.06)$\\
         \cline{2-8}
         
         &\dpm ($\varepsilon=3$)&pooled&$0.12(0.07, 0.15)$&$87(86, 87)$&$0.50(0.50, 0.50)$&$0.11(0.11, 0.11)$&$0.01(0.01, 0.01)$\\

    \bottomrule
    \toprule
          \multirow{7}{*}{\rotatebox[origin=c]{90}{SUPPORT}}&centralized, non-private  &&-& $57 (53; 61)$ & $0.14 (0.13;0.15)$&$0.05 (0.04;0.05) $&$0.01 (0.01;0.01) $\\
         \cline{2-8}
         
         && pooled&$0.20(0.15, 0.24)$ & $60(58, 60)$ & $0.14(0.14, 0.14)$ & $0.05(0.05, 0.05)$ &$0.02(0.01, 0.02)$ \\
         &\dps ($\varepsilon=3$)&Averaged $\s'$&$0.21(0.16, 0.25)$&$59(58, 60)$&$0.14(0.14, 0.14)$&$0.05(0.05, 0.05)$&$0.01(0.01, 0.01)$\\
         &&Averaged $\yh'$&$0.18(0.14, 0.22)$&$59(58, 60)$&$0.14(0.14, 0.14)$&$0.05(0.05, 0.05)$&$0.01(0.01, 0.02)$\\
         \cline{2-8}
         
         && pooled&$0(0.00, 0.00)$ & $167(166, 168)$ & $0.33(0.33, 0.33)$ & $0.19(0.19, 0.19)$ &$0.09(0.09, 0.09)$ \\
         &\dpy ($\varepsilon=3$)&Averaged $\s'$&$0(0.00, 0.00)$&$196(195, 197)$&$0.36(0.35, 0.36)$&$0.21(0.21, 0.21)$&$0.1(0.1, 0.1)$\\
         &&Averaged $\yh'$&$0(0.00, 0.00)$&$196(195, 197)$&$0.36(0.35, 0.36)$&$0.21(0.21, 0.21)$&$0.1(0.1, 0.1)$\\
         \cline{2-8}
         
         &\dpm ($\varepsilon=3$)&pooled&$0(0.00, 0.00)$&$57(56, 57)$&$0.08(0.08, 0.08)$&$0.00(0.00, 0.00)$&$0.00(0.00, 0.00)$\\
         
          \bottomrule
    \end{tabular}
    }
    \label{tab:even-e3}
\end{table*}

\begin{table*}[!ht]
    \centering
        \caption{Collaboration with even data split for $e=1$}
    \scalebox{0.9}{
    \begin{tabular}{c|c|c|c|c|c|c|c}
    \toprule
         & & &$p$- value& median survival time& $25\%$ $T_{\text{max}}$ & $50\%$ $T_{\text{max}}$ & $75\%$ $T_{\text{max}}$ \\
         \midrule
          \multirow{7}{*}{\rotatebox[origin=c]{90}{GBSG}}& centralized, non-private  & &-& $24 (22; 25)$ & $0.58 (0.55;0.60)$&$0.24 (0.22;0.26) $&$0.08 (0.07;0.10) $\\
         \cline{2-8}
         
         && pooled&$0.33(0.29, 0.36)$ & $24(24, 24)$ & $0.58(0.58, 0.58)$ & $0.25(0.24, 0.25)$ &$0.08(0.08, 0.08)$ \\
         &\dps ($\varepsilon=5$)&Averaged $\s'$& $0.28(0.26, 0.31)$& $24(24, 24)$&$0.58(0.58, 0.58)$&$0.24(0.24, 0.25)$&$0.08(0.08, 0.08)$\\
         &&Averaged $\yh'$&$0.27(0.24, 0.29)$&$24(24, 24)$&$0.58(0.58, 0.58)$&$0.25(0.24, 0.25)$&$0.08(0.08, 0.08)$\\
         \cline{2-8}
          
         && pooled&$0.03(0.02, 0.04)$ & $25(25, 25)$ & $0.59(0.59, 0.59)$ & $0.26(0.26, 0.26)$ &$0.10(0.10, 0.10)$ \\
         &\dpy ($\varepsilon=5$)&Averaged $\s'$&$0.01(0.01, 0.02)$&$25(25, 25)$&$0.59(0.59, 0.59)$&$0.26(0.26, 0.27)$&$0.10(0.10, 0.10)$\\
         &&Averaged $\yh'$&$0.01(0.01, 0.02)$&$25(25, 25)$&$0.59(0.59, 0.59)$&$0.26(0.26, 0.27)$&$0.10(0.10, 0.10)$\\
         \cline{2-8}
         
         &\dpm ($\varepsilon=5$)&pooled&$0.21(0.16, 0.26)$&$24(24, 24)$&$0.57(0.57, 0.57)$&$0.22(0.22, 0.22)$&$0.05(0.05, 0.05)$\\

    \bottomrule
    \toprule
         \multirow{7}{*}{\rotatebox[origin=c]{90}{METABRIC}}&centralized, non-private  &&-& $86 (81; 90)$ & $0.49 (0.46;0.51)$&$0.16 (0.14;0.18) $&$0.02 (0.01;0.03) $\\
         \cline{2-8}
          
         && pooled&$0.48(0.43, 0.53)$ & $84(84, 84)$ & $0.48(0.48, 0.49)$ & $0.17(0.17, 0.17)$ &$0.02(0.01, 0.02)$ \\
         &\dps ($\varepsilon=5$)&Averaged $\s'$&$0.16(0.14, 0.18)$&$84(84, 84)$&$0.49(0.49, 0.49)$&$0.18(0.18, 0.18)$&$0.02(0.02, 0.02)$\\
         &&Averaged $\yh'$&$0.16(0.14, 0.18)$&$84(84, 84)$&$0.49(0.49, 0.49)$&$0.18(0.18, 0.18)$&$0.02(0.02, 0.02)$\\
         \cline{2-8}
          
         && pooled&$0(0.00, 0.00)$ & $90(90, 90)$ & $0.52(0.52, 0.53)$ & $0.19(0.18, 0.19)$ &$0.04(0.04, 0.04)$ \\
         &\dpy ($\varepsilon=5$)&Averaged $\s'$&$0(0.00, 0.00)$&$90(90, 90)$&$0.53(0.53, 0.53)$&$0.19(0.19, 0.19)$&$0.04(0.04, 0.04)$\\
         &&Averaged $\yh'$&$0(0.00, 0.00)$&$90(90, 90)$&$0.53(0.53, 0.53)$&$0.19(0.19, 0.19)$&$0.04(0.04, 0.04)$\\
         \cline{2-8}
         
         &\dpm ($\varepsilon=5$)&pooled&$0.46(0.40, 0.51)$&$89(89, 90)$&$0.51(0.51, 0.51)$&$0.14(0.14, 0.14)$&$0.01(0.01, 0.01)$\\

    \bottomrule
    \toprule
          \multirow{7}{*}{\rotatebox[origin=c]{90}{SUPPORT}}&centralized, non-private  &&-& $57 (53; 61)$ & $0.14 (0.13;0.15)$&$0.05 (0.04;0.05) $&$0.01 (0.01;0.01) $\\
         \cline{2-8}
         
         && pooled&$0.32(0.27, 0.36)$ & $59(59, 60)$ & $0.14(0.14, 0.14)$ & $0.05(0.05, 0.05)$ &$0.01(0.01, 0.01)$ \\
         &\dps ($\varepsilon=5$)&Averaged $\s'$&$0.34(0.29, 0.39)$&$59(58, 59)$&$0.14(0.14, 0.14)$&$0.05(0.05, 0.05)$&$0.01(0.01, 0.01)$\\
         &&Averaged $\yh'$&$0.30(0.25, 0.34)$&$59(59, 60)$&$0.14(0.14, 0.14)$&$0.05(0.05, 0.05)$&$0.01(0.01, 0.01)$\\
         \cline{2-8}
         
         && pooled&$0(0.00, 0.00)$ & $101(100, 102)$ & $0.25(0.25, 0.25)$ & $0.13(0.13, 0.13)$ &$0.06(0.06, 0.06)$ \\
         &\dpy ($\varepsilon=5$)&Averaged $\s'$&$0(0.00, 0.00)$&$122(122, 123)$&$0.29(0.29, 0.29)$&$0.16(0.16, 0.16)$&$0.07(0.07, 0.07)$\\
         &&Averaged $\yh'$&$0(0.00, 0.00)$&$121(121, 122)$&$0.28(0.28, 0.29)$&$0.16(0.16, 0.16)$&$0.07(0.07, 0.07)$\\
         \cline{2-8}
         
         &\dpm ($\varepsilon=5$)&pooled&$0(0.00, 0.00)$&$57(57, 57)$&$0.11(0.11, 0.11)$&$0.01(0.01, 0.01)$&$0.00(0.00, 0.00)$\\
         
          \bottomrule
    \end{tabular}
    }
    \label{tab:even-e5}
\end{table*}

\subsubsection{Uneven Split of Data}

Next, we proceed to the collaborative setting with uneven split of the datasets. Here again, the setup is identical as described in Section~\ref{sec:exp-collab-uneven}, but we look at two new privacy values $\varepsilon=\{3, 5\}$. 

Table~\ref{tab:uneven-e3} and Table~\ref{tab:uneven-e5} show the results of our experiments with \dps-based paths A, B, and C. We clearly see consistent and really good results for all these paths. Again, we see very stable behavior of our method across multiple runs of the algorithm with very tight confidence intervals for both median and survival percentages. We also see that paths A, B, and C show very similar behavior, solidifying our prior conclusion that our averaging algorithms work well and give freedom of choice to data collectors. 

To give a visual intuition about these collaborative settings, we provide Figures~\ref{fig:collabe1},~\ref{fig:collabe3} and~\ref{fig:collabe05}, where we show one random run of the \dps-based path B, for each dataset and each data split, for different values of the privacy budget $\varepsilon$. We also depict the non-DP local KM estimators of each client if they used only the local data. We report each site's as well as our methods and also the centralized, non-private dataset's median and \pv with m and p in the plot. For uneven splits of datasets, we also indicate the median and \pv of the minority site (i.e., the site receiving either $5\%$ or $50\%$ of the data). We observe how closely these private estimators mimic the behavior of the centralized dataset and we also observe that in many cases, if sites depend on only their local datasets, they will overestimate or underestimate survival percentages and the median. This again proves that with these private and collaborative paths, there can always be an incentive to participate and learn a better KM estimator without compromising the privacy of the local dataset.

\begin{table*}[ht]
    \centering
        \caption{Collaboration with uneven data split with one site receiving either $50\%$ or $5\%$ of all of the data, for $e=1$ and $\varepsilon=3$}
    \scalebox{0.85}{
    \begin{tabular}{c|c|c|c|c|c|c|c}
    \toprule
         & & &$p$- value& median survival time& $25\%$ $T_{\text{max}}$ & $50\%$ $T_{\text{max}}$ & $75\%$ $T_{\text{max}}$ \\
         \hline
          \multirow{7}{*}{\rotatebox[origin=c]{90}{GBSG}}& centralized, non-private  & &-& $24 (22; 25)$ & $0.58 (0.55;0.60)$&$0.24 (0.22;0.26) $&$0.08 (0.07;0.10) $\\
         \cline{2-8}
         
         && pooled&$0.36(0.31, 0.41)$ & $24(24, 24)$ & $0.58(0.58, 0.58)$ & $0.24(0.24, 0.25)$ &$0.08(0.08, 0.08)$ \\
         &\dps ($\varepsilon=3$)&Averaged $\s'$& $0.29(0.25, 0.32)$& $24(24, 24)$&$0.58(0.58, 0.58)$&$0.24(0.24, 0.25)$&$0.08(0.08, 0.08)$\\
         &minority has $50\%$&Averaged $\yh'$&$0.28(0.24, 0.31)$&$24(24, 24)$&$0.58(0.58, 0.58)$&$0.25(0.25, 0.25)$&$0.08(0.08, 0.09)$\\
         \cline{2-8}
         
         && pooled&$0.30(0.25, 0.34)$ & $24(24, 24)$ & $0.58(0.58, 0.58)$ & $0.25(0.24, 0.25)$ &$0.08(0.08, 0.08)$ \\
         &\dps ($\varepsilon=3$)&Averaged $\s'$& $0.26(0.22, 0.31)$& $24(24, 24)$&$0.58(0.58, 0.58)$&$0.25(0.24, 0.25)$&$0.08(0.08, 0.09)$\\
         &minority has $5\%$&Averaged $\yh'$&$0.23(0.20, 0.27)$&$24(24, 24)$&$0.58(0.58, 0.58)$&$0.25(0.24, 0.25)$&$0.08(0.08, 0.09)$\\

    \bottomrule
    \toprule
         \multirow{7}{*}{\rotatebox[origin=c]{90}{METABRIC}}&centralized, non-private  &&-& $86 (81; 90)$ & $0.49 (0.46;0.51)$&$0.16 (0.14;0.18) $&$0.02 (0.01;0.03) $\\
         \cline{2-8}
          
         && pooled&$0.25(0.20, 0.29)$ & $84(84, 84)$ & $0.49(0.48, 0.49)$ & $0.18(0.17, 0.18)$ &$0.02(0.02, 0.02)$ \\
         &\dps ($\varepsilon=3$)&Averaged $\s'$&$0.13(0.11, 0.15)$&$85(84, 85)$&$0.49(0.49, 0.49)$&$0.18(0.18, 0.18)$&$0.03(0.02, 0.03)$\\
         &minority has $50\%$&Averaged $\yh'$&$0.13(0.11, 0.15)$&$84(84, 85)$&$0.49(0.49, 0.49)$&$0.18(0.18, 0.18)$&$0.03(0.02, 0.03)$\\
         \cline{2-8}
         
         && pooled&$0.26(0.22, 0.30)$ & $84(84, 85)$ & $0.49(0.48, 0.49)$ & $0.17(0.17, 0.18)$ &$0.02(0.02, 0.02)$ \\
         &\dps ($\varepsilon=3$)&Averaged $\s'$&$0.12(0.10, 0.14)$&$84(84, 85)$&$0.49(0.49, 0.49)$&$0.18(0.18, 0.18)$&$0.03(0.02, 0.03)$\\
         &minority has $5\%$&Averaged $\yh'$&$0.14(0.11, 0.17)$&$85(84, 85)$&$0.49(0.49, 0.49)$&$0.18(0.18, 0.18)$&$0.02(0.02, 0.03)$\\

    \bottomrule
    \toprule
          \multirow{7}{*}{\rotatebox[origin=c]{90}{SUPPORT}}&centralized, non-private  &&-& $57 (53; 61)$ & $0.14 (0.13;0.15)$&$0.05 (0.04;0.05) $&$0.01 (0.01;0.01) $\\
         \cline{2-8}
         
         && pooled&$0.18(0.14, 0.22)$ & $59(59, 60)$ & $0.14(0.14, 0.14)$ & $0.05(0.05, 0.05)$ &$0.01(0.01, 0.02)$ \\
         &\dps ($\varepsilon=3$)&Averaged $\s'$&$0.16(0.12, 0.20)$&$60(59, 61)$&$0.14(0.14, 0.14)$&$0.05(0.05, 0.05)$&$0.02(0.01, 0.02)$\\
         &minority has $50\%$&Averaged $\yh'$&$0.18(0.13, 0.22)$&$60(59, 61)$&$0.14(0.14, 0.14)$&$0.05(0.05, 0.05)$&$0.02(0.01, 0.02)$\\
         \cline{2-8}
         
         && pooled&$0.20(0.15, 0.24)$ & $60(60, 61)$ & $0.14(0.14, 0.14)$ & $0.05(0.05, 0.05)$ &$0.01(0.01, 0.02)$ \\
         &\dps ($\varepsilon=3$)&Averaged $\s'$&$0.19(0.15, 0.23)$&$60(59, 61)$&$0.14(0.14, 0.14)$&$0.05(0.05, 0.05)$&$0.01(0.01, 0.01)$\\
         &minority has $5\%$&Averaged $\yh'$&$0.23(0.18, 0.27)$&$60(59, 61)$&$0.14(0.14, 0.14)$&$0.05(0.05, 0.05)$&$0.01(0.01, 0.02)$\\
         
          \bottomrule
    \end{tabular}
    }
    \label{tab:uneven-e3}
\end{table*}

\begin{table*}[ht]
    \centering
        \caption{Collaboration with uneven data split with one site receiving either $50\%$ or $5\%$ of all of the data, for $e=1$ and $\varepsilon=5$}
    \scalebox{0.85}{
    \begin{tabular}{c|c|c|c|c|c|c|c}
    \toprule
         & & &$p$- value& median survival time& $25\%$ $T_{\text{max}}$ & $50\%$ $T_{\text{max}}$ & $75\%$ $T_{\text{max}}$ \\
         \hline
          \multirow{7}{*}{\rotatebox[origin=c]{90}{GBSG}}& centralized, non-private  & &-& $24 (22; 25)$ & $0.58 (0.55;0.60)$&$0.24 (0.22;0.26) $&$0.08 (0.07;0.10) $\\
         \cline{2-8}
         
         && pooled&$0.50(0.46, 0.55)$ & $24(24, 24)$ & $0.58(0.58, 0.58)$ & $0.24(0.24, 0.24)$ &$0.08(0.08, 0.08)$ \\
         &\dps ($\varepsilon=5$)&Averaged $\s'$& $0.34(0.33, 0.36)$& $24(24, 24)$&$0.58(0.58, 0.58)$&$0.24(0.24, 0.24)$&$0.08(0.08, 0.08)$\\
         &minority has $50\%$&Averaged $\yh'$&$0.35(0.33, 0.36)$&$24(24, 24)$&$0.58(0.58, 0.58)$&$0.24(0.24, 0.25)$&$0.08(0.08, 0.09)$\\
         \cline{2-8}
         && pooled&$0.38(0.34, 0.42)$ & $24(24, 24)$ & $0.58(0.58, 0.58)$ & $0.24(0.24, 0.25)$ &$0.08(0.08, 0.08)$ \\
         &\dps ($\varepsilon=5$)&Averaged $\s'$& $0.29(0.26, 0.31)$& $24(24, 24)$&$0.58(0.58, 0.58)$&$0.24(0.24, 0.25)$&$0.08(0.08, 0.08)$\\
         &minority has $5\%$&Averaged $\yh'$&$0.28(0.26, 0.31)$&$24(24, 24)$&$0.58(0.58, 0.58)$&$0.25(0.25, 0.25)$&$0.08(0.08, 0.08)$\\

    \bottomrule
    \toprule
         \multirow{7}{*}{\rotatebox[origin=c]{90}{METABRIC}}&centralized, non-private  &&-& $86 (81; 90)$ & $0.49 (0.46;0.51)$&$0.16 (0.14;0.18) $&$0.02 (0.01;0.03) $\\
         \cline{2-8}
          
         && pooled&$0.37(0.32, 0.42)$ & $84(84, 84)$ & $0.49(0.48, 0.49)$ & $0.17(0.17, 0.17)$ &$0.02(0.02, 0.02)$ \\
         &\dps ($\varepsilon=5$)&Averaged $\s'$&$0.17(0.14, 0.19)$&$84(84, 84)$&$0.49(0.49, 0.49)$&$0.18(0.18, 0.18)$&$0.02(0.02, 0.02)$\\
         &minority has $50\%$&Averaged $\yh'$&$0.19(0.18, 0.20)$&$84(84, 84)$&$0.49(0.49, 0.49)$&$0.18(0.18, 0.18)$&$0.02(0.02, 0.02)$\\
         \cline{2-8}
         && pooled&$0.46(0.41, 0.50)$ & $84(84, 84)$ & $0.48(0.48, 0.49)$ & $0.17(0.17, 0.17)$ &$0.02(0.01, 0.02)$ \\
         &\dps ($\varepsilon=5$)&Averaged $\s'$&$0.15(0.13, 0.16)$&$84(84, 84)$&$0.49(0.49, 0.49)$&$0.18(0.18, 0.18)$&$0.02(0.02, 0.02)$\\
         &minority has $5\%$&Averaged $\yh'$&$0.15(0.14, 0.17)$&$84(84, 84)$&$0.49(0.49, 0.49)$&$0.18(0.18, 0.18)$&$0.02(0.02, 0.02)$\\

    \bottomrule
    \toprule
          \multirow{7}{*}{\rotatebox[origin=c]{90}{SUPPORT}}&centralized, non-private  &&-& $57 (53; 61)$ & $0.14 (0.13;0.15)$&$0.05 (0.04;0.05) $&$0.01 (0.01;0.01) $\\
         \cline{2-8}
         
         && pooled&$0.28(0.24, 0.32)$ & $59(59, 60)$ & $0.14(0.14, 0.14)$ & $0.05(0.05, 0.05)$ &$0.01(0.01, 0.01)$ \\
         &\dps ($\varepsilon=5$)&Averaged $\s'$&$0.33(0.28, 0.37)$&$59(58, 59)$&$0.14(0.14, 0.14)$&$0.05(0.05, 0.05)$&$0.01(0.01, 0.01)$\\
         &minority has $50\%$&Averaged $\yh'$&$0.30(0.26, 0.34)$&$59(58, 59)$&$0.14(0.14, 0.14)$&$0.05(0.05, 0.05)$&$0.01(0.01, 0.01)$\\
         \cline{2-8}
         && pooled&$0.29(0.25, 0.33)$ & $59(58, 59)$ & $0.14(0.14, 0.14)$ & $0.05(0.05, 0.05)$ &$0.01(0.01, 0.01)$ \\
         &\dps ($\varepsilon=5$)&Averaged $\s'$&$0.31(0.26, 0.35)$&$58(58, 59)$&$0.14(0.14, 0.14)$&$0.05(0.05, 0.05)$&$0.01(0.01, 0.01)$\\
         &minority has $5\%$&Averaged $\yh'$&$0.27(0.23, 0.30)$&$59(59, 59)$&$0.14(0.14, 0.14)$&$0.05(0.05, 0.05)$&$0.01(0.01, 0.01)$\\
         
          \bottomrule
    \end{tabular}
    }
    \label{tab:uneven-e5}
\end{table*}

\begin{figure*}[!ht]
\centering
\hspace*{-.5cm}
\begin{minipage}[l]{0.73\columnwidth}
        \centering
        \includegraphics[width=\linewidth]{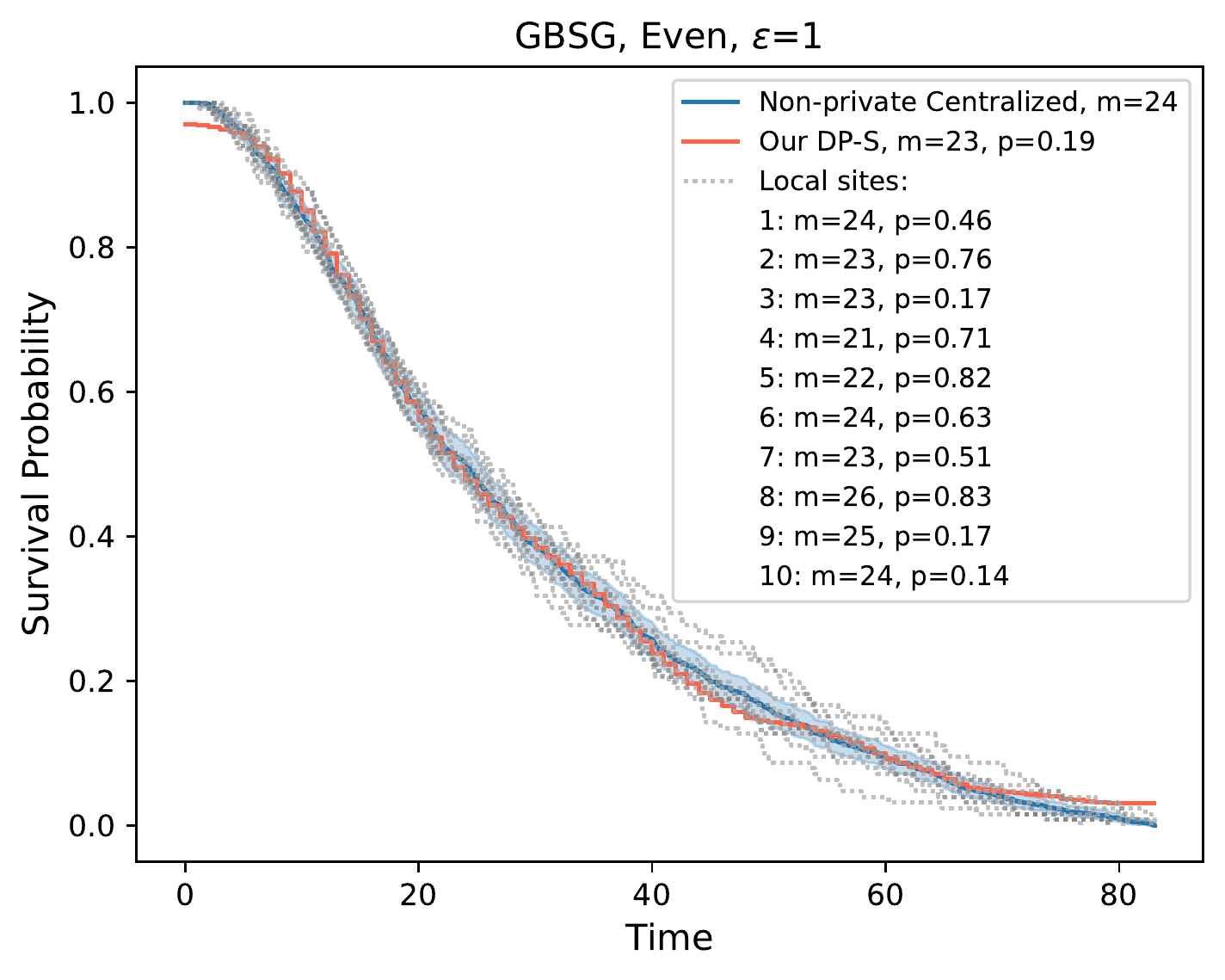}
\end{minipage}%
\hspace*{-.2cm}
\begin{minipage}[l]{0.73\columnwidth}
        \centering
        \includegraphics[width=\linewidth]{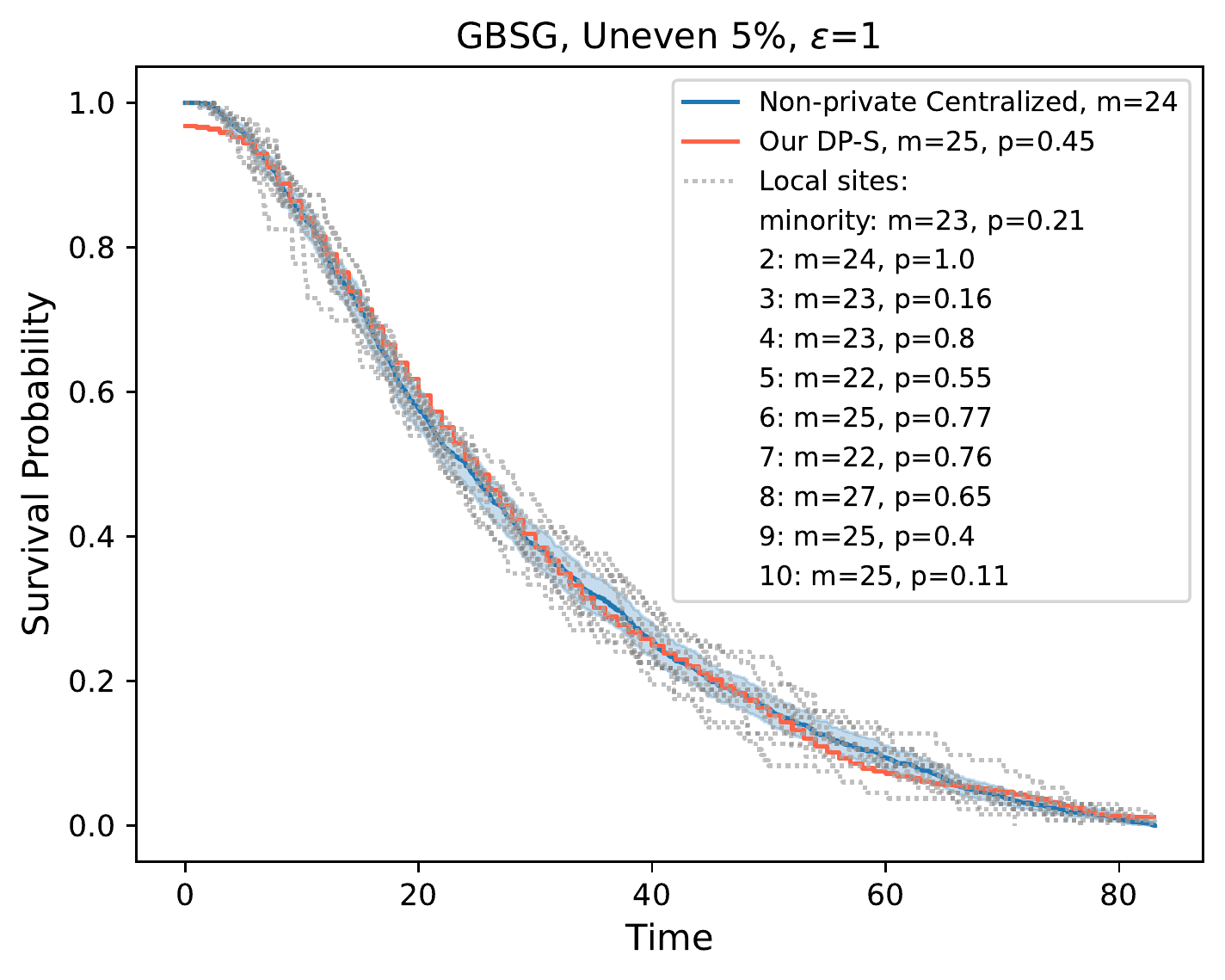}
\end{minipage}%
\hspace*{-.2cm}
\begin{minipage}[l]{0.73\columnwidth}
        \centering
        \includegraphics[width=\linewidth]{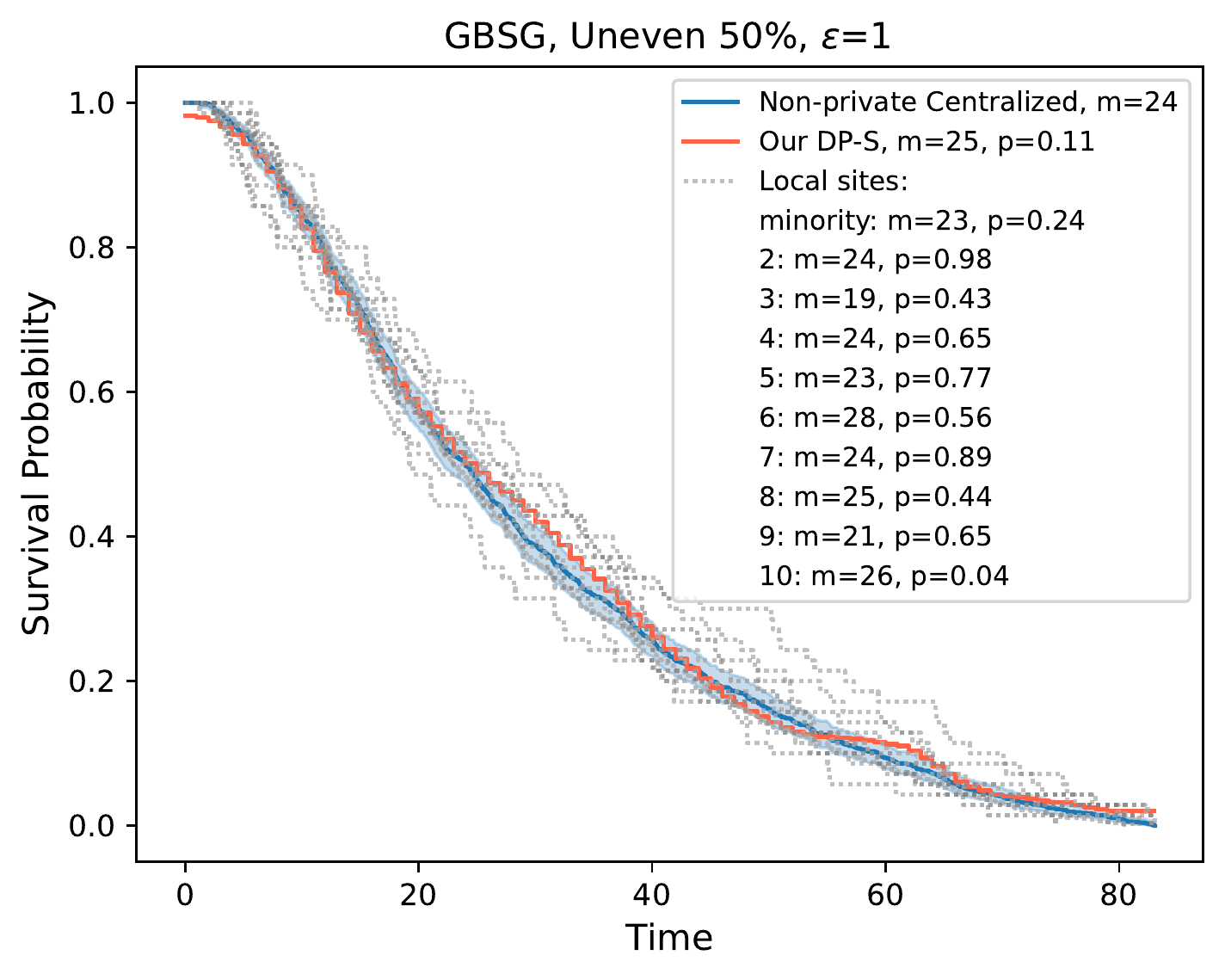}
\end{minipage}
    
\hspace*{-.5cm}
\begin{minipage}[l]{0.73\columnwidth}
        \centering
        \includegraphics[width=\linewidth]{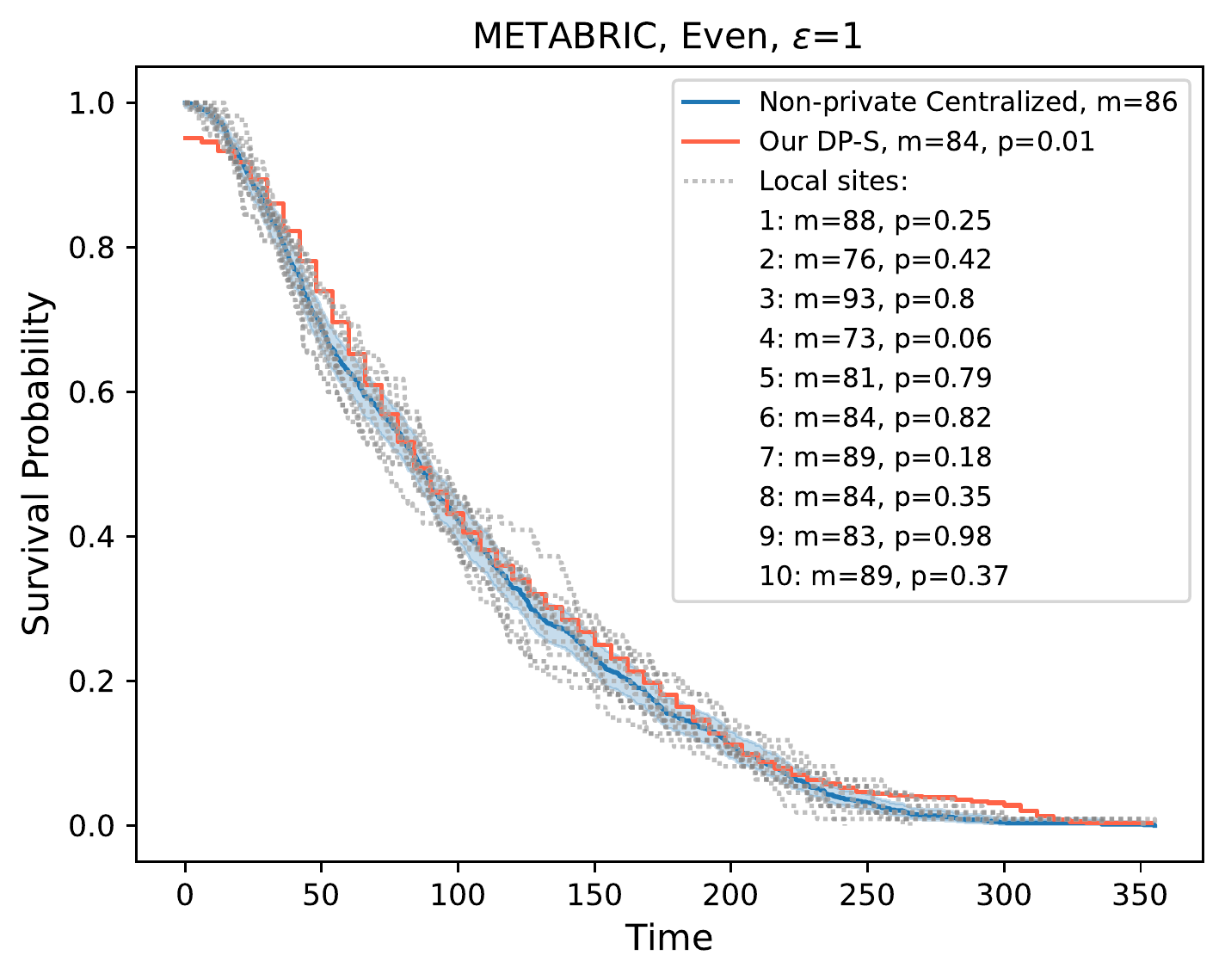}
\end{minipage}%
\hspace*{-.2cm}
\begin{minipage}[l]{0.73\columnwidth}
        \centering
        \includegraphics[width=\linewidth]{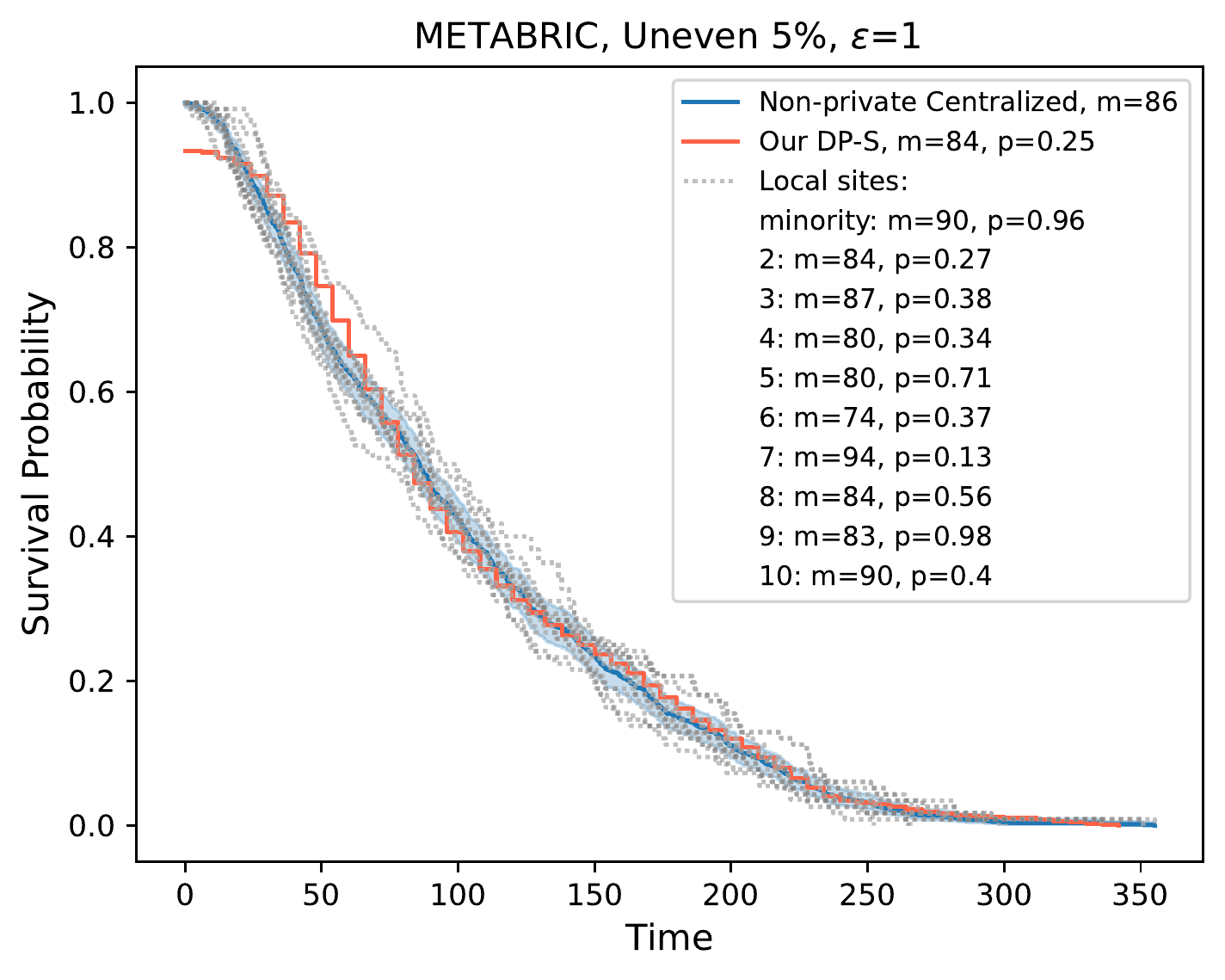}
\end{minipage}%
\hspace*{-.2cm}
\begin{minipage}[l]{0.73\columnwidth}
        \centering
        \includegraphics[width=\linewidth]{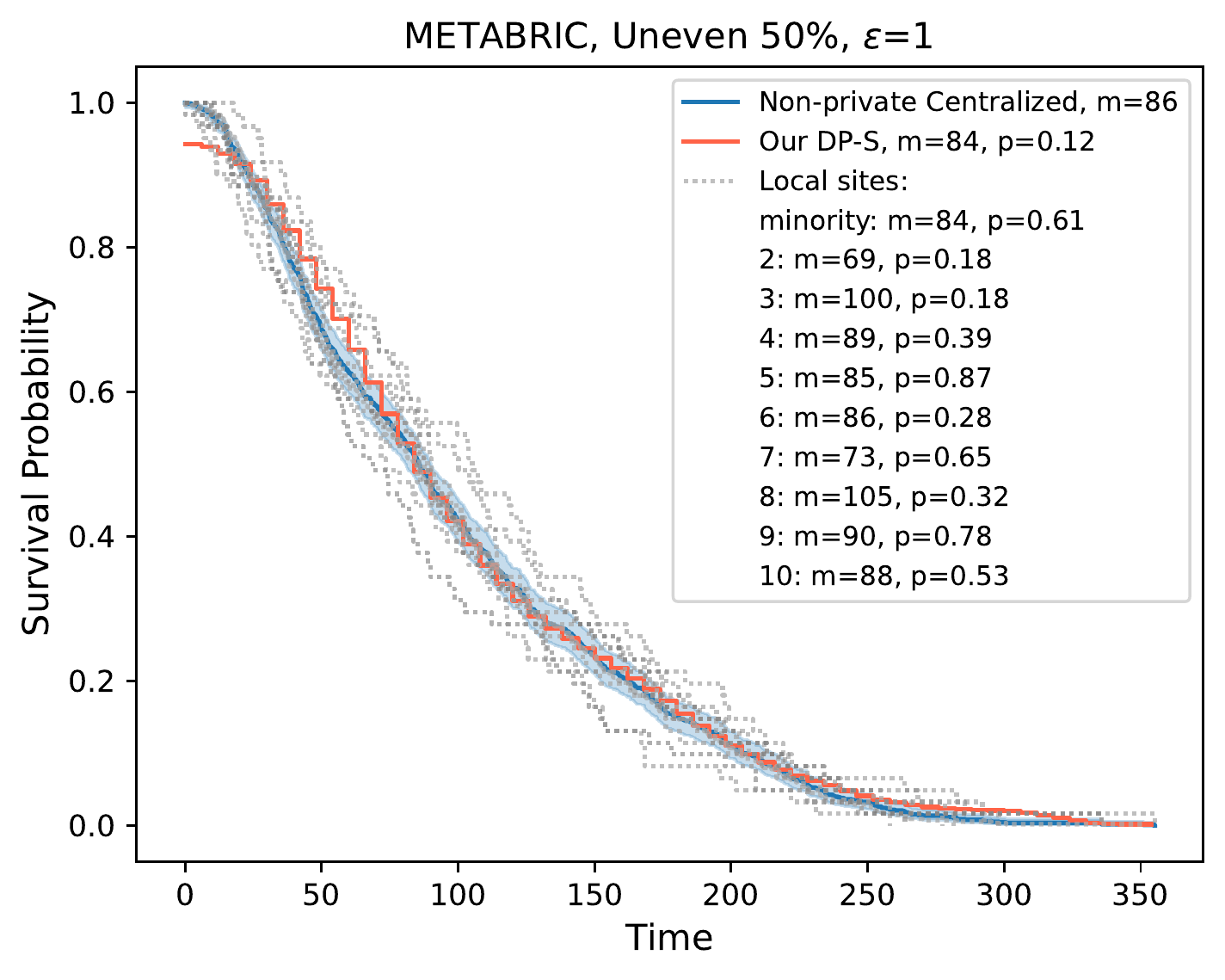}
\end{minipage}

\hspace*{-.5cm}
\begin{minipage}[l]{0.73\columnwidth}
        \centering
        \includegraphics[width=\linewidth]{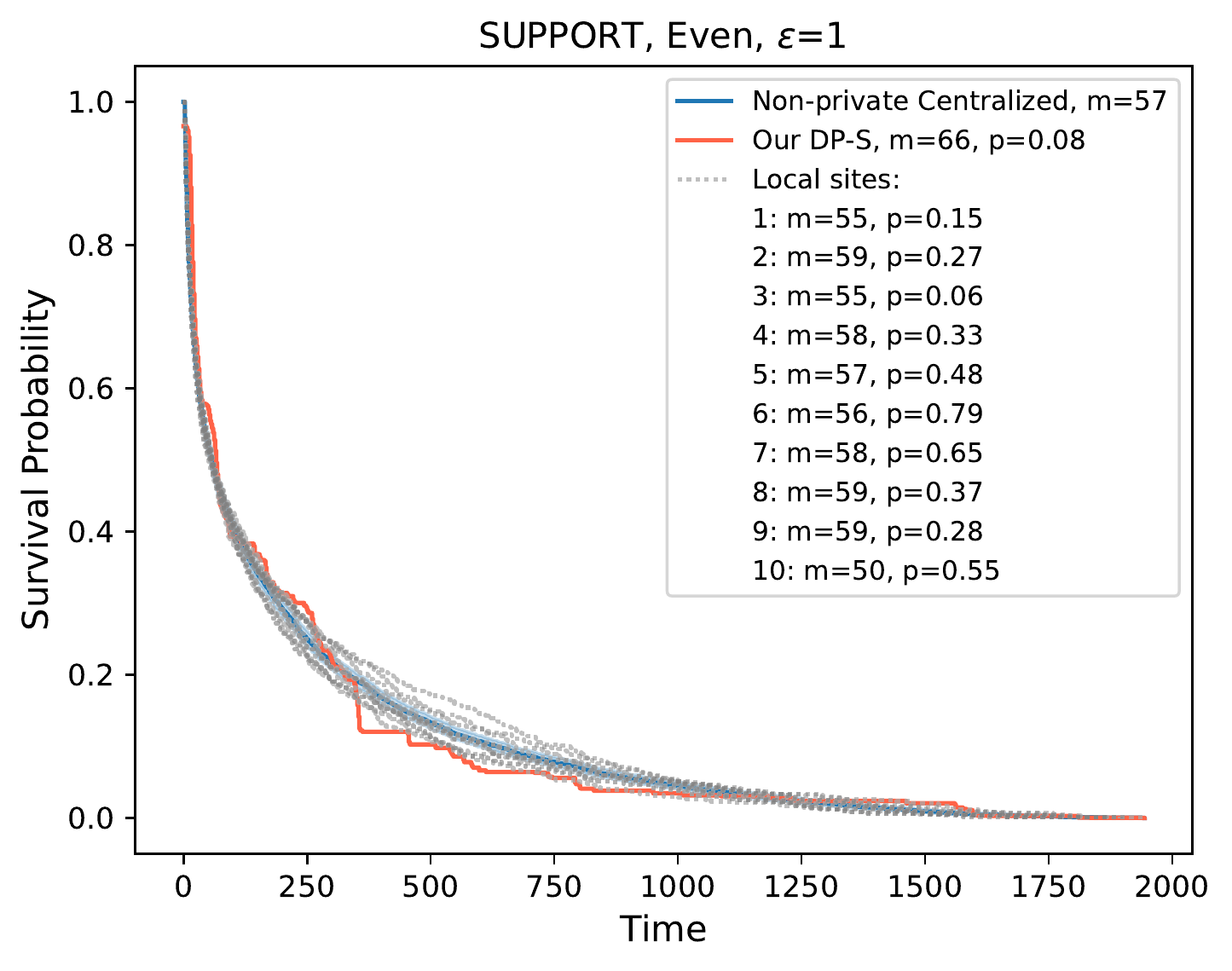}
\end{minipage}%
\hspace*{-.2cm}
\begin{minipage}[l]{0.73\columnwidth}
        \centering
        \includegraphics[width=\linewidth]{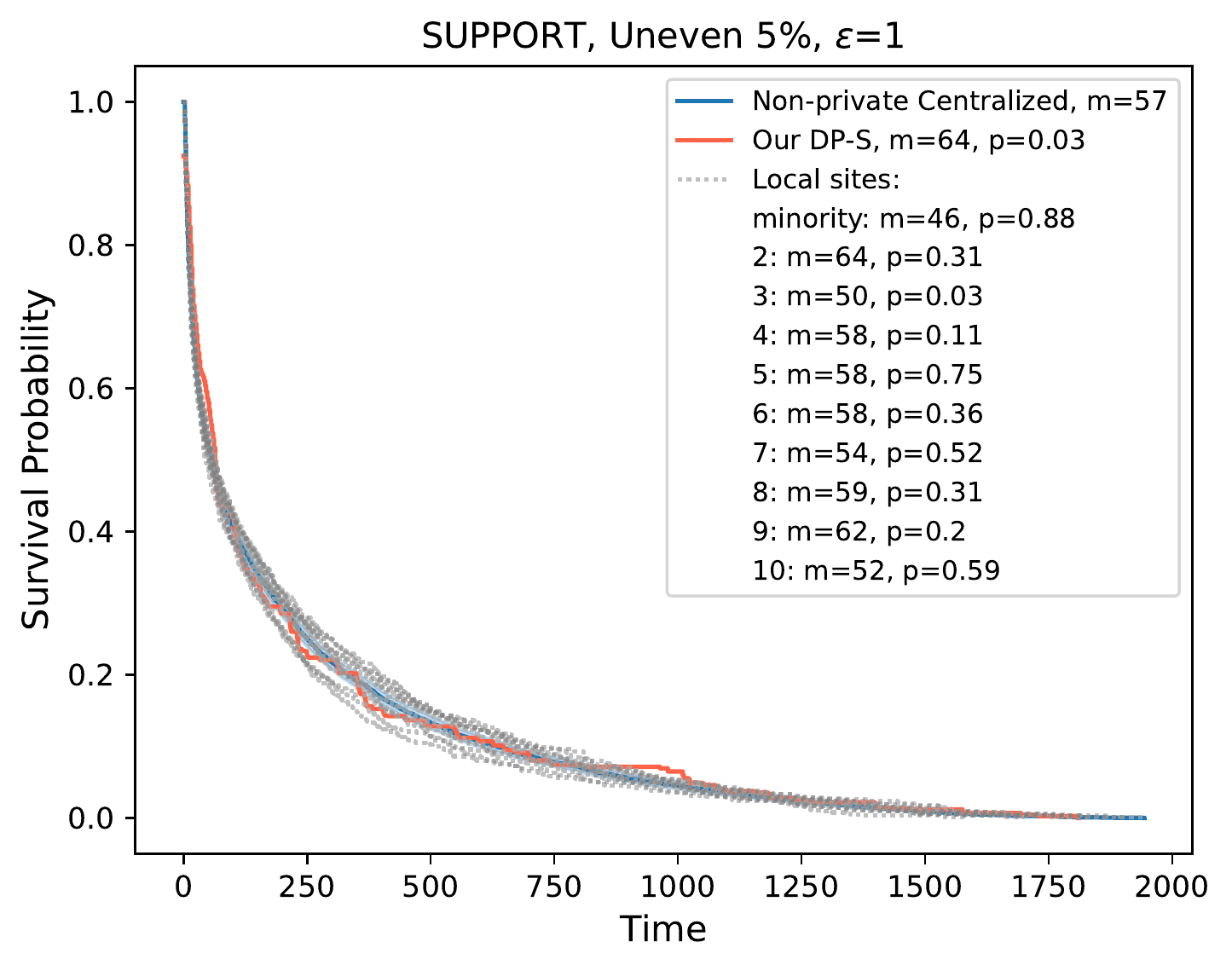}
\end{minipage}%
\hspace*{-.2cm}
\begin{minipage}[l]{0.73\columnwidth}
        \centering
        \includegraphics[width=\linewidth]{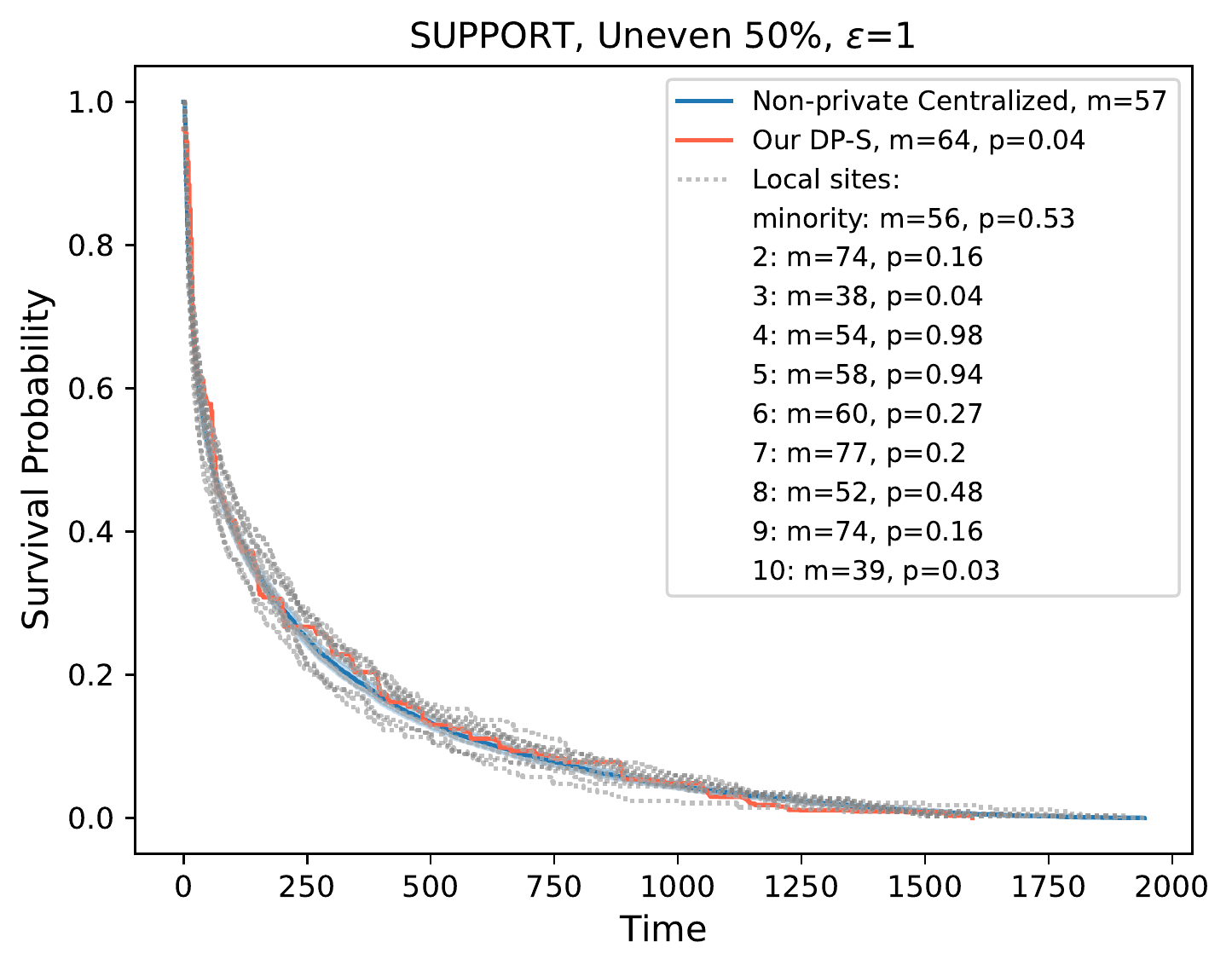}
\end{minipage}
\caption{Collaboration among 10 sites for 3 types of data splitting. Our private \dps method is shown with the red line. The median and the \pv to the non-private, centralized estimator is shown by m and p for our method and also for each site when only the local data is used to construct the KM curve.}
\label{fig:collabe1}
\end{figure*}

\begin{figure*}[!ht]
\centering
\hspace*{-.5cm}
\begin{minipage}[l]{0.73\columnwidth}
        \centering
        \includegraphics[width=\linewidth]{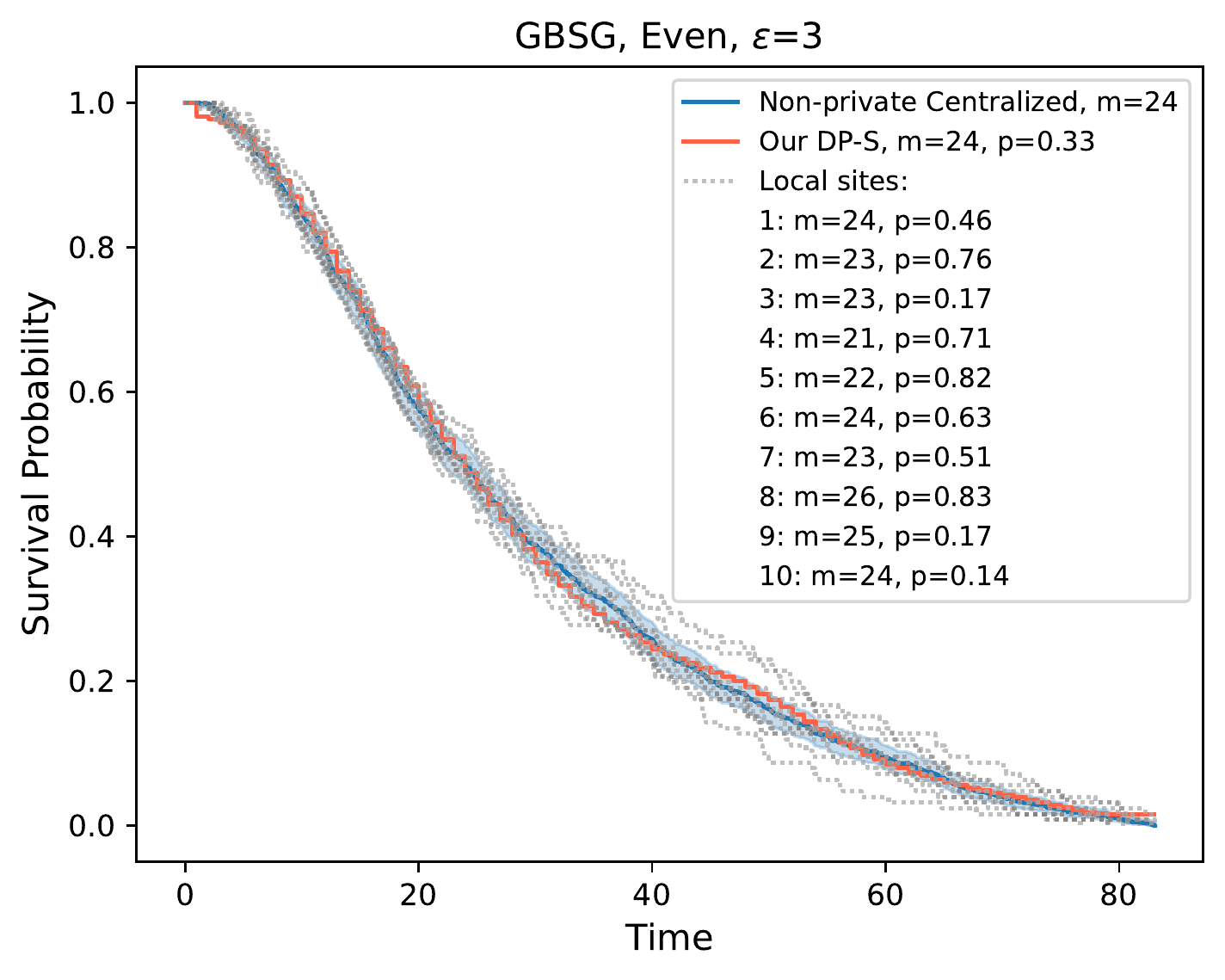}
\end{minipage}%
\hspace*{-.2cm}
\begin{minipage}[l]{0.73\columnwidth}
        \centering
        \includegraphics[width=\linewidth]{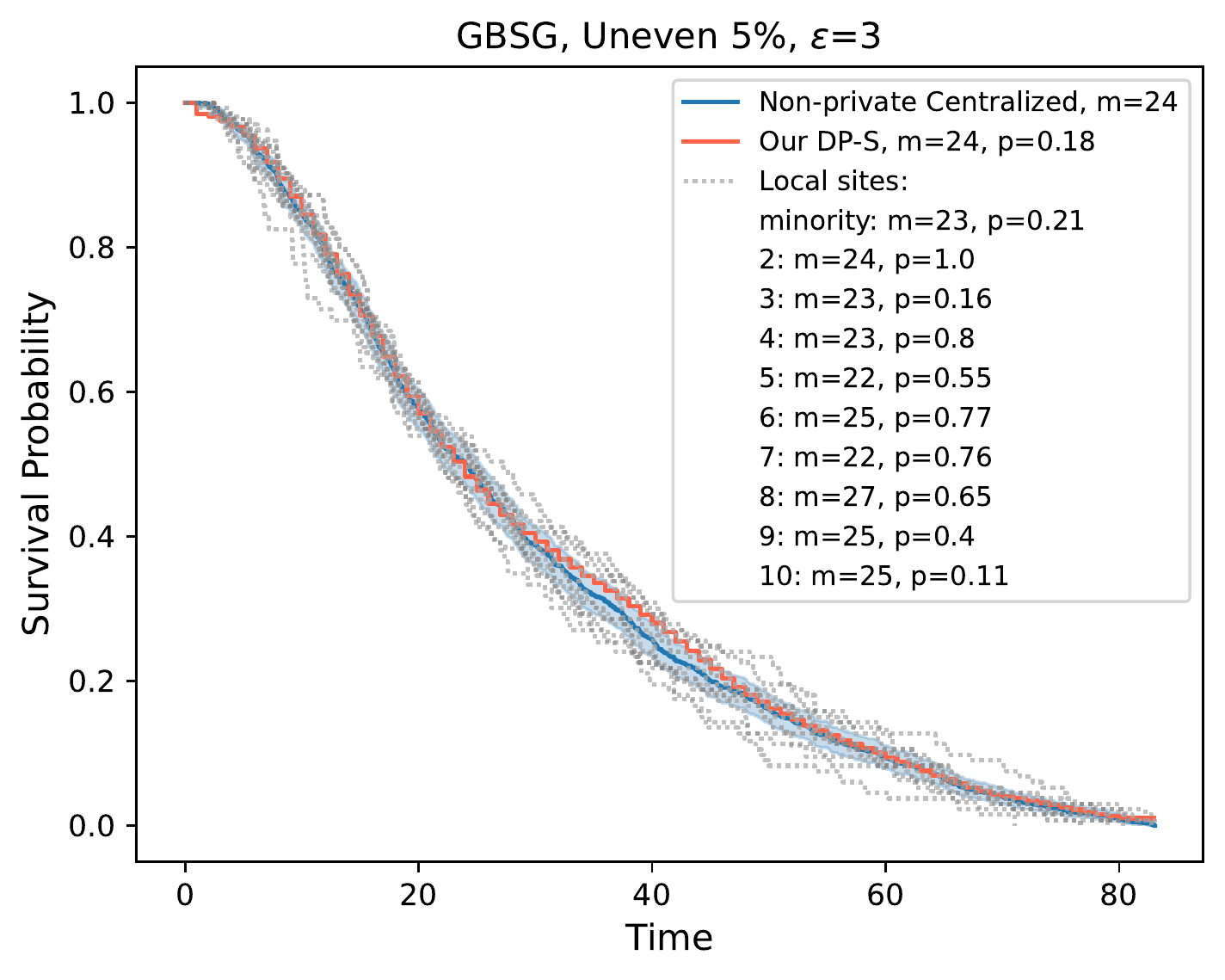}
\end{minipage}%
\hspace*{-.2cm}
\begin{minipage}[l]{0.73\columnwidth}
        \centering
        \includegraphics[width=\linewidth]{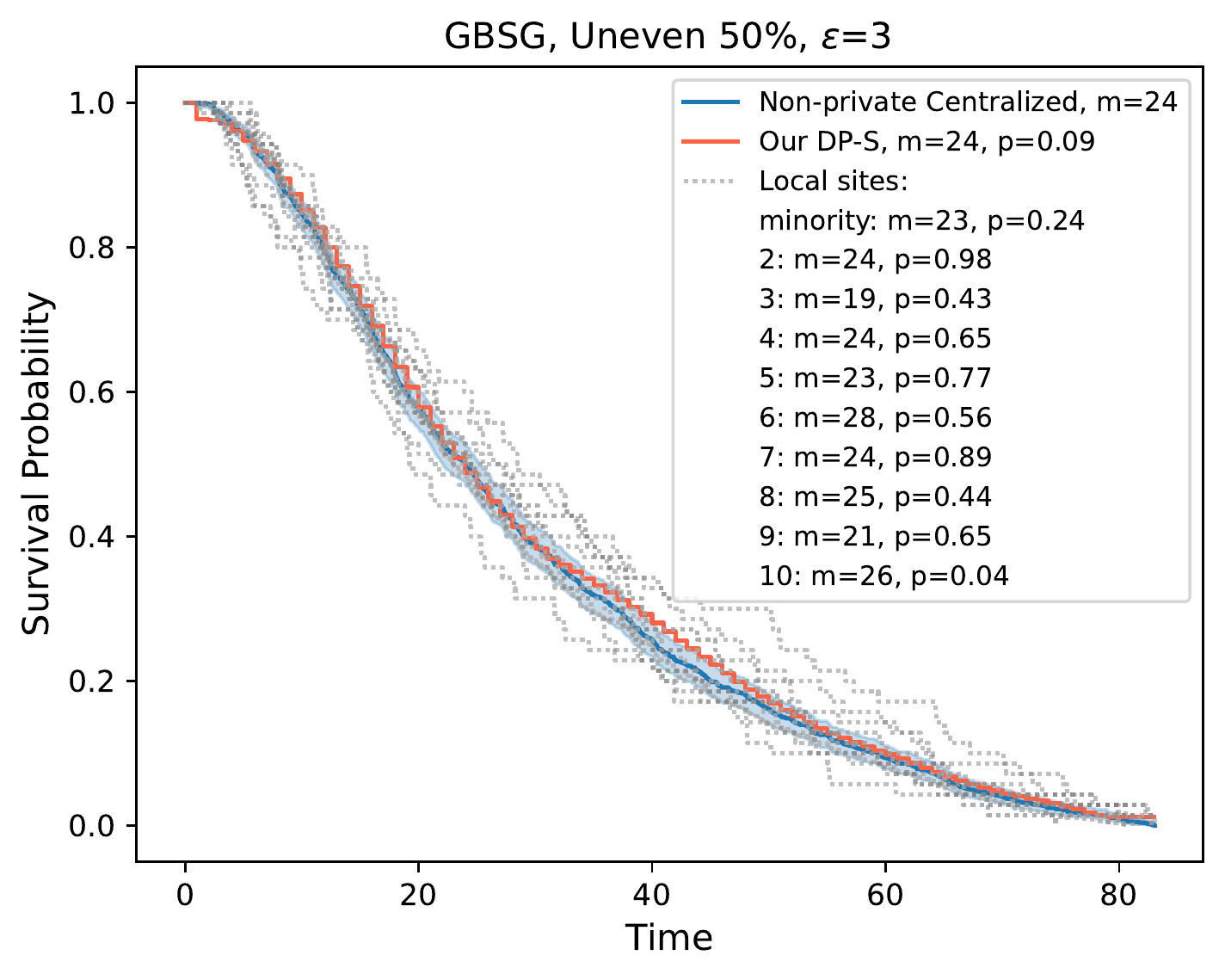}
\end{minipage}
    
\hspace*{-.5cm}
\begin{minipage}[l]{0.73\columnwidth}
        \centering
        \includegraphics[width=\linewidth]{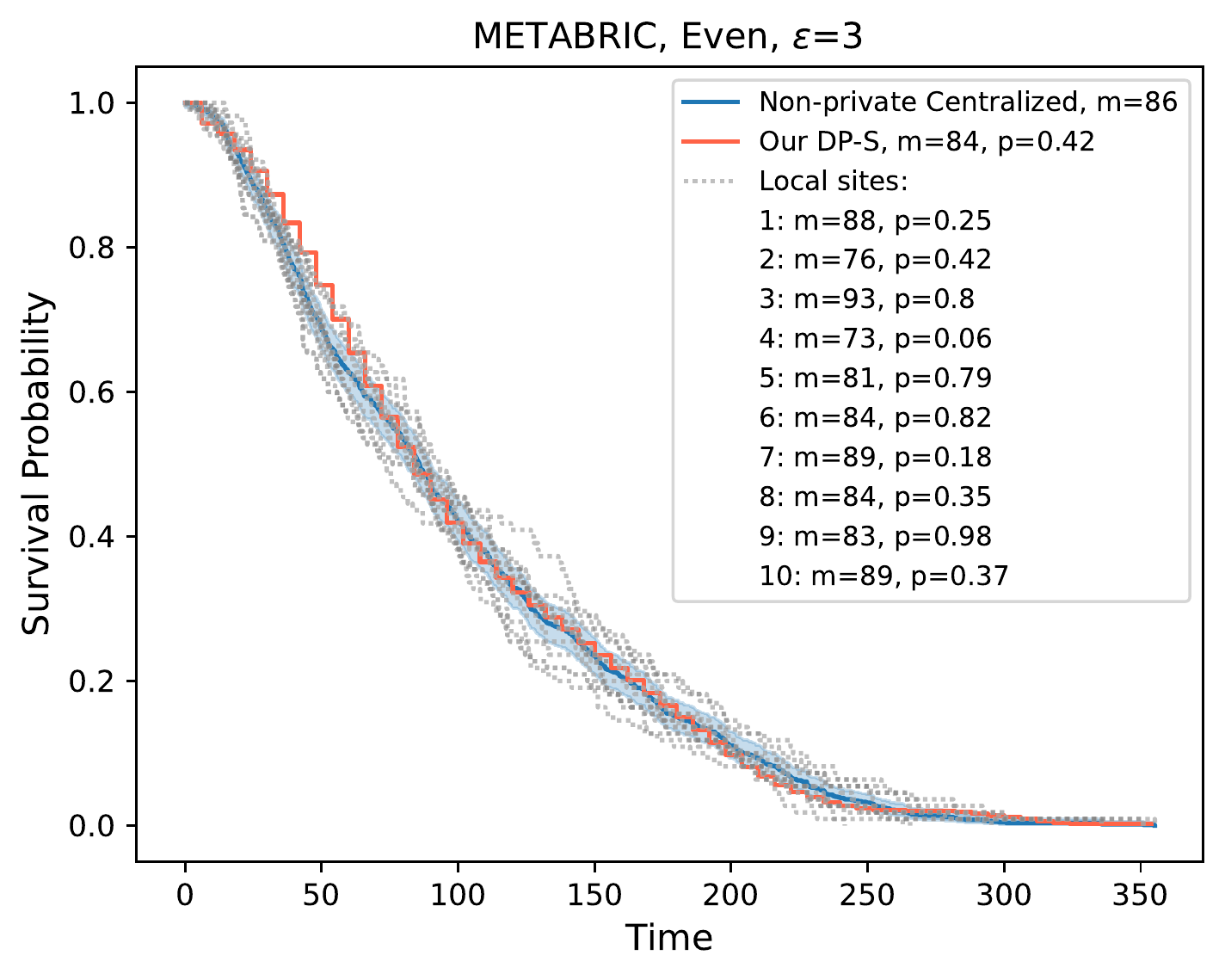}
\end{minipage}%
\hspace*{-.2cm}
\begin{minipage}[l]{0.73\columnwidth}
        \centering
        \includegraphics[width=\linewidth]{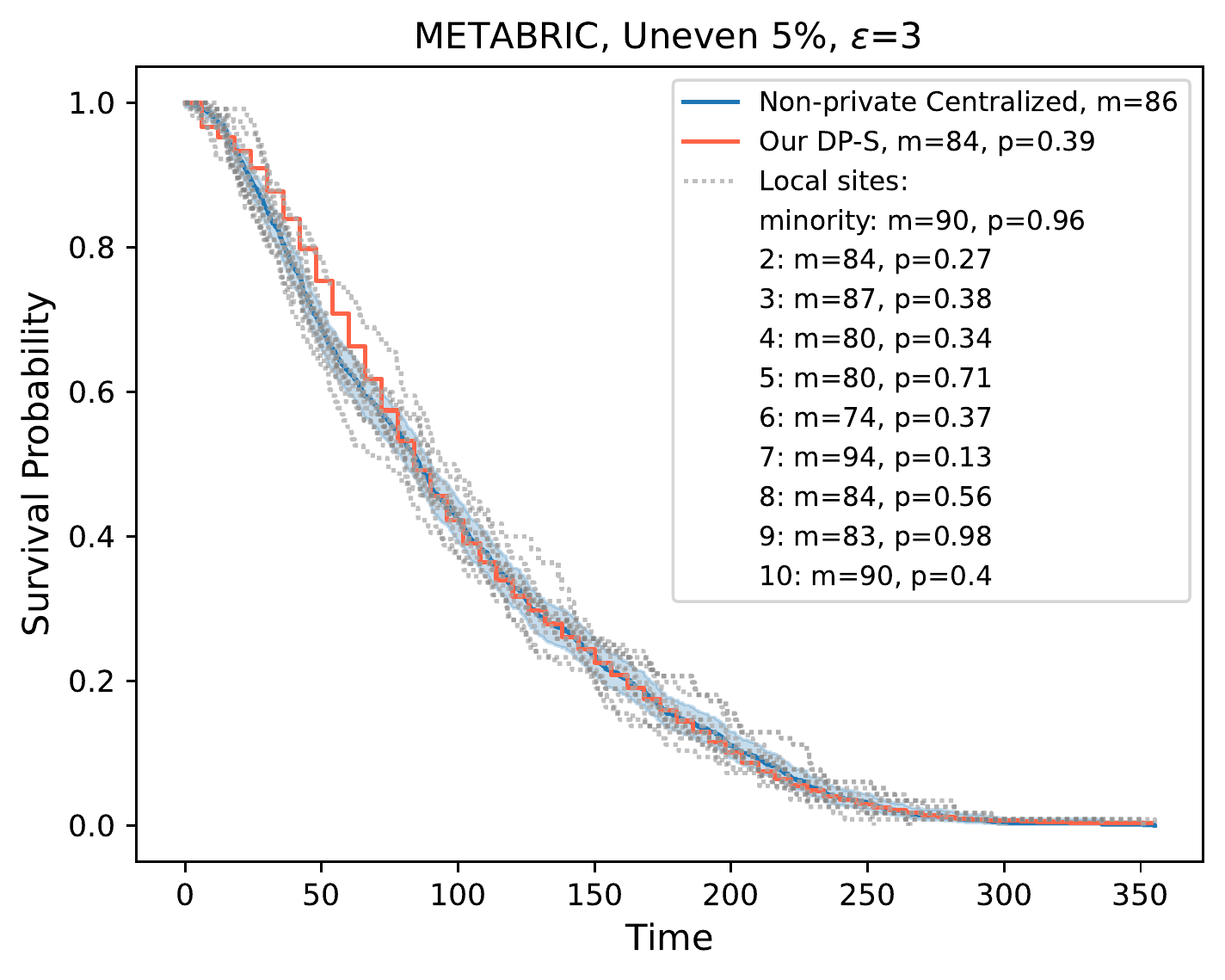}
\end{minipage}%
\hspace*{-.2cm}
\begin{minipage}[l]{0.73\columnwidth}
        \centering
        \includegraphics[width=\linewidth]{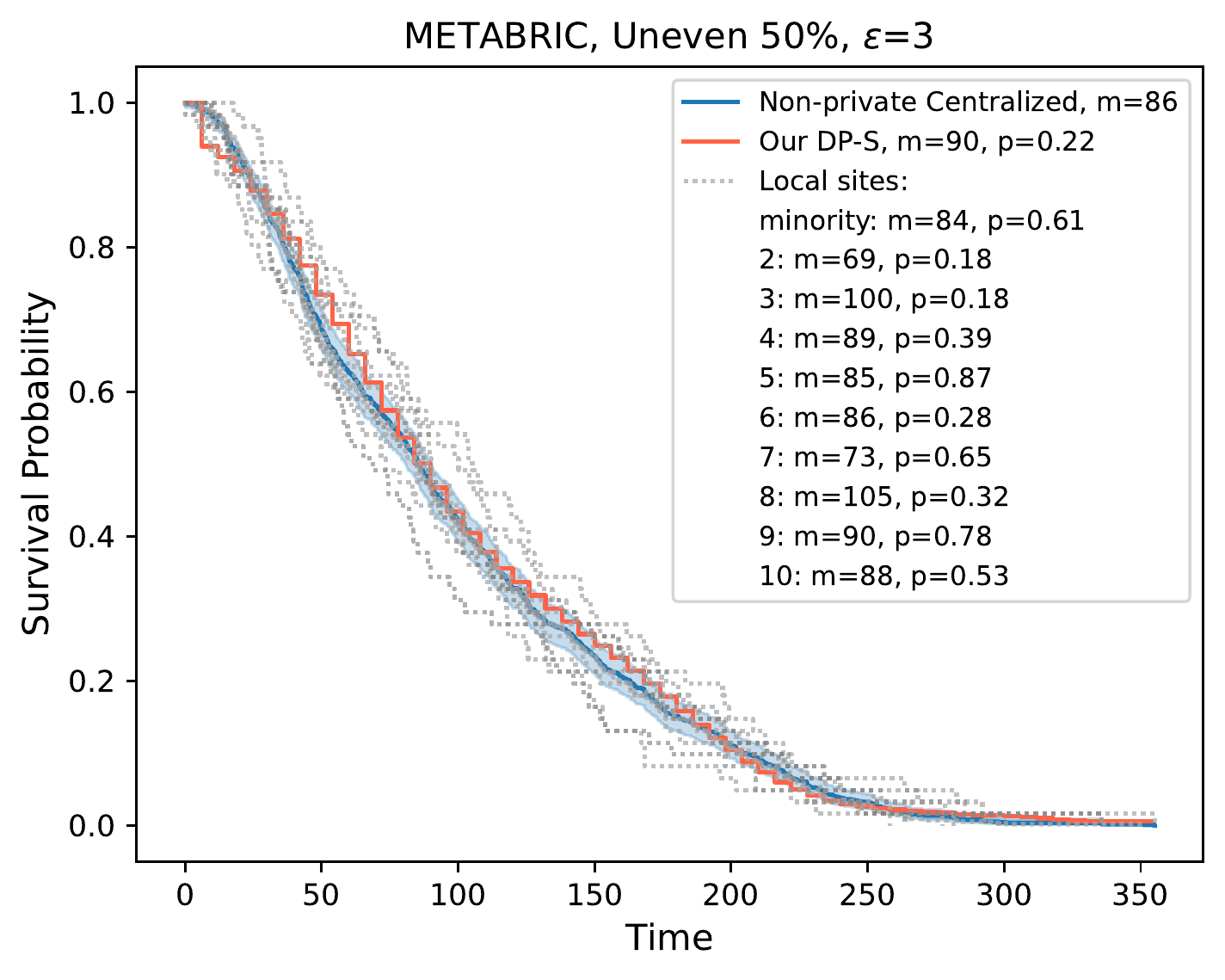}
\end{minipage}

\hspace*{-.5cm}
\begin{minipage}[l]{0.73\columnwidth}
        \centering
        \includegraphics[width=\linewidth]{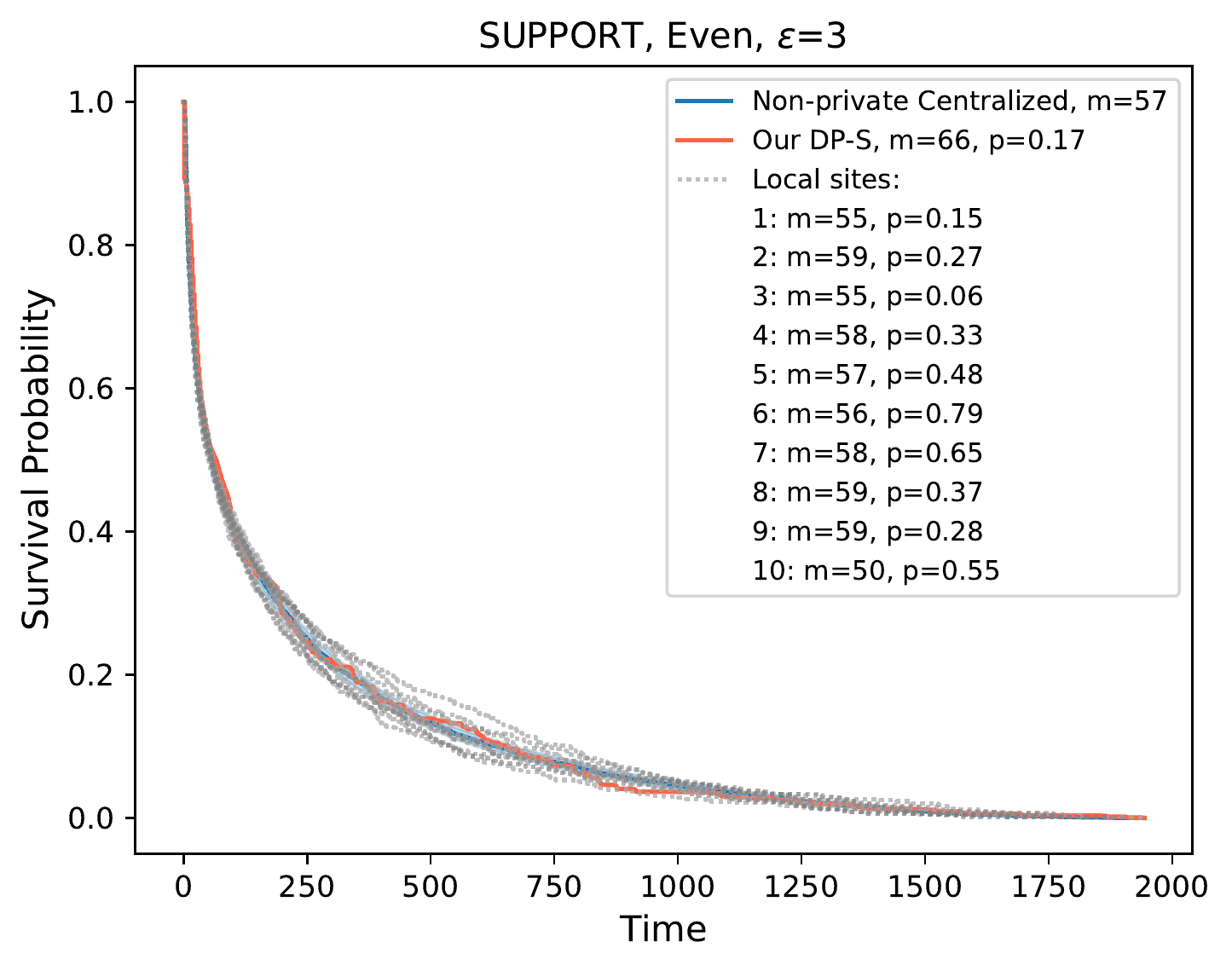}
\end{minipage}%
\hspace*{-.2cm}
\begin{minipage}[l]{0.73\columnwidth}
        \centering
        \includegraphics[width=\linewidth]{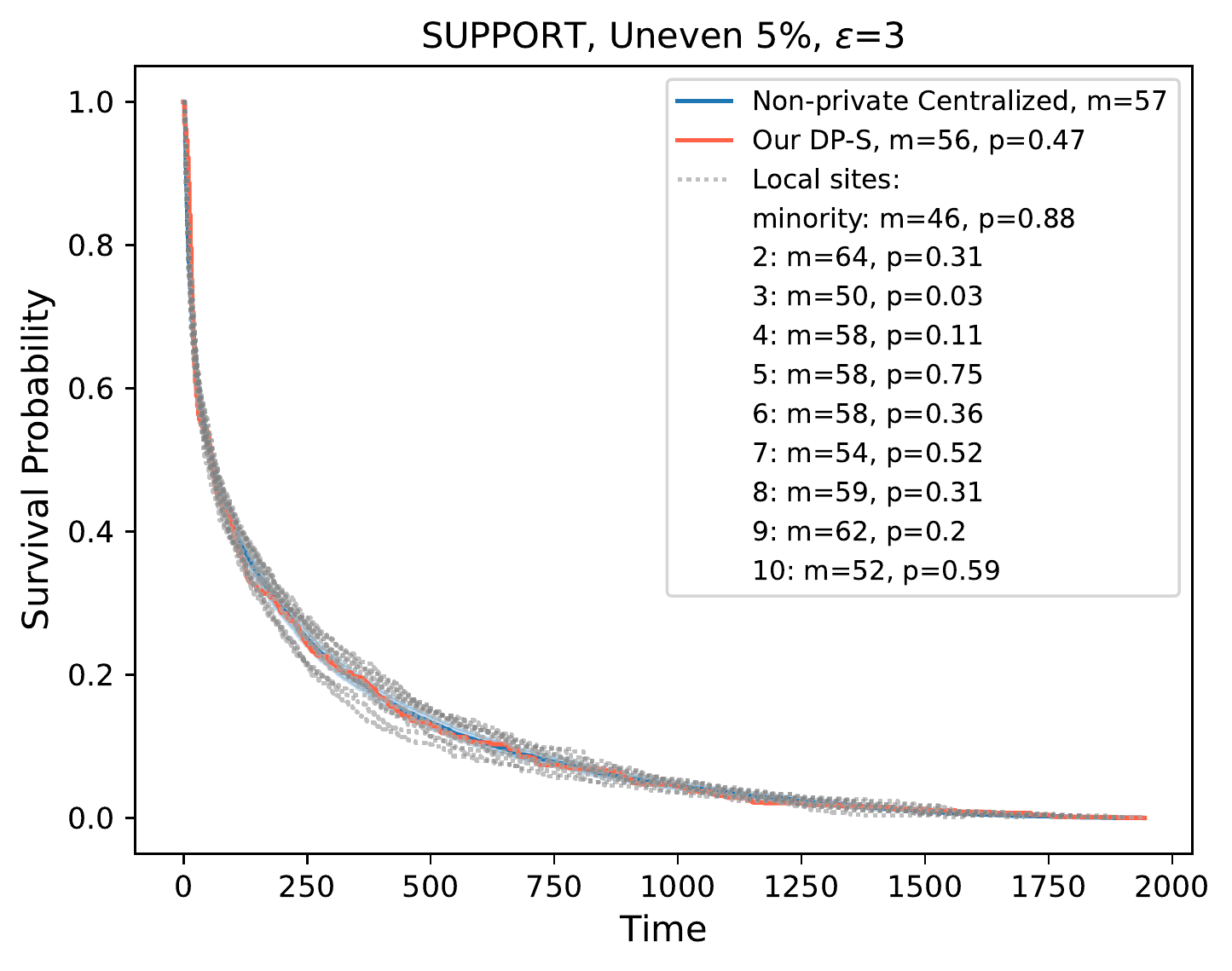}
\end{minipage}%
\hspace*{-.2cm}
\begin{minipage}[l]{0.73\columnwidth}
        \centering
        \includegraphics[width=\linewidth]{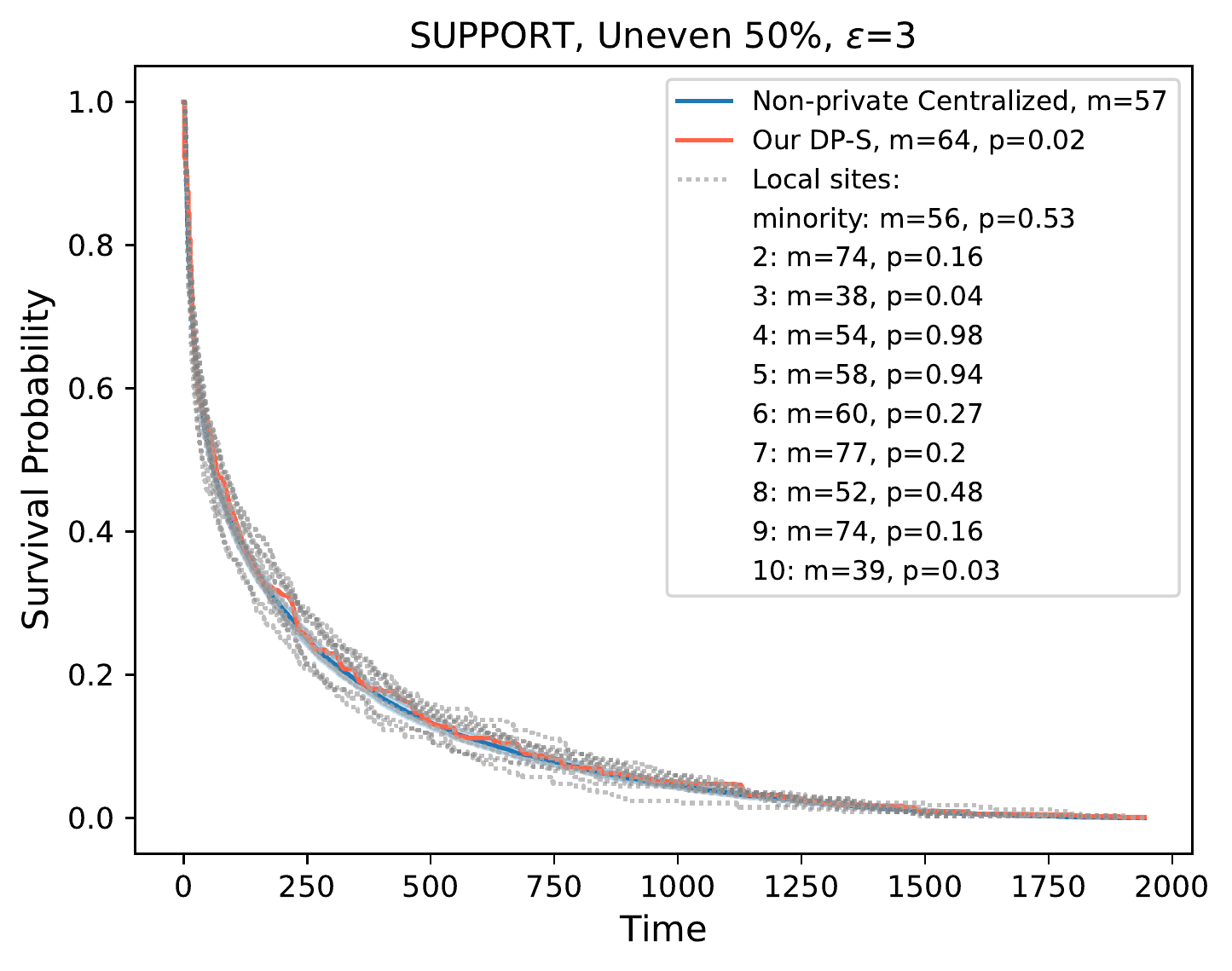}
\end{minipage}
\caption{Collaboration among 10 sites for 3 types of data splitting. Our private \dps method is shown with the red line. The median and the \pv to the non-private, centralized estimator is shown by m and p for our method and also for each site when only the local data is used to construct the KM curve.}
\label{fig:collabe3}
\end{figure*}

\begin{figure*}[ht]
\centering
\begin{minipage}[l]{0.67\columnwidth}
        \centering
        \includegraphics[width=\linewidth]{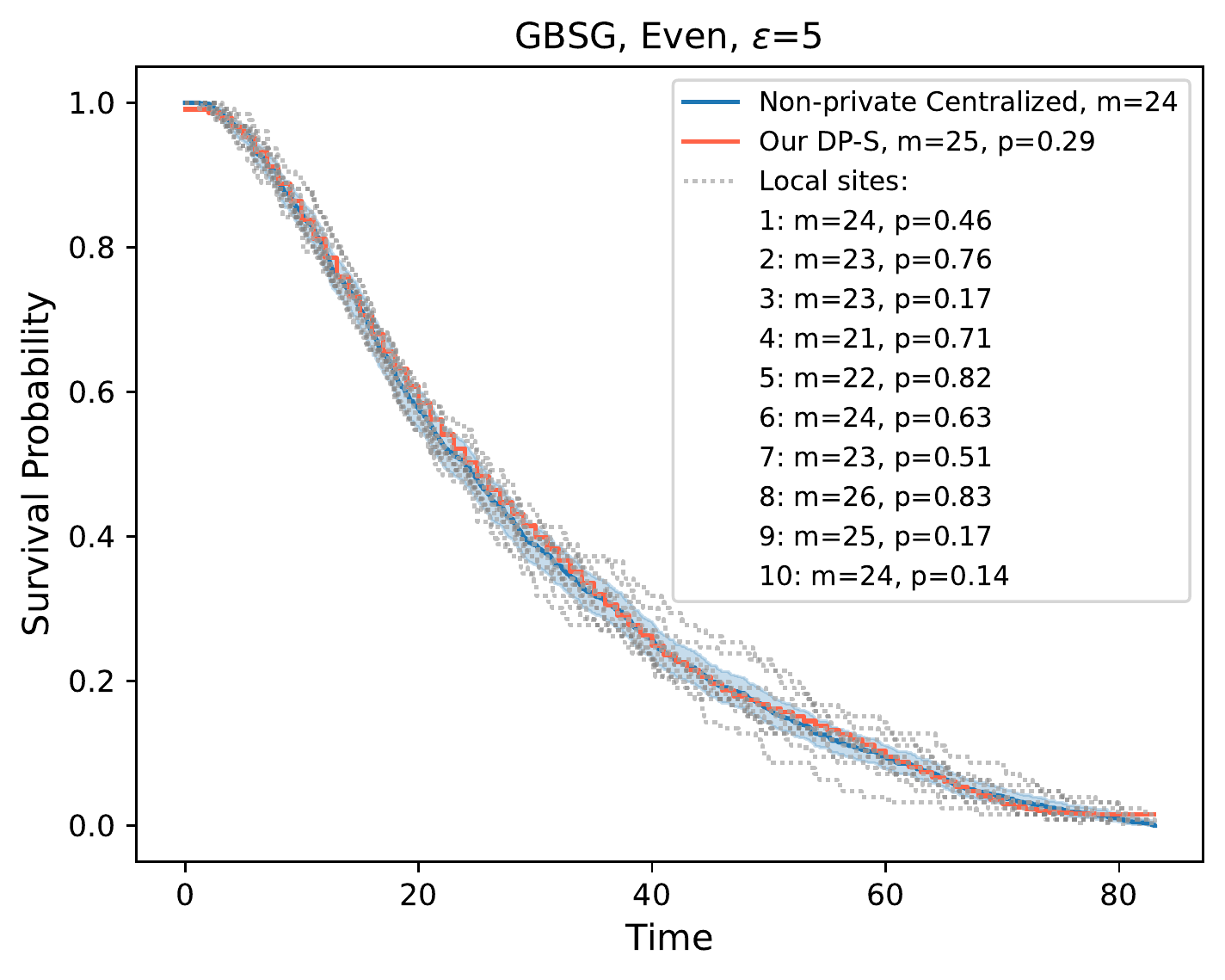}
\end{minipage}
\begin{minipage}[l]{0.67\columnwidth}
        \centering
        \includegraphics[width=\linewidth]{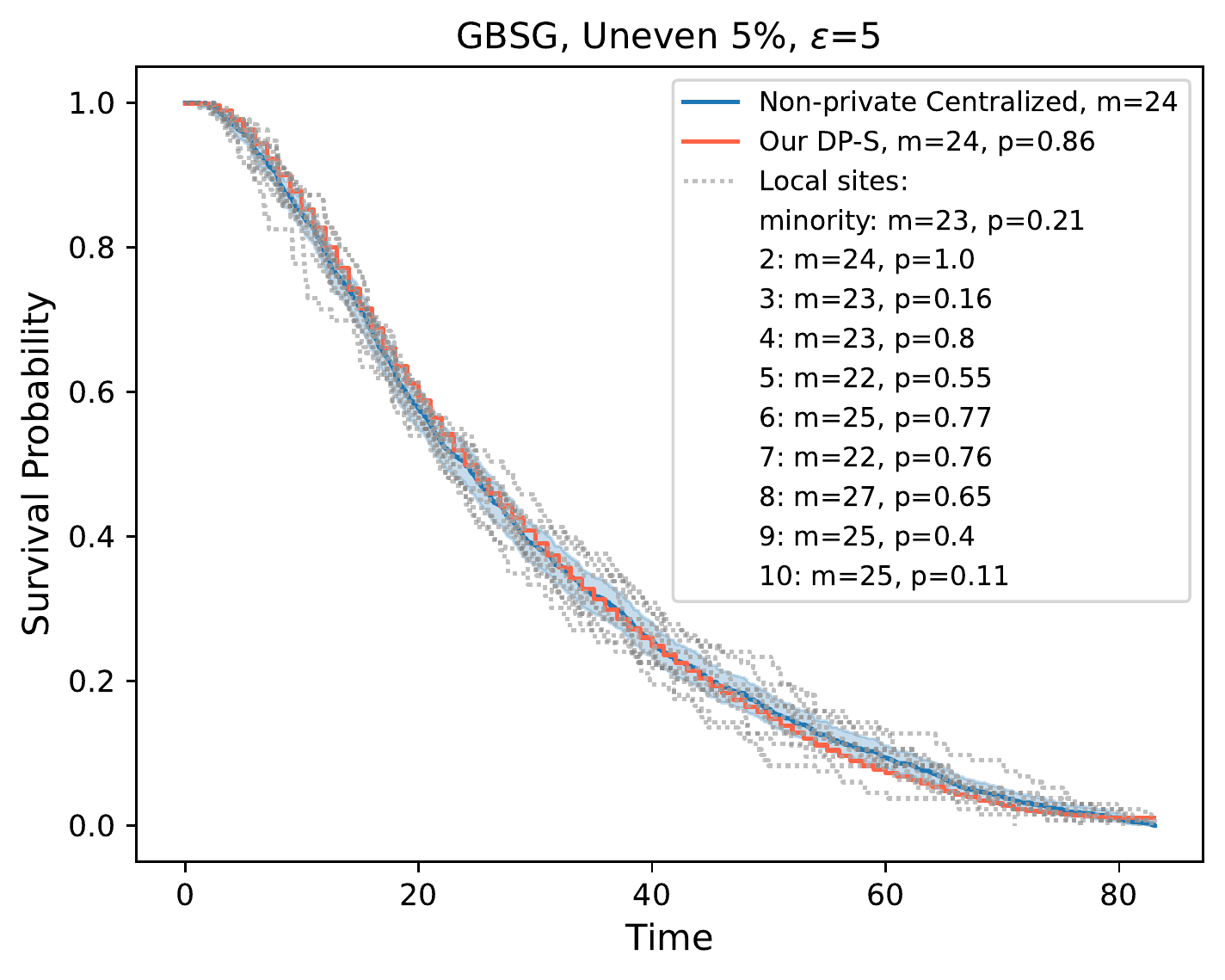}
\end{minipage}
\begin{minipage}[l]{0.67\columnwidth}
        \centering
        \includegraphics[width=\linewidth]{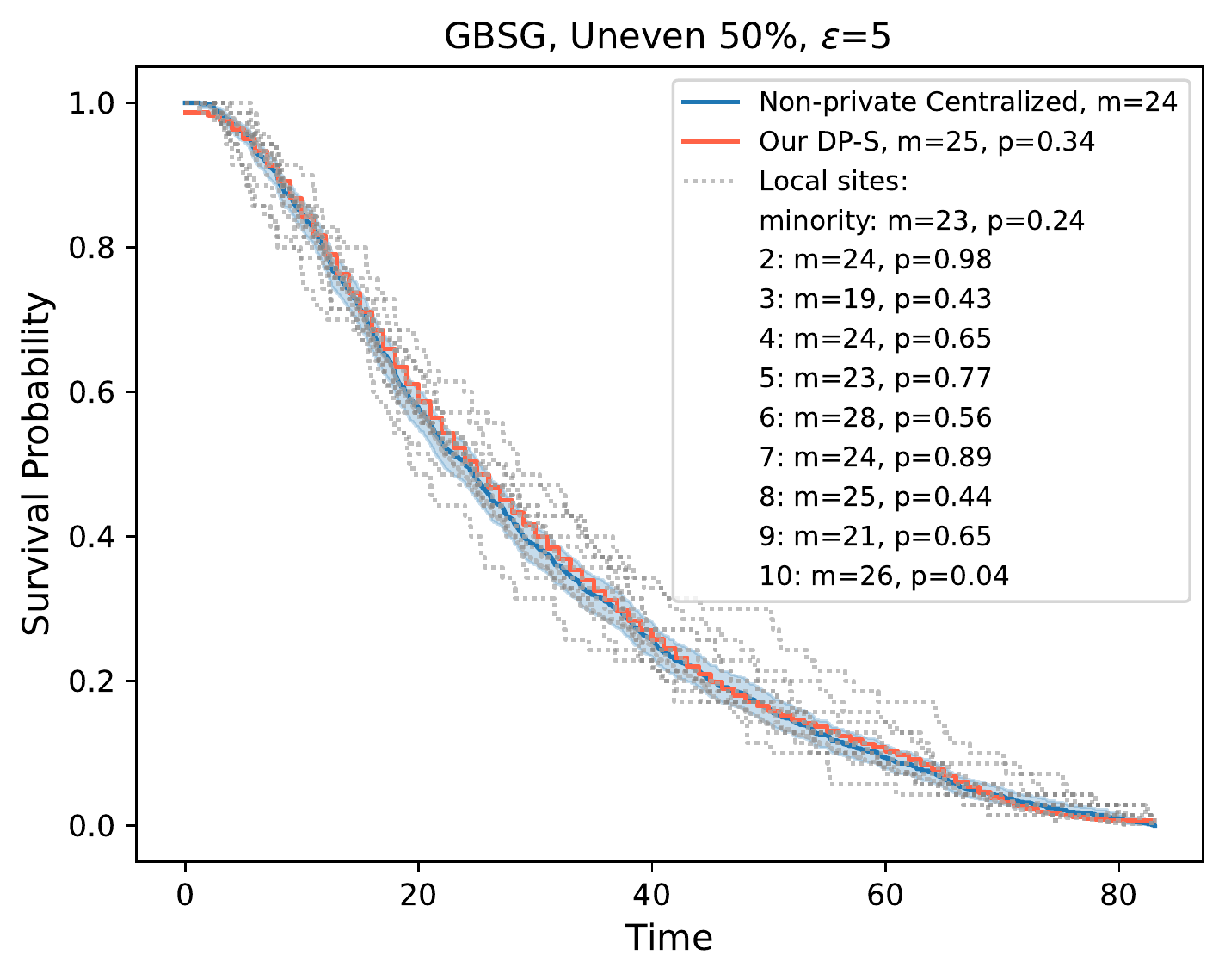}
\end{minipage}
    
\begin{minipage}[l]{0.67\columnwidth}
        \centering
        \includegraphics[width=\linewidth]{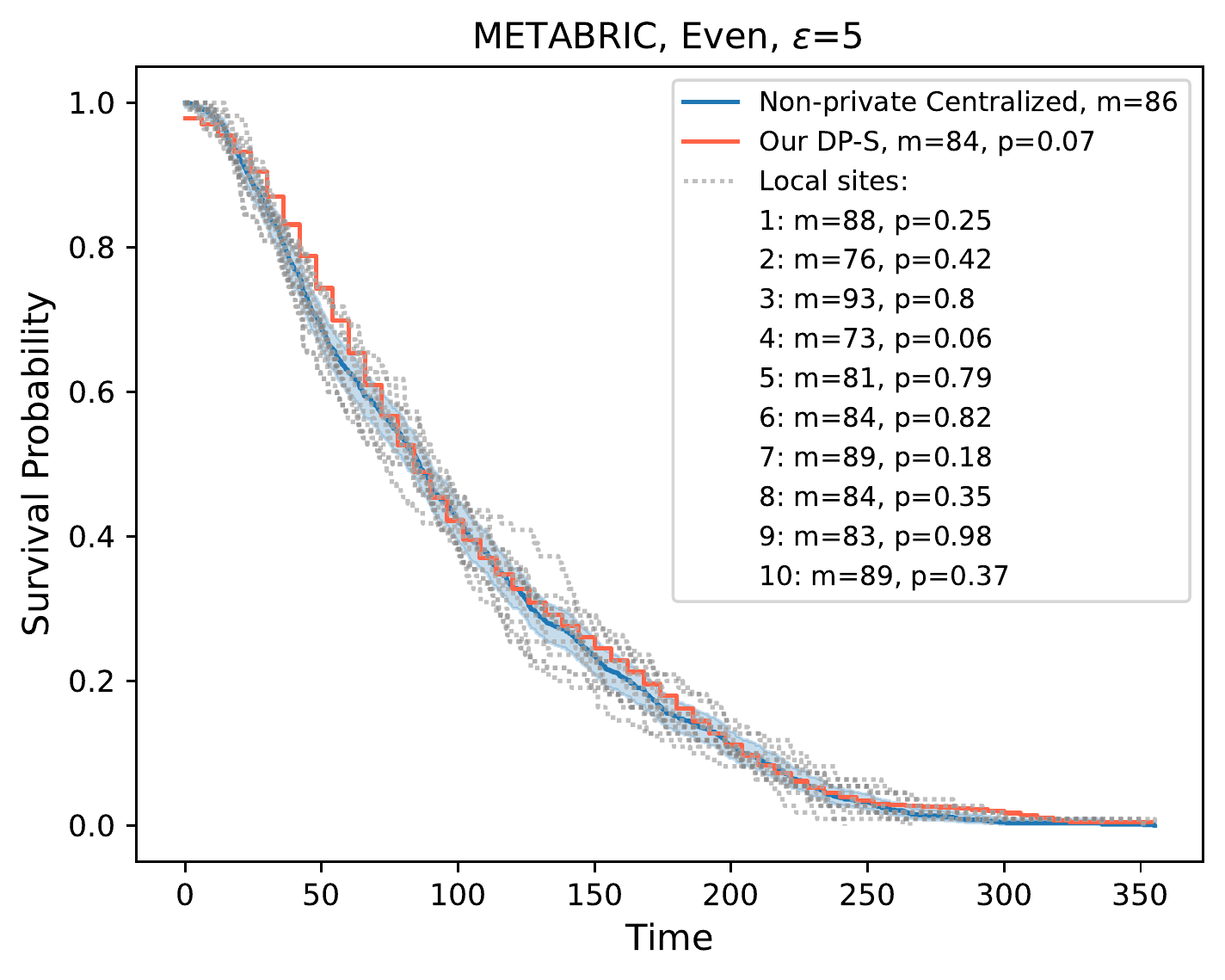}
\end{minipage}
\begin{minipage}[l]{0.67\columnwidth}
        \centering
        \includegraphics[width=\linewidth]{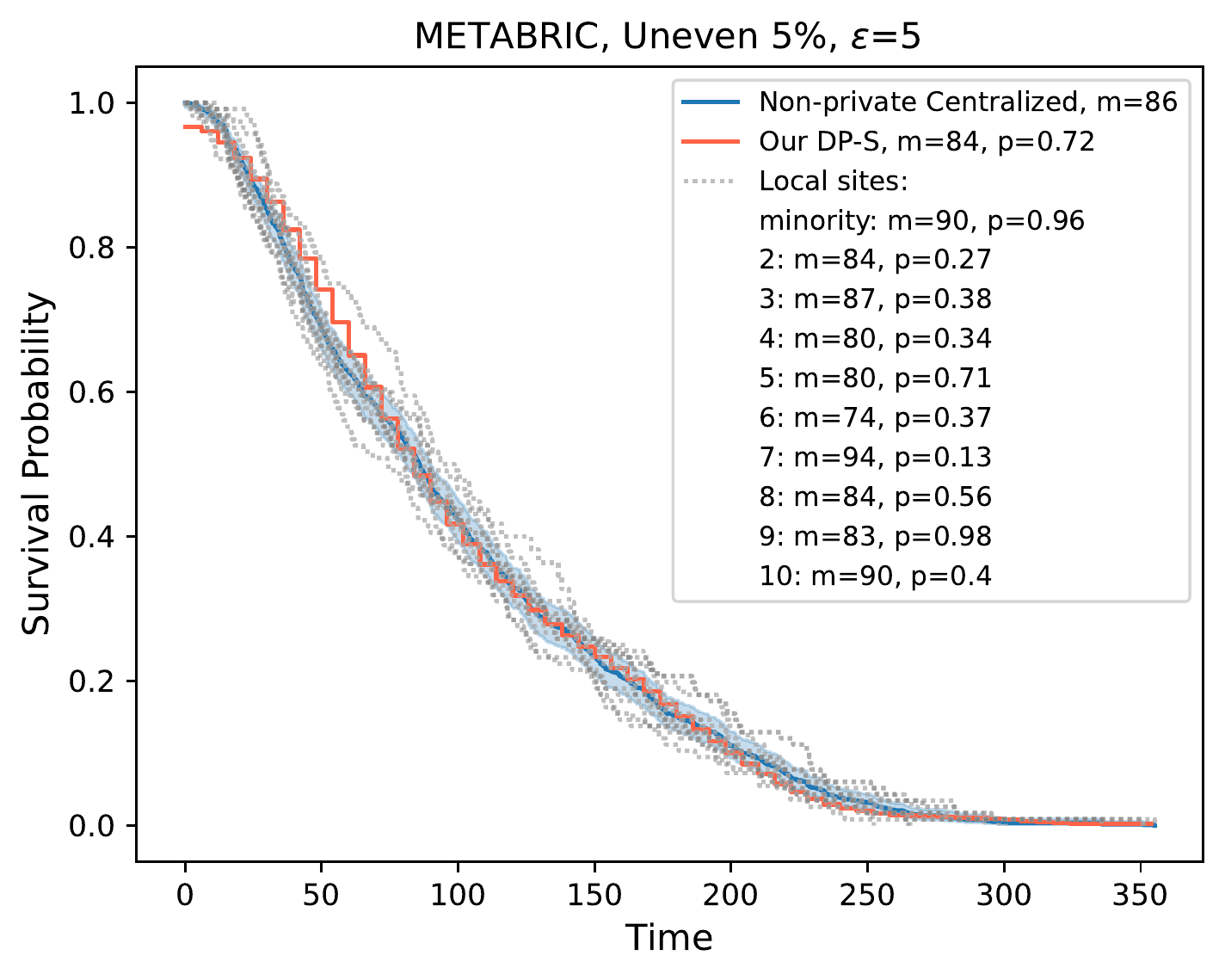}
\end{minipage}
\begin{minipage}[l]{0.67\columnwidth}
        \centering
        \includegraphics[width=\linewidth]{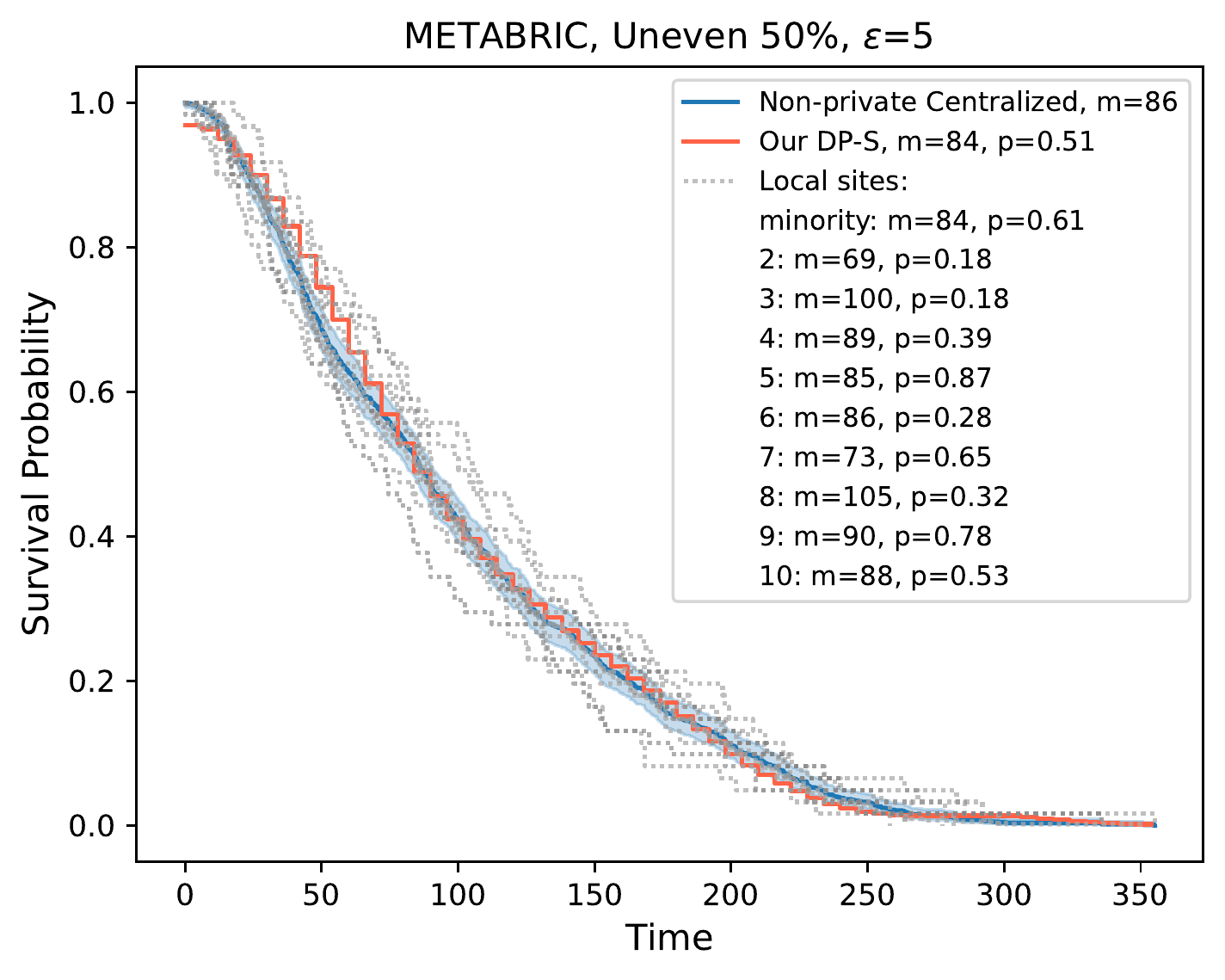}
\end{minipage}

\begin{minipage}[l]{0.67\columnwidth}
        \centering
        \includegraphics[width=\linewidth]{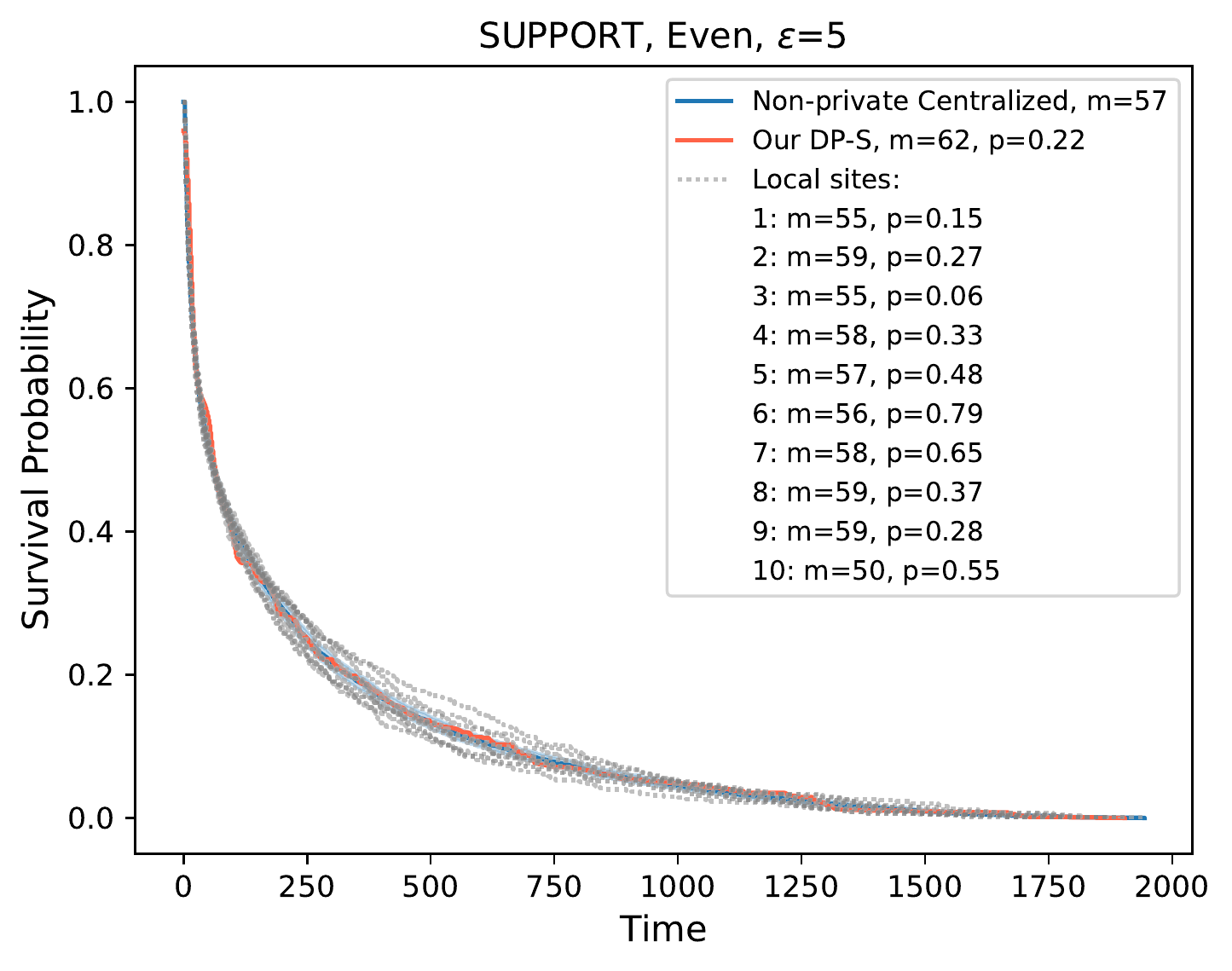}
\end{minipage}
\begin{minipage}[l]{0.67\columnwidth}
        \centering
        \includegraphics[width=\linewidth]{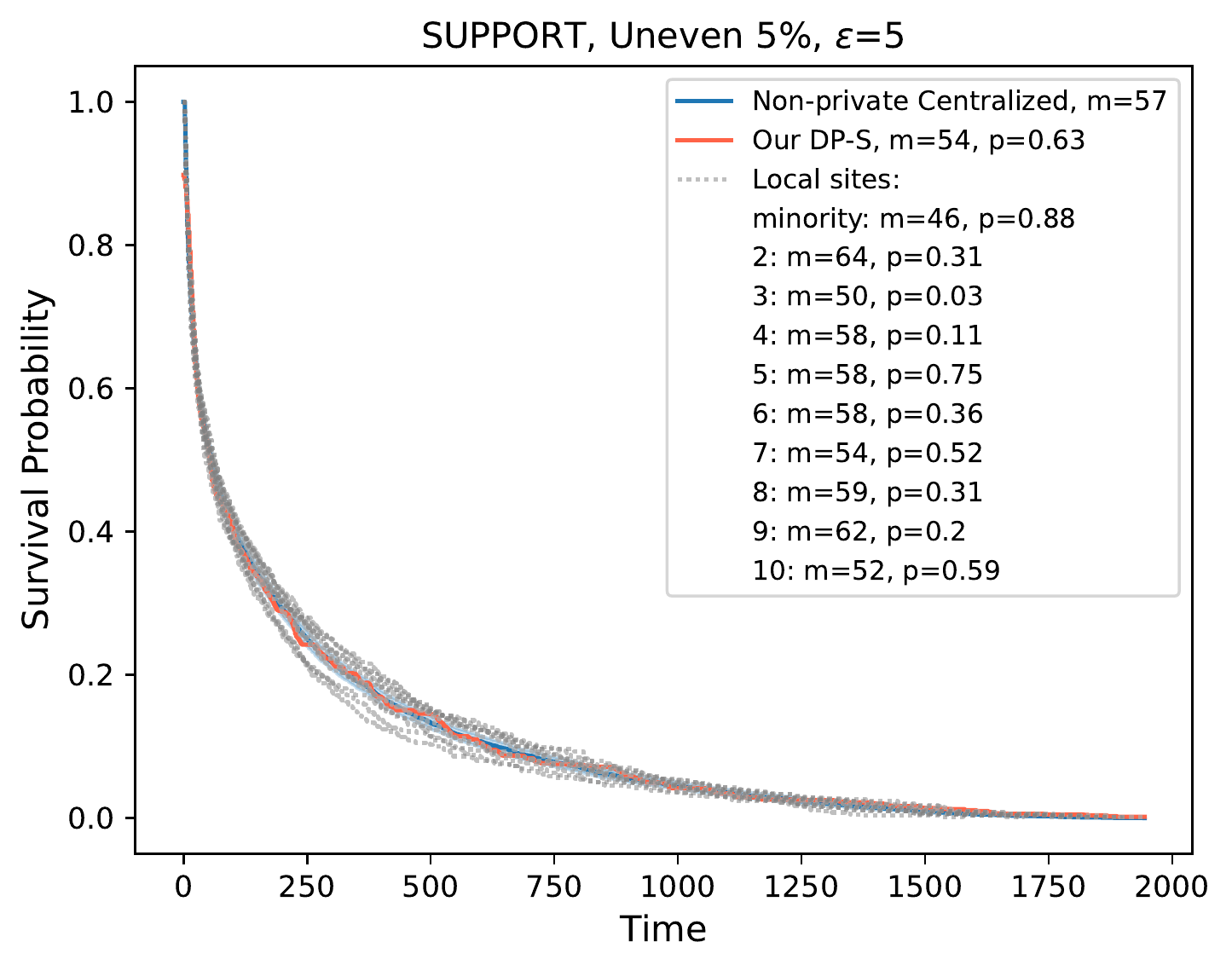}
\end{minipage}
\begin{minipage}[l]{0.67\columnwidth}
        \centering
        \includegraphics[width=\linewidth]{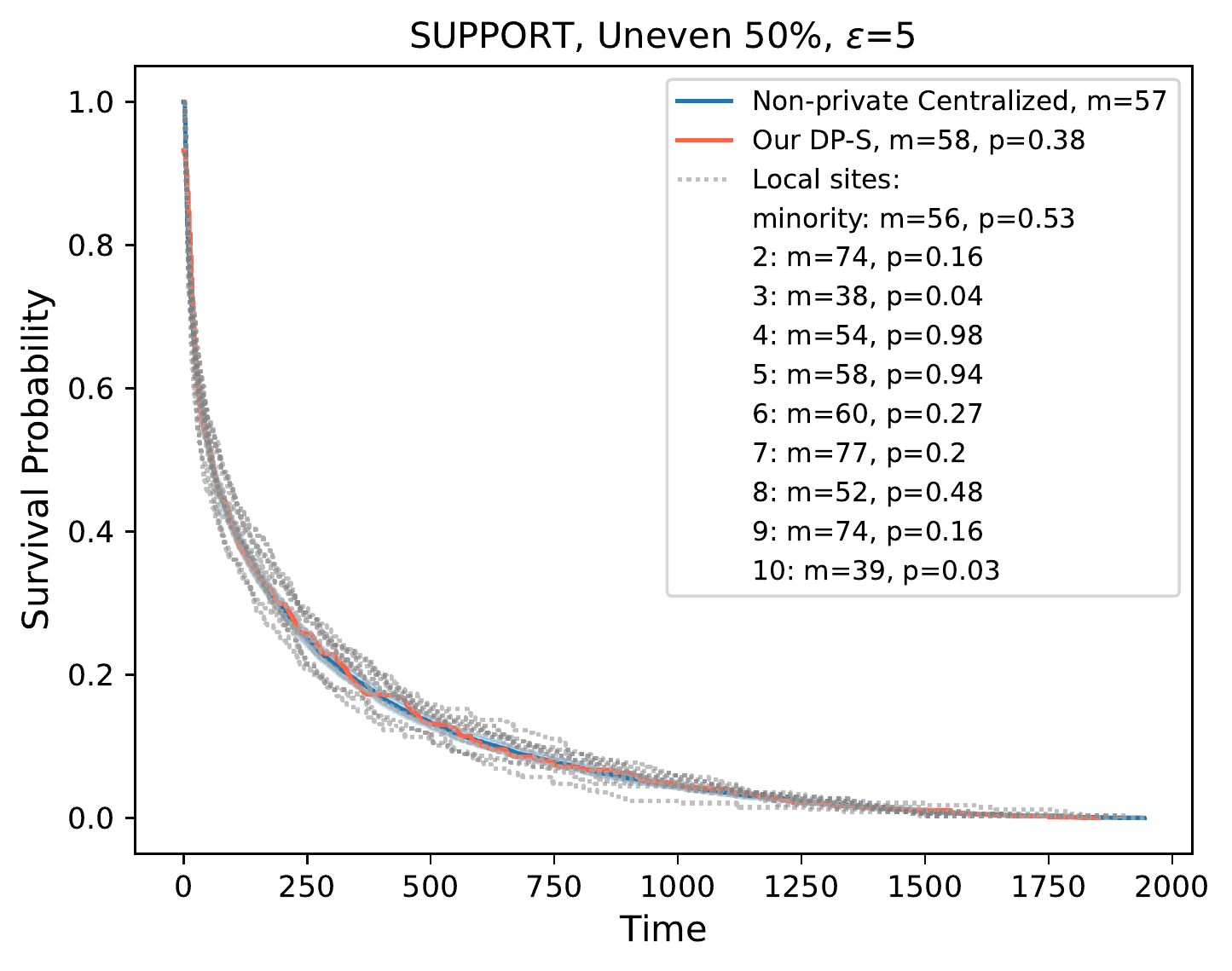}
\end{minipage}
\caption{Collaborative private estimation of Kaplan-Meier curves among 10 participating sites for 3 types of data splitting. The blue line shows the non-private centralized case and the red line shows our performance after constructing a joint private estimator. We also plot the locally estimated KM curve with dotted lines for all 10 sites. The median and the \pv to the non-DP centralized estimator is shown by m and p.}
\label{fig:collabe05}
\end{figure*}

\clearpage
\newpage
\subsection{Datasets with Censoring}
\label{sec:app-smooth}

In this section, we examine the effect of including censored points in the datasets. Figure~\ref{fig:centralized-smooth} shows the KM estimators for the 3 datasets. The complete dataset and its confidence interval are colored blue. The dataset which contains only the noncensored portion of the data points is colored orange. 
As explained in Section~\ref{sec:dpmdef}, \dpm is also defined for datasets containing censoring. The KM curve generated by \dpm in a centralized setting is shown in green. We also explained in Appendix~\ref{sec:appendix-proof-s}, that to make the general sensitivities that we find for the KM curve differentially private, we can consider the extreme case of $C=N$. The performance of our \dps method using this sensitivity is shown in red. 

As we had predicted, the noise of the DP mechanism, in the case of \dps with $C=N$ for sensitivity, renders the utility useless, especially for the value $\varepsilon = 10$. By increasing the value of $\varepsilon$ to much larger amounts, we start to see that the \dps curve becomes closer in behavior to the non-private curve.  

\begin{figure*}[!h]
\centering
\begin{minipage}[l]{0.68\columnwidth}
        \centering
        \includegraphics[width=\linewidth]{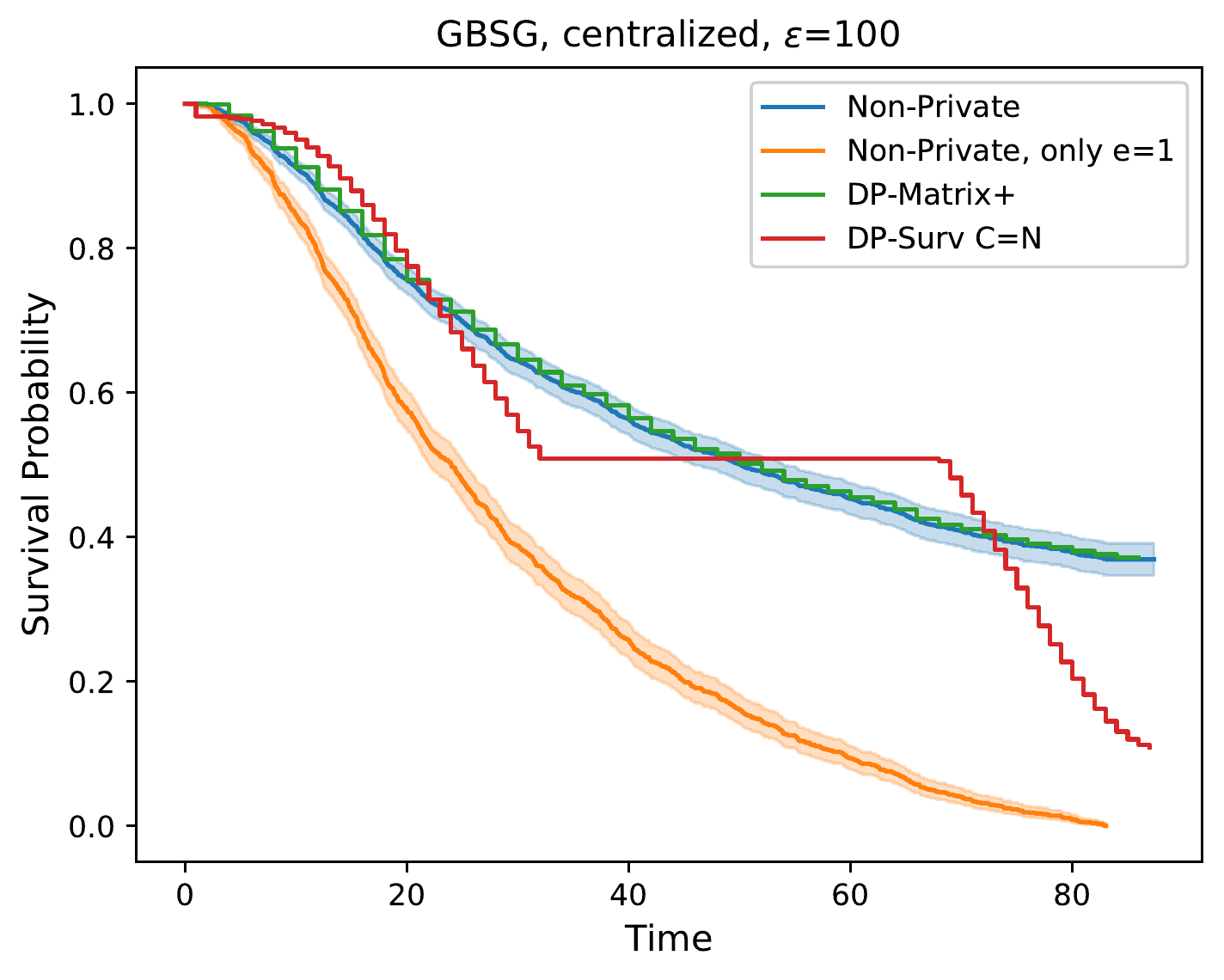}
\end{minipage}
\begin{minipage}[l]{0.68\columnwidth}
        \centering
        \includegraphics[width=\linewidth]{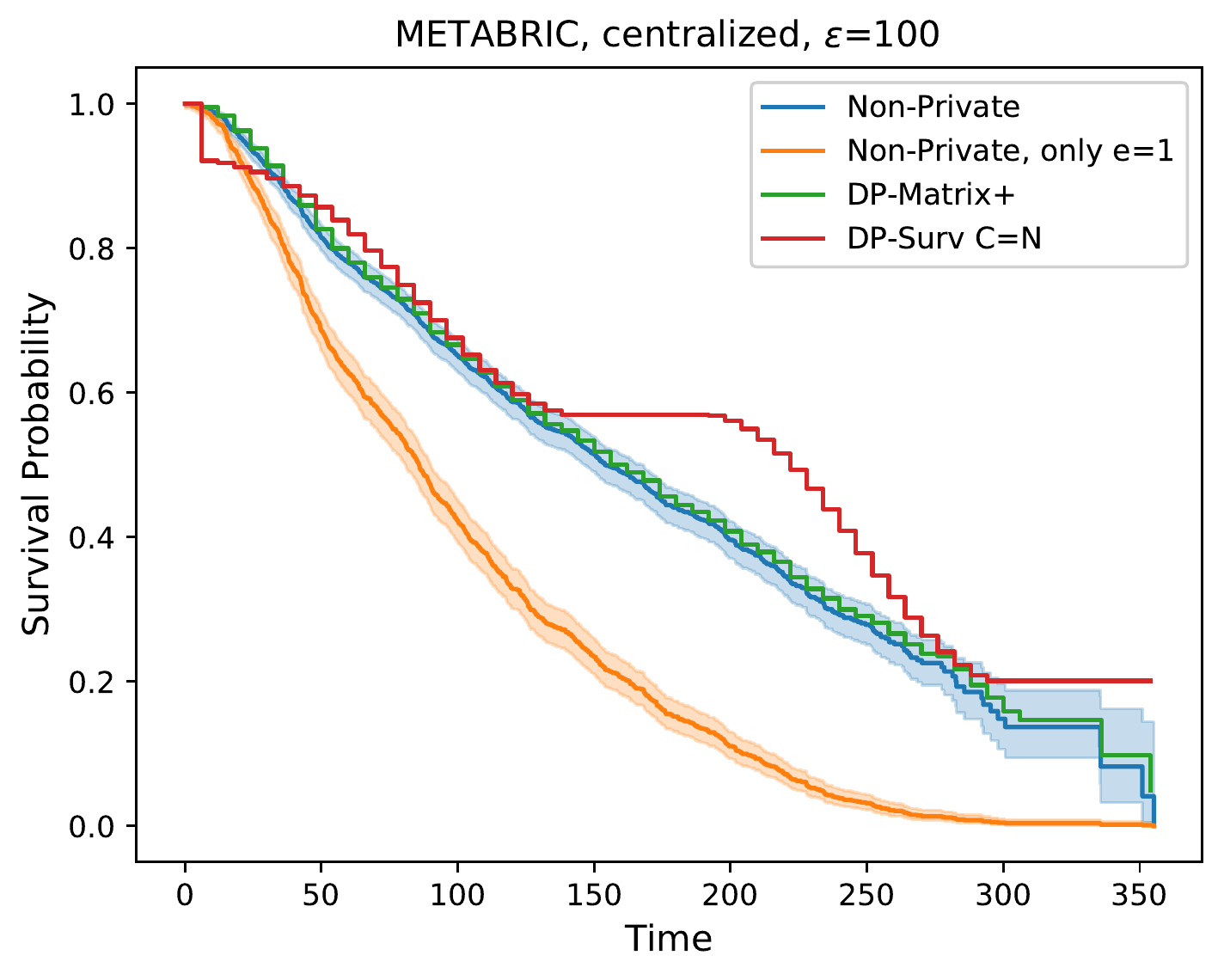}
\end{minipage}
\begin{minipage}[l]{0.68\columnwidth}
        \centering
        \includegraphics[width=\linewidth]{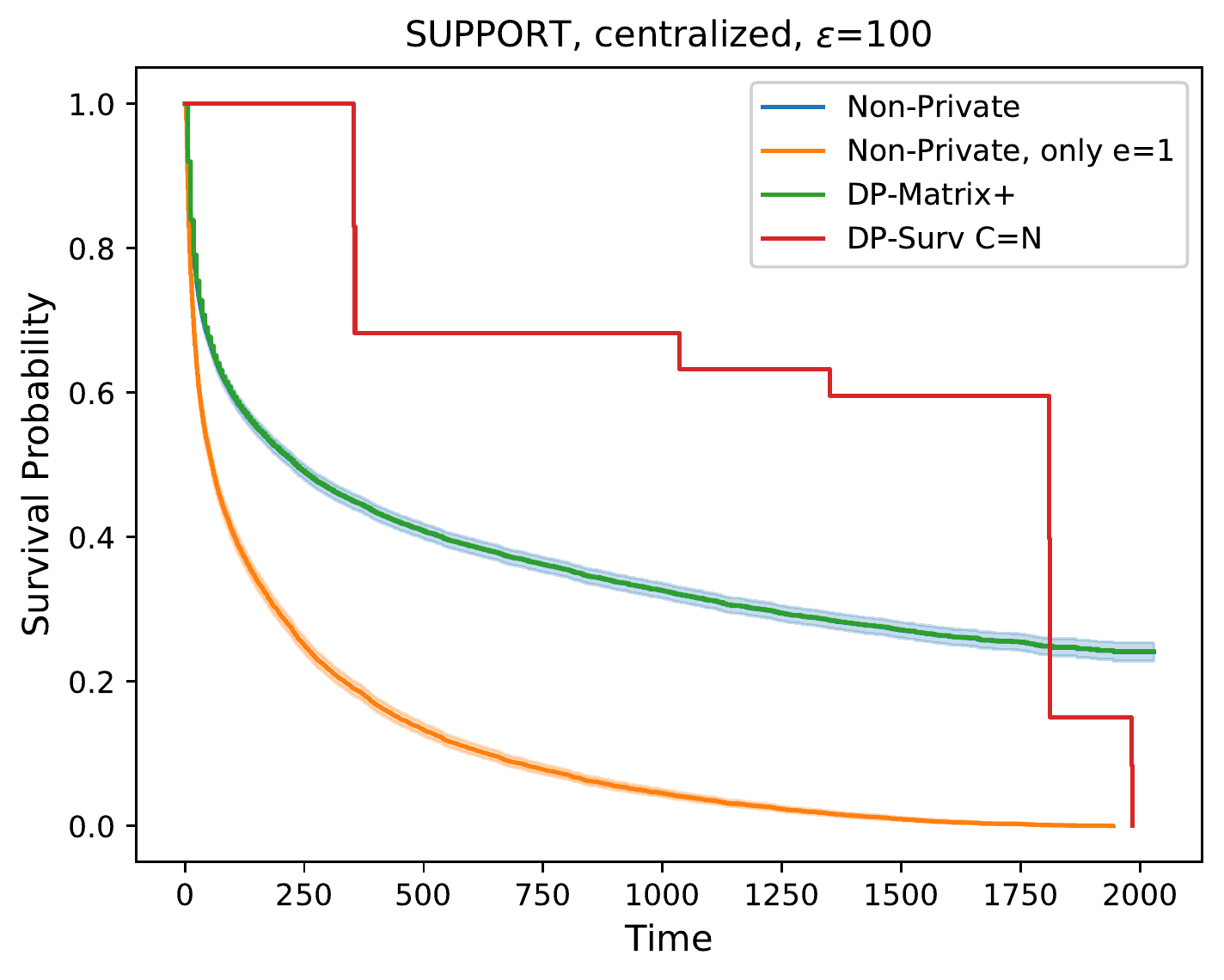}
\end{minipage}
\begin{minipage}[l]{0.68\columnwidth}
        \centering
        \includegraphics[width=\linewidth]{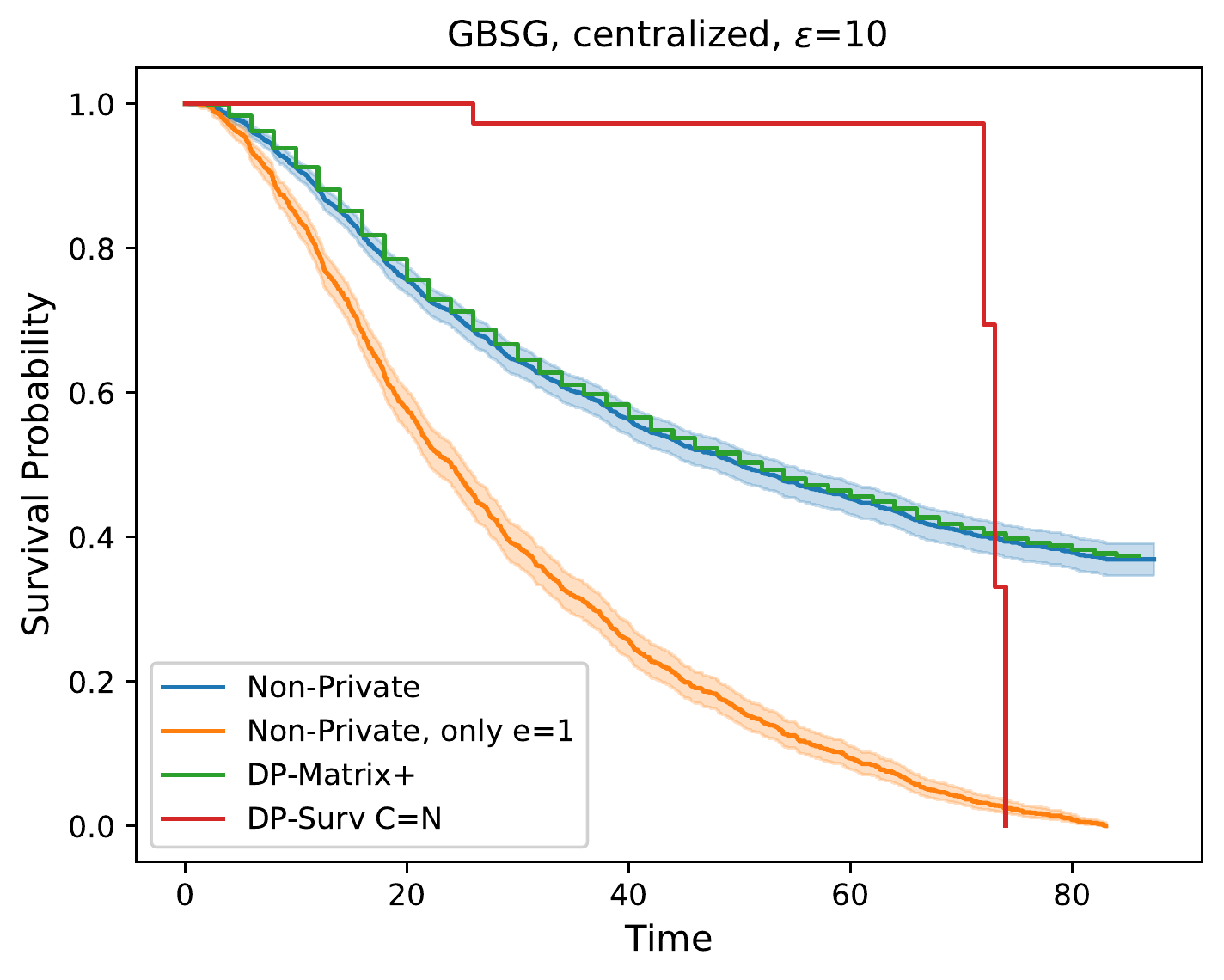}
\end{minipage}
\begin{minipage}[l]{0.68\columnwidth}
        \centering
        \includegraphics[width=\linewidth]{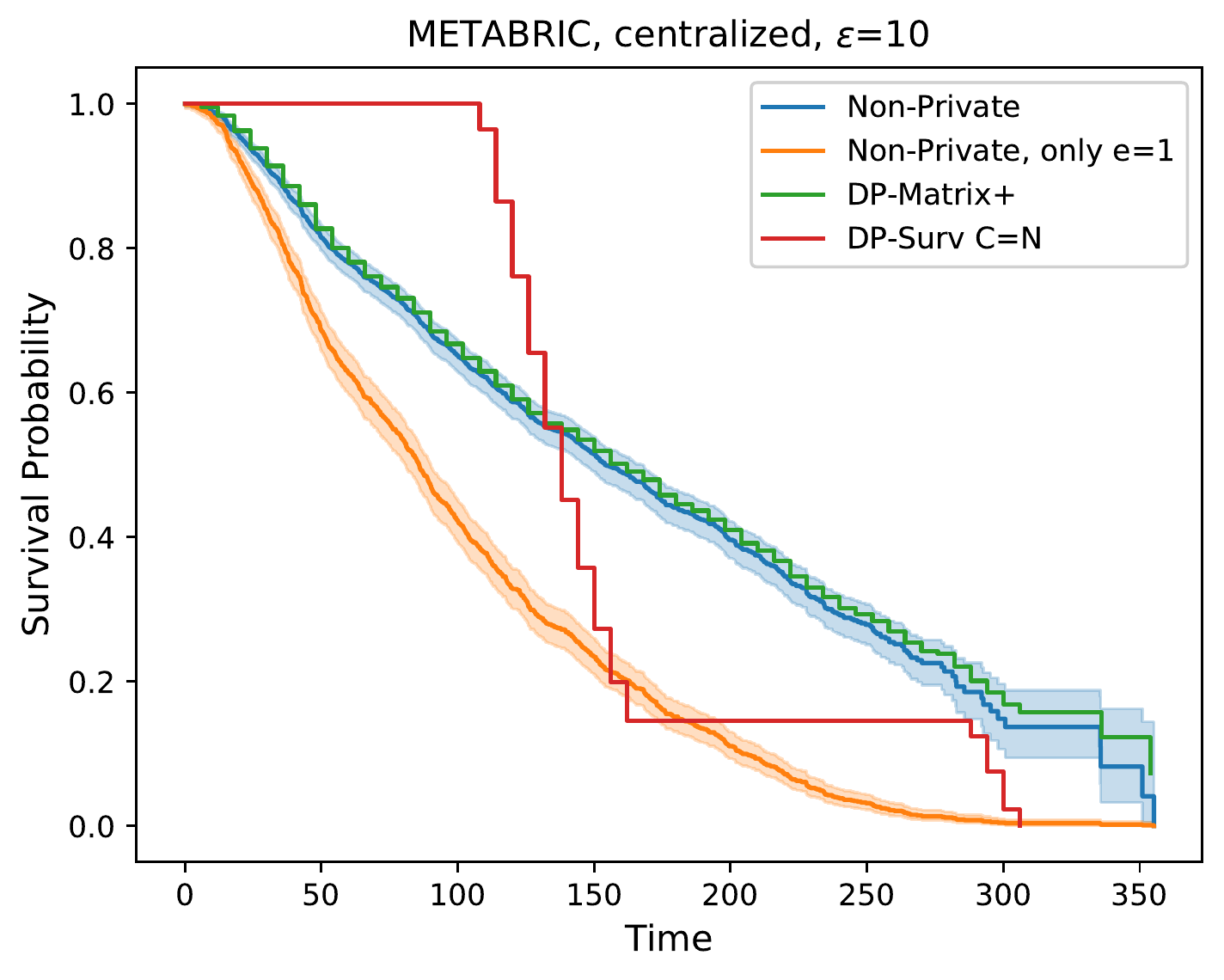}
\end{minipage}
\begin{minipage}[l]{0.68\columnwidth}
        \centering
        \includegraphics[width=\linewidth]{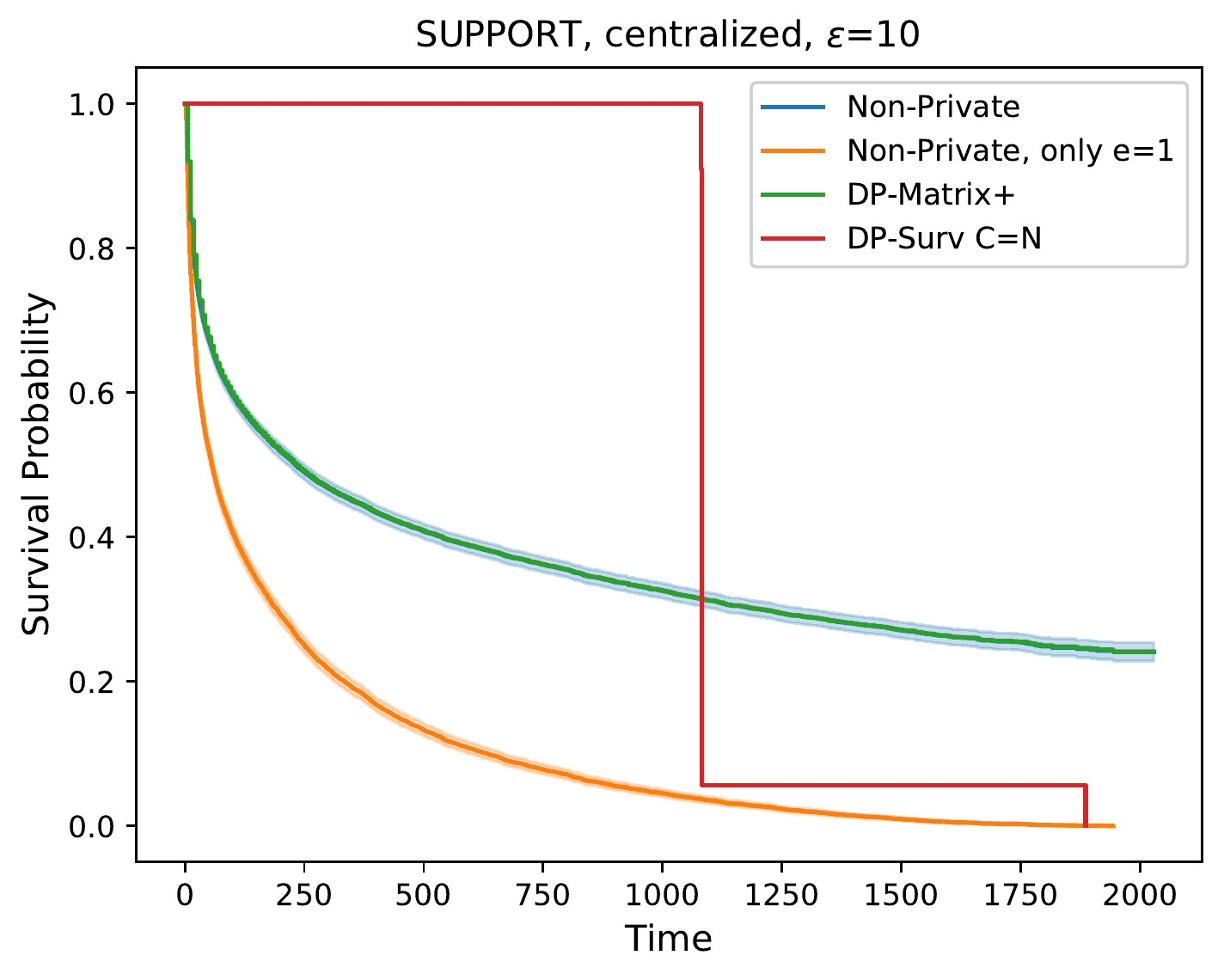}
\end{minipage}
\caption{Comparison of the complete dataset vs dataset only having the non-censored points with $e=1$ and also the centralized DP curves obtained by \dpm and \dps.}
\label{fig:centralized-smooth}
\end{figure*}

\end{document}